\PassOptionsToPackage{hypertexnames=false}{hyperref}
\documentclass[acmsmall,screen,natbib]{acmart}

\usepackage[T1]{fontenc}

\usepackage{xparse}
\usepackage{xkeyval}

\usepackage{algorithm2e}
\usepackage{adjustbox}
\usepackage{colortbl}
\usepackage{doi}
\usepackage[inline]{enumitem}
\usepackage{graphicx}
\usepackage{listings}
\usepackage{siunitx}
\usepackage{subcaption}
\usepackage{soul}
\usepackage[many]{tcolorbox}
\usepackage[dvipsnames,table]{xcolor}

\usepackage{diagbox}
\usepackage{scalerel}

\usepackage{amsmath}
\usepackage{amssymb}
\usepackage{amsthm}
\usepackage{mathtools}
\usepackage{thmtools}
\usepackage[mathscr]{eucal}
\usepackage{textcomp} 

\usepackage{wasysym}
\usepackage{stmaryrd}
\usepackage{quiver}
\usepackage[bb=boondox]{mathalpha}

\usepackage{tikz}
\usepackage{tikz-qtree}
\usetikzlibrary{shapes}
\usetikzlibrary{arrows.meta}
\usetikzlibrary{calc}
\usetikzlibrary{positioning}
\usetikzlibrary{decorations.pathmorphing}
\usetikzlibrary{decorations.markings}
\usetikzlibrary{decorations.shapes}
\usetikzlibrary{backgrounds}

\usepackage{multibib}
\newcites{app}{Supplemental References}



\DeclarePairedDelimiter{\set}{\lbrace}{\rbrace}

\SetKwIF{If}{ElseIf}{Else}{$\codeIf$}{$\codeThen$}{$\elseif$}{$\codeElse$}{}
\SetKwBlock{Block}{}{}
\LinesNumbered
\DontPrintSemicolon
\SetNlSty{}{\color{gray}}{}
\SetKwComment{Comment}{\# }{} 
\SetInd{0.25em}{0.25em}

\usepackage{newfloat}
\DeclareFloatingEnvironment{listing}

\makeatletter
\renewcommand{\nllabel}[1]
 {{\let\@currentlabel\algocf@currentlabel
  \let\@currentcounter\algocf@currentcounter
  \label{#1}}}%

\renewcommand{\algocf@nl@sethref}[1]{%
  \renewcommand{\theHAlgoLine}{\thealgocfproc.#1}%
  \hyper@refstepcounter{AlgoLine}%
  \gdef\@currentlabel{#1}%
  \gdef\@currentcounter{AlgoLine}%
  \gdef\algocf@currentlabel{#1}%
  \gdef\algocf@currentcounter{AlgoLine}%
 }%
\makeatother

\usepackage[capitalize]{cleveref}
\crefalias{AlgoLine}{line}
\newcommand{\crefrangeconjunction}{--}
\makeatletter
\DeclareRobustCommand{\labelcrefrange}[2]{\@crefrangenostar{labelcref}{#1}{#2}}
\makeatother
\crefname{equation}{}{}
\Crefname{equation}{Eq.}{Eqs.}

\crefname{footnote}{footnote}{footnotes}
\Crefname{footnote}{Footnote}{Footnotes}
\Crefname{definition}{Definition}{Definitions}
\crefname{definition}{Def.}{Defs.}
\Crefname{line}{Line}{Lines}
\crefname{line}{line}{lines}
\crefformat{section}{#2\S#1#3}
\crefmultiformat{section}{#2\S#1#3}{ and~#2\S#1#3}{, #2\S#1#3}{, and~#2\S#1#3}
\crefrangeformat{section}{#3\S#1#4\crefrangeconjunction{}#5\S#2#6}
\Crefformat{section}{Section~#2#1#3}
\Crefmultiformat{section}{Sections~#2#1#3}{ and~#2#1#3}{, #2#1#3}{, and~#2#1#3}
\Crefrangeformat{section}{Sections~#3#1#4\crefrangeconjunction{}#5#2#6}

\Crefname{sublisting}{Listing}{Listings}

\AddToHook{cmd/appendix/before}{%
  \crefalias{section}{appendix}%
  \crefalias{subsection}{appendix}
}



\numberwithin{equation}{section}

\newcommand{\defnstyle}{acmdefinition}
\newcommand{\plainstyle}{acmplain}

\declaretheorem[style=\defnstyle,sibling=definition,numberwithin=section]{definition}
\declaretheorem[style=\defnstyle,sibling=definition]{notation}
\declaretheorem[style=\defnstyle,sibling=definition]{remark}

\declaretheorem[style=\plainstyle,sibling=definition]{lemma}

\makeatletter
\declaretheoremstyle
    [headformat={\NOTE}, 
    notebraces={}{}, 
    notefont=\bfseries, 
    preheadhook=\def\thmt@space{}, 
    numbered=no,
    ]{expandedDef}
\makeatother

\makeatletter
\declaretheoremstyle
    [headformat={\NOTE}, 
    notebraces={}{}, 
    notefont=\bfseries, 
    preheadhook=\def\thmt@space{}, 
    bodyfont     = \normalfont\itshape,  
    numbered=no
    ]{expandedThm}
\makeatother

\newcommand{\highlightBox}[2]{%
  \tcbox[
    on line,
    colback=#1,
    colframe=#1,
    boxrule=0mm,
    boxsep=0.5mm,
    left=0mm,
    right=0mm,
    top=0mm,
    bottom=0mm]{\ensuremath{#2}}%
}

\tcbset{
  plaincode-base/.style={
    enhanced,
    sharp corners,
    colback=white,
    colframe=white,
    left=-0.5mm,
    right=-8mm,
    top=-1mm,
    bottom=-1mm,
    boxsep=1.0mm,
    toptitle=0.0mm,
  }
}

\newtcolorbox{plainCodeBox}[2][]{
  plaincode-base,
  width=#2,
  #1
}

\newtcolorbox{plainCodeBoxFramed}[2][]{
  plaincode-base,
  width=#2,
  overlay={
    \draw[line width=0.8pt, black]
      ([xshift=4mm]frame.north west) -- ([xshift=-1.5mm]frame.north east);
    \draw[line width=0.8pt, black]
      ([xshift=4mm]frame.south west) -- ([xshift=-1.5mm]frame.south east);
  },
  #1
}

\newtcolorbox{syntaxbox}[2][]{
  enhanced,
  sharp corners,
  colbacktitle=gray!20!white,
  coltitle=black,
  colback=white,
  colframe=white,
  boxrule=0.5pt,
  left=0mm,
  right=0mm,
  top=0mm,
  bottom=-0.5mm,
  boxsep=1.0mm,
  toptitle=0.0mm,
  nobeforeafter,
  box align=top,
  halign title=flush center,
  subtitle style={halign=center},
  width=#2,
  #1
}

\newtcolorbox{titlecodebox}[2][]{
  enhanced,
  sharp corners,
  colback=white,
  colframe=black!50,      
  boxrule=0.4pt,
  colbacktitle=gray!20!white,   
  coltitle=black,
  left=-1mm,
  right=-4mm,
  top=0.5mm,
  bottom=0.5mm,
  boxsep=1mm,
  toptitle=0.5mm,
  titlerule=0pt,          
  nobeforeafter,
  box align=top,
  halign title=flush center,
  width=#2,
  #1
}

\tikzset{
  decoration={
    shape backgrounds,
    shape=circle,
    shape size=0.3pt,
    shape sep=2.0pt,
    },
  frameDotted/.style = {decorate, fill=black}
}

\newcommand{\transform}{}
\newcommand{\transformCoupling}{}
\newcommand{\transformModCoupling}{}
\newcommand{\transformMerge}{}
\newcommand{\transformModMerge}{}
\newcommand{\transformErasure}{}
\newcommand{\transformPPEval}{}
\newcommand{\transformAD}{}
\newcommand{\transformInference}{}
\newcommand{\transformGradInfFD}{}
\newcommand{\transformGradInfAD}{}

\RenewDocumentCommand{\transform}{g}{\mathcal{T}\IfValueT{#1}{\set{#1}}}
\RenewDocumentCommand{\transformCoupling}{g}{\mathcal{C}\IfValueT{#1}{\set{#1}}}
\RenewDocumentCommand{\transformModCoupling}{g}{\mathcal{C}^*\IfValueT{#1}{\set{#1}}}
\RenewDocumentCommand{\transformMerge}{g}{\mathcal{M}\IfValueT{#1}{\set{#1}}}
\RenewDocumentCommand{\transformModMerge}{g}{\mathcal{M}^*\IfValueT{#1}{\set{#1}}}
\RenewDocumentCommand{\transformErasure}{g}{\mathcal{E}\IfValueT{#1}{\set{#1}}}
\RenewDocumentCommand{\transformPPEval}{g}{\mathcal{S}\IfValueT{#1}{\set{#1}}}
\RenewDocumentCommand{\transformAD}{g}{\mathcal{D}\IfValueT{#1}{\set{#1}}}
\RenewDocumentCommand{\transformInference}{g}{\mathcal{I}\IfValueT{#1}{\set{#1}}}
\RenewDocumentCommand{\transformGradInfFD}{g}{\mathcal{G}_{\Delta}\IfValueT{#1}{\set{#1}}}
\RenewDocumentCommand{\transformGradInfAD}{g}{\mathcal{G}\IfValueT{#1}{\set{#1}}}

\newcommand{\symbolProofRelevant}{\scriptscriptstyle\square}
\newcommand{\relationCoupling}{R_\mathcal{C}}
\newcommand{\proofCoupling}{\mathrm{Prf}_{\mathcal{C}}}
\newcommand{\relationCouplingProof}{R_\mathcal{C}^{\symbolProofRelevant}}

\newcommand{\relationPPEval}{R_{\mathcal{S}}}
\newcommand{\relationPPEvalProof}{R_{\mathcal{S}}^{\symbolProofRelevant}}
\newcommand{\proofPPEval}{\mathrm{Prf}_{\mathcal{S}}}

\newcommand{\systemname}[0]{\textsf{GradInf}}

\newcommand{\diff}{\mathrm{d}}%
\newcommand{\expect}{\mathbb{E}}%
\newcommand{\variance}{\mathbb{V}}%
\newcommand{\ter}[0]{t}

\newcommand{\type}[0]{\tau}
\newcommand{\base}[0]{\sigma}
\newcommand{\var}[0]{x}

\newcommand{\emptyEnv}{\mathbin{\vcenter{\hbox{$\scriptstyle\bullet$}}}}
\newcommand{\extend}[3]{#1, #2 : #3}

\newcommand{\lang}{\lambda}
\newcommand{\langP}{\lang_P}
\newcommand{\langPP}{\lang_{PP}}

\newcommand{\sem}[1]{\llbracket #1\rrbracket}
\newcommand{\semP}[1]{\llbracket #1\rrbracket}

\newcommand{\bind}[0]{\gg\mkern-10mu\scalebox{1}[1]{=}}
\newcommand{\defas}{\coloneqq}
\newcommand{\asdef}{\eqcolon}
\newcommand{\setvar}{\coloneqq}

\newcommand{\colforsyntax}[1]{\mathbf{\color{syntaxcolor}#1}}

\colorlet{highlightColor}{blue!5}
\colorlet{highlightColorDeep}{blue!10}
\colorlet{highlightAltColor}{gray!10}
\colorlet{highlightAltColorDeep}{gray!18}
\colorlet{highlightAltAltColor}{yellow!12}
\colorlet{highlightAltAltColorDeep}{yellow!20}

\colorlet{highlightColorPurpose}{highlightAltAltColor}
\colorlet{highlightColorSemantics}{highlightAltColor}

\colorlet{highlightColorPrimitive}{highlightAltAltColorDeep}
\colorlet{highlightColorCoupling}{highlightAltAltColorDeep}
\colorlet{highlightColorPartialEval}{highlightAltAltColorDeep}
\colorlet{highlightColorInferenceOverview}{highlightColorDeep}

\colorlet{highlightColorTheorem}{highlightAltAltColorDeep}

\colorlet{highlightColorProbability}{highlightColor}
\newcommand{\highlightProbability}[1]{\highlightBox{highlightColorProbability}{#1}}
\newcommand{\highlightProbabilityName}{a \highlightProbability{\text{blue shade}} }

\colorlet{highlightColorDCC}{highlightAltColor}
\newcommand{\highlightDCC}[1]{\highlightBox{highlightColorDCC}{#1}}
\newcommand{\highlightDCCName}{a \highlightDCC{\text{gray shade}} }
\colorlet{highlightColorFactorization}{highlightAltAltColor}
\colorlet{highlightColorFactorizationDeep}{highlightAltAltColorDeep}
\newcommand{\highlightFactorization}[1]{\highlightBox{highlightColorFactorization}{#1}}
\newcommand{\highlightFactorizationDeep}[1]{\highlightBox{highlightColorFactorizationDeep}{#1}}
\newcommand{\highlightFactorizationName}{a \highlightFactorization{\text{yellow shade}} }

\colorlet{highlightColorBernoulli}{highlightColor}
\colorlet{highlightColorBernoulliDeep}{highlightColorDeep}
\newcommand{\highlightBernoulliName}{a \highlightBox{highlightColorBernoulli}{\text{blue shade}} }

\colorlet{syntaxcolor}{blue!60!black}
\colorlet{fadedcolor}{gray!80!black}

\newcommand{\LstBangBang}{{\small\ttfamily\char`\!\char`\!}}

\newcommand{\codeIf}[0]{\colforsyntax{if}}

\newcommand{\codeLet}[0]{\colforsyntax{let}}
\newcommand{\semLet}[0]{\mathbf{let}}
\newcommand{\codeThen}[0]{\colforsyntax{then}}
\newcommand{\codeElse}[0]{\colforsyntax{else}}
\newcommand{\codeIfThenElse}[3]{\codeIf~#1~\codeThen~#2~\codeElse~#3}
\newcommand{\semIfThenElse}[3]{\mathbf{if}~#1~\mathbf{then}~#2~\mathbf{else}~#3}
\newcommand{\codeReturn}[0]{\colforsyntax{return}}
\newcommand{\semReturn}[0]{\mathbf{return}}
\newcommand{\codeFst}[0]{\colforsyntax{prj}_1}
\newcommand{\codeSnd}[0]{\colforsyntax{prj}_2}
\newcommand{\semFst}[0]{\mathbf{prj}_1}
\newcommand{\semSnd}[0]{\mathbf{prj}_2}
\newcommand{\semProj}[0]{\mathbf{prj}}

\newcommand{\codePGets}[0]{\gets_{\colforsyntax{P}}}
\newcommand{\semPGets}[0]{\gets_{\mathbf{P}}}
\newcommand{\codePReturn}[0]{\codeReturn_{\colforsyntax{P}}}
\newcommand{\semPReturn}[0]{\semReturn_{\mathbf{P}}}
\newcommand{\semPBind}[0]{\bind_{\mathbf{P}}~}

\newcommand{\codePPGets}[0]{\gets_{\colforsyntax{PP}}}
\newcommand{\codePPReturn}[0]{\codeReturn_{\colforsyntax{PP}}}

\newcommand{\codePPResidual}[0]{\colforsyntax{residual}_{\colforsyntax{PP}}}
\newcommand{\codePPPrimal}[0]{\colforsyntax{primal}_{\colforsyntax{PP}}}

\newcommand{\codeRGets}[0]{\gets_{\colforsyntax{R}}}

\newcommand{\codeRHide}[0]{\colforsyntax{hide}}

\newcommand{\judgeHidden}[0]{\mathsf{Hidden}}
\newcommand{\judgeTransparent}[0]{\mathsf{Transparent}}

\newcommand{\codeSplit}[0]{\colforsyntax{split}}
\newcommand{\semSplit}[0]{\mathbf{split}}

\newcommand{\codeIterate}[0]{\colforsyntax{iterate}}
\newcommand{\codePIterate}[0]{\codeIterate_{\colforsyntax{P}}}

\newcommand{\codeUniform}[0]{\colforsyntax{uniform}}
\newcommand{\codeNormal}[0]{\colforsyntax{normal}}
\newcommand{\codeNormalCRN}[0]{\colforsyntax{normal_\texttt{CRN}}}
\newcommand{\codeFlip}[0]{\colforsyntax{flip}}
\newcommand{\codeFlipWeighted}[0]{\colforsyntax{flip_\texttt{SHARED-SAMPLE}}}
\newcommand{\codeNormalWeighted}[0]{\colforsyntax{normal_\texttt{SHARED-SAMPLE}}}
\newcommand{\codeFlipAlwaysFlip}[0]{\colforsyntax{flip_\texttt{FIXED-PROP}}}

\newcommand{\codeFlipIndep}[0]{\colforsyntax{flip_\texttt{INDEP}}}
\newcommand{\codeFlipCRN}[0]{\colforsyntax{flip_\texttt{CRN}}}
\newcommand{\codeFlipLayerAlwaysFlip}[1]{\colforsyntax{flip\text{-}layer_{\texttt{FIXED-PROP}}}_{#1}}
\newcommand{\codeFlipLayerSameOrIndep}[1]{\colforsyntax{flip\text{-}layer_{\texttt{SAME-OR-INDEP}}}_{#1}}
\newcommand{\codeFlipLayerSameOrAnti}[1]{\colforsyntax{flip\text{-}layer_{\texttt{SAME-OR-ANTI}}}_{#1}}
\newcommand{\codeFlipEnum}[0]{\colforsyntax{flip_\texttt{ENUM}}}
\newcommand{\codeCategorical}[0]{\colforsyntax{categorical}}
\newcommand{\codeCategoricalCRN}[0]{\colforsyntax{categorical_\texttt{CRN}}}
\newcommand{\codeCategoricalENUM}[0]{\colforsyntax{categorical_{\texttt{ENUM}}}}
\newcommand{\codeCategoricalMaxIndep}[0]{\colforsyntax{categorical_\texttt{MAX-INDEP}}}

\newcommand{\codePoisson}[0]{\colforsyntax{poisson}}
\newcommand{\codePoissonFixedProp}[0]{\colforsyntax{poisson_\texttt{FIXED-PROP}}}

\newcommand{\codeSeed}[0]{\colforsyntax{seed}}
\newcommand{\semSeed}[0]{\mathbf{seed}}

\newcommand{\semMerge}[0]{\mathbf{merge}}
\newcommand{\semMergeProof}[0]{\mathbf{merge}^{\symbolProofRelevant}}

\newcommand{\codeSum}[0]{\colforsyntax{sum}}
\newcommand{\codeNormalize}[0]{\colforsyntax{normalize}}
\newcommand{\codeMap}[0]{\colforsyntax{map}}
\newcommand{\codeScanl}[0]{\colforsyntax{scanl}}
\newcommand{\codeSequenceR}[0]{\colforsyntax{sequence}_{\colforsyntax{R}}}
\newcommand{\codeSequenceP}[0]{\colforsyntax{sequence}_{\colforsyntax{P}}}

\newcommand{\codeZip}[0]{\colforsyntax{zip}}

\newcommand{\codeIndex}[0]{\colforsyntax{index}}

\newcommand{\semExpect}[0]{\mathbf{expect}}

\ExplSyntaxOff

\newcommand{\typeProduct}{\times}
\newcommand{\typeFunction}{\to}

\newcommand{\typeReal}{\texttt{Real}}
\newcommand{\typeInteger}{\texttt{Int}}
\newcommand{\typeBool}{\texttt{Bool}}
\newcommand{\typeUnit}{\texttt{Unit}}
\newcommand{\typeProb}{\texttt{Prob}}

\newcommand{\typeResidual}{\texttt{Residual}}

\newcommand{\typePProb}{\texttt{PProb}}
\newcommand{\typeSeed}{\texttt{Seed}}

\newcommand{\codeUnit}{\colforsyntax{unit}}
\newcommand{\semUnit}{\mathbf{unit}}
\newcommand{\codeTrue}[0]{\colforsyntax{true}}
\newcommand{\semTrue}[0]{\mathbf{true}}
\newcommand{\codeFalse}[0]{\colforsyntax{false}}
\newcommand{\semFalse}[0]{\mathbf{false}}

\newcommand{\setVariables}{\mathsf{Variables}}
\newcommand{\setConstants}{\mathsf{Constants}}
\newcommand{\setBaseTypes}{\mathsf{BaseTypes}}
\newcommand{\setTypes}{\mathsf{Types}}
\newcommand{\setContexts}{\mathsf{Contexts}}
\newcommand{\setPrimitives}{\mathsf{Primitives}}
\newcommand{\setTerms}{\mathsf{Terms}}
\newcommand{\setBool}{\mathbb{B}}
\newcommand{\setReal}{\mathbb{R}}
\newcommand{\setInteger}{\mathbb{Z}}
\newcommand{\setUnit}{\mathbf{1}}

\newcommand{\catQBS}{\mathbf{QBS}}
\newcommand{\catMeas}{\mathbf{Meas}}

\newcommand{\catSyntax}{\mathrm{Syntax}}

\newcommand{\GradInfA}{\textsf{GradInf-CRN-VE}}
\newcommand{\GradInfB}{\textsf{GradInf-CRN-TSMC}}
\newcommand{\GradInfC}{\textsf{GradInf-MI-SIR}}
\newcommand{\GradInfD}{\textsf{GradInf-MI-TSMC}}

\setcopyright{cc}
\setcctype{by}
\acmDOI{10.1145/3808321}
\acmYear{2026}
\acmJournal{PACMPL}
\acmVolume{10}
\acmNumber{PLDI}
\acmArticle{243}
\acmMonth{6}
\acmSubmissionID{pldi26main-p518-p}
\received{2025-11-13}
\received[accepted]{2026-04-03}

\begin{document}

\title{GradInf: Gradient Estimation as Probabilistic Inference}

\author{Gaurav Arya}
\orcid{0000-0003-1662-3037}
\affiliation{%
  \institution{Carnegie Mellon University}
  \city{Pittsburgh}
  \country{USA}
}
\email{gauravar@cmu.edu}

\author{Mathieu Huot}
\orcid{0000-0002-5294-9088}
\affiliation{%
  \institution{Massachusetts Institute of Technology}
  \city{Cambridge}
  \country{USA}
}
\email{mhuot@mit.edu}

\author{Moritz Schauer}
\orcid{0000-0003-3310-7915}
\affiliation{%
  \institution{Chalmers University of Technology \& University of Gothenburg}
  \city{Gothenburg}
  \country{Sweden}
}
\email{smoritz@chalmers.se}

\author{Alexander K. Lew}
\orcid{0000-0002-9262-4392}
\affiliation{%
  \institution{Yale University}
  \city{New Haven}
  \country{USA}
}
\email{alexander.lew@yale.edu}

\author{Feras A. Saad}
\orcid{0000-0002-0505-795X}
\affiliation{%
  \institution{Carnegie Mellon University}
  \city{Pittsburgh}
  \country{USA}
}
\email{fsaad@cmu.edu}




\begin{abstract}
\textit{Gradient estimation}---the task of computing the gradient of
the expected value of a probabilistic program---has diverse applications
in scientific computing, but is notoriously difficult because
of issues such as high-dimensional integration, discrete random choices,
and complex stochastic dependencies.
This article introduces \textit{gradient inference}, a new approach to
developing sound and efficient gradient estimators for probabilistic
programs.
Gradient inference rests on a formal reduction from a gradient estimation
problem to a closely related probabilistic inference problem, whose
solution can be differentiated to obtain a gradient estimator.
This inference problem is obtained by applying two powerful
statistical operations---\textit{coupling} and \textit{factorization}---to the
input probabilistic program.
Our reduction lets us
leverage the rich toolkit of probabilistic inference algorithms to design
novel gradient estimators that extend and improve upon existing methods.

We introduce \systemname{}, a probabilistic programming system that
facilitates the sound and automated implementation of gradient inference.
\systemname{} is centered around programmable source-to-source
transformations for coupling and factorizing
higher-order probabilistic programs,
whose soundness is proven
in terms of a denotational semantics.
Key to our development is the use of information-flow typing
to allow random choices in a probabilistic program to be factored out
and \emph{partially evaluated}, which improves
our ability to deploy sophisticated probabilistic inference algorithms.
The resulting system offers practitioners a principled framework for
designing gradient estimators.
We apply \systemname{} to several challenging case studies, showing that it
can express prominent gradient estimators from the literature and
enables the construction of new state-of-the-art estimators that
outperform the best existing baselines.
\end{abstract}

\maketitle


\section{Introduction}
\label{sec:introduction}

Suppose we are given a function $\mu: \setReal^n \to P(\setReal)$
that maps a parameter vector $\theta \defas (\theta_1,\dots,\theta_n)$
to a probability distribution $\mu_\theta$ over the reals.
The expectation of $\mu_\theta$ is given by
\begin{equation}
\mathbb{E}_{x \sim \mu_\theta}[x] = \int_{\setReal} x \mu_\theta(\diff{x}) \asdef \ell(\theta).
\label{eq:expected-value}
\end{equation}
\textit{Gradient estimation} is the problem of computing the gradient
$\nabla\ell(\theta)\,{:}\,\setReal^n$ of the expectation of $\mu_\theta$,
which is a vector whose $i$th coordinate is the partial
derivative of $\ell$ with respect to $\theta_i$.
Gradient estimation is a fundamental computational problem,
with applications in diverse domains such as
biochemical modeling~\citep{anderson2012efficient},
queuing theory~\citep{fu1997conditional},
statistical epidemiology~\citep{chopra2023differentiable},
particle detector simulations~\citep{kagan2023branches},
reinforcement learning~\citep{sutton1998reinforcement},
neural networks~\citep{bengio2013estimating},
and large language models~\citep{schulman2017proximal}.
Use cases for gradient estimation include optimizing a
stochastic objective function and assessing the sensitivity of a stochastic
process to perturbations in its parameters.

\subsection{The Challenge of Gradient Estimation}
\label{sec:introduction-challenge}


\begin{figure}
\centering

\captionsetup{belowskip=-10pt}

\colorlet{colorCore}{black}
\colorlet{colorCoupling}{black}
\colorlet{colorPEval}{black}
\colorlet{colorInference}{black}

\input{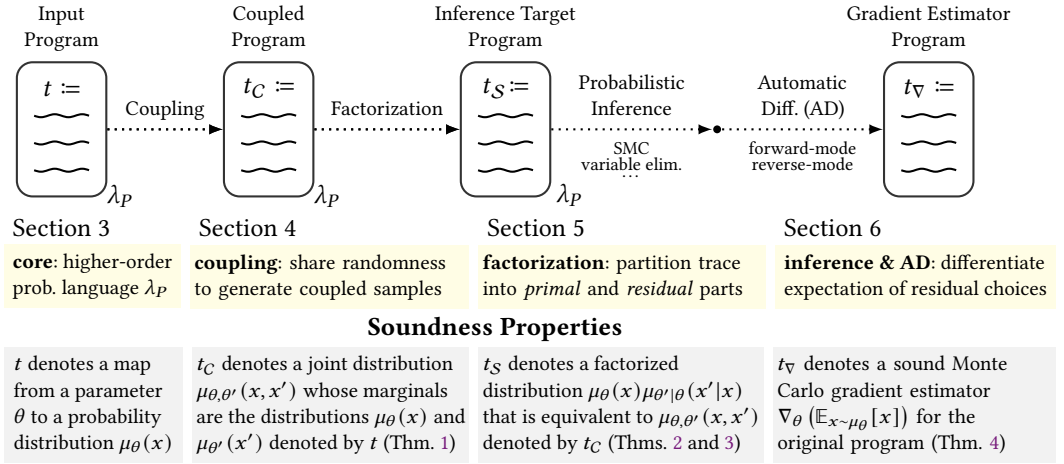}

\begin{adjustbox}{max width=\linewidth}
\begin{tikzpicture}[thick,scale=1,transform shape]

\node[
    name=input,
    prog node,
    label={[font=\footnotesize]above:{\shortstack[c]{Input\\Program}}},
    ] {};

\begin{scope}[shift={(input.south west)},x={(input.south east)},y={(input.north west)}]
    \node at (0.5,0.8) {$t \defas$};
    \foreach \yy in {0.2,0.4,0.6} {
            \draw[
                decorate,
                decoration={complete sines, amplitude=0.4mm, segment length=3mm},
                line width=0.8pt,
                colorCore] (0.2,\yy) -- (0.8,\yy);
        }
    \node at (1.14,0.01) {\color{colorCore}$\langP$};
\end{scope}

\node[
    name=coupled,
    prog node,
    label={[font=\footnotesize]above:{\shortstack[c]{Coupled\\ Program}}},
    right=1.55 of input
    ] {};

\begin{scope}[shift={(coupled.south west)},x={(coupled.south east)},y={(coupled.north west)}]
    \node at (0.5,0.8) {$t_{\transformCoupling} \defas$};
    \foreach \yy in {0.2,0.4,0.6} {
            \draw[
                decorate,
                decoration={complete sines, amplitude=0.4mm, segment length=3mm},
                line width=0.8pt,
                colorCoupling] (0.2,\yy) -- (0.8,\yy);
        }
    \node at (1.14,0.01) {\color{colorCore}$\langP$};
\end{scope}

\node[
    name=annotated,
    prog node,
    label={[font=\footnotesize]above:{\shortstack[c]{Inference Target \\Program}}},
    right=1.95 of coupled
    ] {};

\begin{scope}[shift={(annotated.south west)},x={(annotated.south east)},y={(annotated.north west)}]
    \node at (0.5,0.8) {$t_{\transformPPEval}\hspace{-0.2em}\defas$};
    \foreach \yy in {0.2,0.4,0.6} {
            \draw[
                decorate,
                decoration={complete sines, amplitude=0.4mm, segment length=3mm},
                line width=0.8pt,
                colorPEval] (0.2,\yy) -- (0.8,\yy);

        }
    \node at (1.2,0.02) {\color{colorPEval}$\langP$};
\end{scope}

\node[
    name=estimator,
    prog node,
    label={[font=\footnotesize]above:{\shortstack[c]{Gradient Estimator\vphantom{g} \\ Program}}},
    right=4.45 of annotated
    ] {};

\begin{scope}[shift={(estimator.south west)},x={(estimator.south east)},y={(estimator.north west)}]
    \node at (0.5,0.8) {$t_{\nabla}\defas$};
    \foreach \yy in {0.2,0.4,0.6}   {
            \draw[
                decorate,
                decoration={complete sines, amplitude=0.4mm, segment length=3mm},
                line width=0.8pt,
                colorInference] (0.2,\yy) -- (0.8,\yy);
        }
\end{scope}

\node[
    circle,
    fill=black,
    minimum size=3pt,
    inner sep=0pt,
    name=after-inference,
    at=($(annotated)!0.5!(estimator)$)] {};

\draw[-latex, dotted] (input.east) -- node[midway,below,color=colorCoupling] {} (coupled.west);
\draw[-latex, dotted] (coupled.east) -- node[pos=0.84,yshift=-0.25mm,above,color=colorPEval] {} (annotated.west);

\draw[-latex, dotted] (input.east)
    --
    node[midway,above,color=colorInference, font=\footnotesize] {
        \stackbox[c]{Coupling}
    }
(coupled);

\draw[-latex, dotted] (coupled.east)
    --
    node[midway,above,color=colorInference, font=\footnotesize] {
        \stackbox[c]{Factorization\vphantom{g}}
    }
(annotated);

\draw[-latex, dotted] (annotated.east)
    --
    node[midway,above,color=colorInference, font=\footnotesize] {
        \stackbox[c]{Probabilistic\\Inference\vphantom{Diff. (AD)}}
    }
    node[midway,below,color=colorInference, font=\scriptsize] {
        \begin{tabular}{@{}c@{}}SMC\\[-2pt]variable elim.\\[-5pt]\dots\end{tabular}
    }
(after-inference);

\draw[-latex, dotted] (after-inference)
    --
    node[midway,above,color=colorInference, font=\footnotesize] {
        \stackbox[c]{Automatic\\Diff. (AD)}
    }
    node[midway,below,color=colorInference, font=\scriptsize] {
        \begin{tabular}{@{}c@{}}forward-mode\\[-2pt]reverse-mode\end{tabular}
    }
(estimator.west);

\node[name=correctness-header, at=($(input.west)!0.5!(estimator.east)$),yshift=-2.67cm]
    {\textbf{Soundness Properties}};

\node[below=0.15 of input,color=colorCore] (core-section) {{\hypersetup{linkcolor=colorCore}\Cref{sec:preliminaries}}};
\node[below=-0.025 of core-section.south west,anchor=north west,color=colorCore,font=\footnotesize,fill=highlightColorPurpose] (core-purpose) {
    \stackbox[l]{
        \textbf{core}: higher-order
        \\
        prob. language $\langP$
    }
};
\path let \p1 = (core-purpose.east), \p2 = (core-purpose.west) in
  node[
    below=0.5 of core-purpose.south west,
    anchor=north west,
    color=colorCore,
    font=\footnotesize,
    fill=highlightColorSemantics,
    inner xsep=1.5pt,
    minimum width={\x1-\x2}
  ] (core-semantics) {
    \stackbox[l]{
      $t$ denotes a map
      \\
      from a parameter
      \\
      $\theta$ to a probability
      \\
      distribution $\mu_\theta(x)$
    }
  };

\node[right=0.1 of core-semantics.north east,anchor=north west,font=\footnotesize,fill=highlightColorSemantics, inner xsep=1.5pt] (coupling-semantics) {
    \stackbox[l]{
        \crefname{theorem}{Thm.}{Thms.}
        $t_{\transformCoupling}$ denotes a joint distribution
        \\
        $\mu_{\theta,\theta'}(x,x')$ whose marginals
        \\
        are the distributions $\mu_\theta(x)$ and
        \\
         $\mu_{\theta'}(x')$ denoted by $t$
        (\cref{thm:coupling-correctness})
    }
};
\path let \p1 = (coupling-semantics.east), \p2 = (coupling-semantics.west) in
    node[right=0.1 of core-purpose.north east,anchor=north west,font=\footnotesize,fill=highlightColorPurpose,text width={\x1-\x2-3pt},align=left,inner xsep=1.5pt] (coupling-purpose) {
    \stackbox[l]{
        \textbf{coupling}: share randomness \\
        to generate coupled samples
    }
    };
\node[above=0.22 of coupling-purpose.north west,color=colorCoupling,anchor=west] (coupling-section) {{\hypersetup{linkcolor=colorCoupling}\Cref{sec:coupling}}};

\node[right=0.1 of coupling-semantics.north east,anchor=north west,font=\footnotesize,fill=highlightColorSemantics, inner xsep=1.5pt] (partial-eval-semantics) {
\stackbox[l]{
\crefname{theorem}{Thm.}{Thms.}
$t_{\transformPPEval}$ denotes a factorized
\\
distribution $\mu_{\theta}(x)\mu_{\theta'|\theta}(x'|x)$
\\
that is equivalent to $\mu_{\theta,\theta'}(x,x')$
\\
denoted by $t_{\transformCoupling}$
(\cref{thm:factorized-coupling-correctness,thm:langPP-realization})
}
};
\path let \p1 = (partial-eval-semantics.east), \p2 = (partial-eval-semantics.west) in
    node[right=0.1 of coupling-purpose.north east,anchor=north west,color=colorPEval,font=\footnotesize,fill=highlightColorPurpose,text width={\x1-\x2-3pt},align=left,inner xsep=1.5pt] (partial-eval-purpose){
        \stackbox[l]{
            \textbf{factorization}: partition trace
            \\
            into \emph{primal} and \emph{residual} parts
        }
    };
\node[above=0.22 of partial-eval-purpose.north west,color=colorPEval,anchor=west] (partial-eval-section) {{\hypersetup{linkcolor=colorPEval}\Cref{sec:partial-evaluation}}};

\node[right=0.1 of partial-eval-purpose.north east,anchor=north west,color=colorInference,font=\footnotesize,fill=highlightColorPurpose] (inference-purpose) {
    \stackbox[l]{
        \textbf{inference \& AD}: differentiate
        \\
        expectation of residual choices
    }
};
\path let \p1 = (inference-purpose.east), \p2 = (inference-purpose.west) in
    node[right=0.1 of partial-eval-semantics.north east,anchor=north west,font=\footnotesize,fill=highlightColorSemantics, inner xsep=1.5pt,text width={\x1-\x2-3pt},align=left] (inference-semantics) {
        \stackbox[l]{
            \crefname{theorem}{Thm.}{Thms.}
            $t_{\nabla}$
            denotes a sound Monte
            \\
            Carlo gradient estimator
            \\
            $\nabla_\theta \left(\mathbb{E}_{x\sim \mu_\theta}[x]\right)$
            for the
            \\
            original program
            (\cref{thm:grad-inference-correctness})
        }
    };
\node[above=0.22 of inference-purpose.north west,color=colorInference,anchor=west] (inference-section) {{\hypersetup{linkcolor=colorInference}\Cref{sec:gradient-inference}}};

\end{tikzpicture}
\end{adjustbox}

\caption{
Overview of gradient inference workflow using \systemname{}.
The purpose of each component and its corresponding soundness criteria are
summarized informally along the top and bottom rows, respectively.
}
\label{fig:gradient-inference-system}
\end{figure}

Computerized methods for differentiation such as
symbolic differentiation~\citep{cohen2002}
or automatic differentiation (AD)~\citep{baydin2017,uwe2011}
are given the source code of a target function $\ell: \setReal^n \to \setReal$
and return a new program for its gradient $\nabla \ell: \setReal^n \to \setReal^n$.
The key challenge presented by gradient estimation is that our target
function \cref{eq:expected-value} is typically an intractable
integral, rendering the direct application of standard symbolic
manipulation rules infeasible.
Moreover, the probability distribution $\mu_\theta$ being integrated may
be specified by a probabilistic program that
includes discrete randomness, stochastic control flow, and other
challenging constructs that cause standard AD to fail~\citep{lew2023adev}.

The dominant approach to gradient estimation is grounded in
Monte Carlo methods~\citep{mohamed2020monte,fu1997conditional},
which construct a
random vector $L(\theta)$ that estimates
the gradient $\nabla\ell(\theta)$.
An ideal Monte Carlo gradient estimator is
\textit{computationally efficient} (i.e., fast to compute),
\textit{sound} (i.e., unbiased), and
\textit{statistically efficient} (i.e., low variance, or equivalently low mean squared error):
\begin{align}
\mbox{\underline{Unbiased}:}\enskip
  \expect\left[L(\theta)\right] = \nabla\ell(\theta),
&&
\mbox{\underline{Low Variance}:}\enskip
  \expect\left[ || L(\theta) - \nabla\ell(\theta) ||^2  \right]~\mbox{small}.
\label{eq:estimator-bias-variance}
\end{align}
Over the decades, researchers from a range of fields
have developed gradient estimators that aim to be sound and efficient.
Examples include score function estimation~\citep{glynn1990},
also known as REINFORCE~\citep{williams1992};
infinitesimal perturbation analysis (IPA~\citep{glasserman1991gradient}),
also known as the reparameterization trick~\citep{kingma2013};
smoothed perturbation analysis (SPA~\citep{gong1987});
and
measure-valued derivatives (MVD~\citep{heidergott2008measurevalued}).
Practitioners, however, know that no single estimator works well
(or even at all) in all settings.
Developing an effective gradient estimator requires a careful understanding
of the characteristics of the probability distribution $\mu_\theta$.
But the diversity of existing estimation strategies
and their varying mathematical formalisms make it difficult
to interpret, apply, combine, or extend these methods.

\subsection{Our Approach}
\label{sec:introduction-our-approach}
We introduce \emph{gradient inference}, which is a new approach to gradient
estimation that aims to
\begin{enumerate}[label={\bfseries(A\arabic*)},ref=A\arabic*,leftmargin=*]
  \item \label{aim:modular} provide a modular and compositional
    workflow for developing gradient estimators;

  \item\label{aim:novel} enable the construction of novel gradient
    estimators that are sound and efficient.
\end{enumerate}
Gradient inference rests on a reduction of a gradient estimation problem to
a probabilistic inference problem.
This reduction is enabled by two programmable techniques for
statistical variance reduction: \emph{coupling} (i.e., common random variables~\citep[\S8.6]{owen2013}) and
\emph{factorization} (i.e., conditioning~\citep[\S8.7]{owen2013}).
\Cref{tab:gradient-estimation-schemes} shows how gradient inference lets us
interpret existing estimators as a combination of a factorized
coupling and a probabilistic inference algorithm, which we can then generalize and improve upon
by using more powerful programmable probabilistic inference
methods~\citep{mansinghka2018pldi,scibior2018denotational,towner2019pldi}.

\noindentparagraph{\normalfont\bfseries \systemname.}

We contribute a probabilistic programming system, \systemname,
that supports the sound and automated construction of gradient inference schemes.
\Cref{fig:gradient-inference-system} shows an overview of \systemname.
The input is a probabilistic program $t$ denoting the target measure
$\mu_\theta$ over $\setReal$ whose expectation we seek to differentiate.
The final output is a new probabilistic program $t_{\nabla}$ that defines a
sound (and typically efficient) Monte Carlo estimator $L(\theta)$ for the
gradient $\nabla_\theta\left(\expect_{x\sim\mu_\theta}\left[x\right]\right)$.
The program $t_{\nabla}$ is obtained by first synthesizing
an intermediate \emph{inference target} program
$t_{\transformPPEval}$ that exhibits a semantic equivalence to $t$,
but incorporates coupling and factorization
to enable low-variance gradient estimation.
Next, a probabilistic inference algorithm---%
such as variable elimination~\citep[\S9]{koller2009},
sequential Monte Carlo (SMC~\citep{chopin2020introduction}),
or importance sampling~\citep[\S9]{owen2013}---%
is used to estimate expectations over the return value of $t_{\transformPPEval}$.
Finally, the overall gradient estimator $t_\nabla$ is obtained by applying \emph{standard} AD
to the probabilistic inference algorithm itself, which is constructed
so as to avoid the usual pitfalls of na{\"i}vely
differentiating a probabilistic program (e.g., \citep[Fig.~1]{lew2023adev}).

\begin{table}
\captionsetup{skip = 4pt}
\caption{
  Gradient inference can be used to reinterpret a variety of existing
  gradient estimators from the literature in terms of a
  factorized coupling and a probabilistic inference algorithm. This framework
  lets us develop novel estimators that harness more powerful inference
  algorithms to achieve lower variance.
  }
\label{tab:gradient-estimation-schemes}
\footnotesize
{\setlength{\tabcolsep}{4pt}
\begin{tabular*}{\linewidth}{|l@{\,\extracolsep{\fill}}lll|} \hline
\bfseries{~}                                                  & \bfseries{Factorized Coupling} & \bfseries{Probabilistic Inference Algorithm} & \bfseries{Reference}\\ \hline
\multicolumn{4}{|@{\,}l|}{\bfseries \textit{Existing Estimators}} \\
\citep{glasserman1991gradient} Pathwise / IPA                 & $\colforsyntax{\langle \emptyEnv{} \rangle_\texttt{CRN}}$ (e.g., $\codeNormalCRN{}$) & Simple Monte Carlo                           & \Cref{appendix:gradient-inference-pathwise} \\
\citep{gong1987} Phantom (SPA)                                & $\colforsyntax{\langle \emptyEnv{} \rangle_\texttt{CRN}}$ (e.g., $\codeFlipCRN{}$)   & Stratified Importance Sampling               & \Cref{appendix:gradient-inference-spa} \\
\citep{fu1997conditional} Randomized Phantom (SPA)            & $\colforsyntax{\langle \emptyEnv{} \rangle_\texttt{CRN}}$ (e.g., $\codeFlipCRN{}$)   & Stratified Importance Resampling             & \Cref{appendix:gradient-inference-spa} \\
\citep{heidergott2008measurevalued} Measure-Valued Derivative & $\colforsyntax{\langle \emptyEnv{} \rangle_\texttt{FIXED-PROP}}$ & Stratified Importance Sampling                & \Cref{appendix:gradient-inference-mvd} \\
\citep{glynn1990} Score (w/ Control Variates)         & $\colforsyntax{\langle \emptyEnv{} \rangle_\texttt{SHARED-SAMPLE}}$   & Varies by Control Variate Method              & \Cref{appendix:gradient-inference-score} \\
\citep{mnih2016variational} RLOO                                      & $\codeFlipLayerSameOrIndep{}$  & Stratified Importance Sampling               & \Cref{appendix:gradient-inference-vae} \\
\citep{dong2020disarm} DisARM                                 & $\codeFlipLayerSameOrAnti{}$   & Stratified Importance Sampling               & \Cref{appendix:gradient-inference-vae} \\
\citep{kunes2023gradient} BitFlip1                            & $\codeFlipLayerAlwaysFlip{}$   & Stratified Importance Resampling             & \Cref{appendix:gradient-inference-vae} \\
\citep{bengio2013estimating} Straight-Through                 & $\codeFlipLayerAlwaysFlip{}$   & Moment Closure                               & \Cref{appendix:gradient-inference-vae} \\ \hline
\multicolumn{4}{|@{\,}l|}{\bfseries \textit{Novel Estimators}} \\
\GradInfA                                             & $\codeFlipCRN$                 & Variable Elimination                         & \Cref{sec:overview,sec:applications-queuing} \\
\GradInfB                                             & $\codeCategoricalCRN$          & Twisted Sequential Monte Carlo               & \Cref{sec:applications-option-pricing} \\
\GradInfC                                             & $\codeCategoricalMaxIndep$     & Stratified Importance Resampling             & \Cref{sec:applications-gene-transcription} \\
\GradInfD                                             & $\codeCategoricalMaxIndep$     & Twisted Sequential Monte Carlo               & \Cref{sec:applications-gene-transcription} \\ \hline\hline
\end{tabular*}
}
\vspace{-.6cm}
\end{table}

\noindentparagraph{\normalfont\bfseries Applications.}
We evaluate \systemname{} on several challenging case studies,
testing its ability to compositionally express gradient estimators
(\ref{aim:modular}) and produce novel schemes (\ref{aim:novel}).
In a queuing theory application (\cref{sec:overview}), we generalize
the classic SPA~\citep{gong1987} estimator by using a variable
elimination inference algorithm that achieves
up to 16x reductions in time-adjusted variance.
In case studies from mathematical finance (\cref{sec:applications-option-pricing}) and systems
biology (\cref{sec:applications-gene-transcription}), we
use gradient inference with sequential Monte Carlo inference to produce novel
estimators that deliver up to 370x improvements compared to existing baselines.
Finally, in both these applications and supplementary studies
in \cref{appendix:gradient-inference},
we show that gradient inference can also express and unify a number of
existing state-of-the-art estimators (cf. \cref{tab:gradient-estimation-schemes}),
surfacing new connections between diverse methodologies.

\subsection{Contributions}

In summary, this paper makes the following contributions.

\begin{itemize}[wide=0pt, leftmargin=*]
\item \textbf{Gradient Inference}: We introduce
      \textit{gradient inference}, an algorithmic technique that converts a
      gradient estimation problem into a probabilistic inference problem
      and interoperates with standard forward- and reverse-mode
      automatic differentiation.

\item \textbf{Language Constructs}: We present \systemname,
      a higher-order probabilistic programming system that facilitates gradient
      inference. \systemname{} is based on source-to-source program
      transformations that automate two statistical variance-reduction techniques:
      \textit{coupling} and \textit{factorization}.

\item \textbf{Formalism}: We characterize
      the key soundness properties for gradient inference and prove
      the correctness of the program transformations in \systemname{}
      via proof-relevant logical relations.
      We leverage information-flow typing to ensure that the factored out choices
      in a probabilistic program can be \emph{partially evaluated}, which
      enables accurate probabilistic inference algorithms.
\item \textbf{Case Studies}:
      We demonstrate the effectiveness of gradient inference on challenging
      problems from queuing theory, finance, and genetics.
      \systemname{} can both express existing estimators and deliver new
      estimators that significantly outperform existing state-of-the-art baselines.
\end{itemize}

\section{Overview of Gradient Inference}
\label{sec:overview}

In this section,
we give an overview of the key ideas in gradient inference
using a probabilistic model from queuing theory
(\cref{sec:overview-model}),
and outline the key programming language challenges
that \systemname{} addresses to facilitate sound and automated gradient
inference workflows (\cref{sec:overview-pl}).

\subsection{How Gradient Inference Delivers Efficient Gradient Estimators}
\label{sec:overview-model}

In the usual approach to gradient estimation, a Monte Carlo
estimator for the target gradient is derived directly,
with the key variance-reduction and algorithmic choices
built into a single construction~\citep{mohamed2020monte}.
In contrast, our approach to gradient estimation, called gradient inference,
modularly decomposes the problem into multiple programmable steps.
We describe each step below, showing how gradient inference both
recovers established estimators and yields novel estimators with
lower variance.

\noindentparagraph{\normalfont \bfseries Example Model.}
\Cref{fig:overview-mathematical-model} shows a probabilistic model called
an M/M/c queue, which describes the flow of packets
in a network queue.
The discrete-time stochastic process $(X_0, X_1, \dots, X_n)$ gives the
number of packets $X_i$ in the queue after the $i$th queuing event,
where packets either join or leave the queue
with probability dependent
on the number of routers $c = 25$ and the packet arrival rate
$\theta:\setReal$.
\Cref{fig:overview-sample-traces} shows 10 sample paths of $X_{0:n}$ for
$n=50$ queuing events and $\theta \in \set{10, 15}$.
Our goal is to assess the sensitivity of the \textit{expected} final queue
length after $n$ queuing events with respect to $\theta$,
which is the gradient ${\nabla_\theta \expect[X_n]}$.
We will write $X_n$ as $X(\theta)$ to emphasize the $\theta$-dependence,
dropping the subscript when it is clear from context.

\begin{figure}

\begin{minipage}[t]{.43\linewidth}
\footnotesize
\begin{minipage}[t]{\linewidth}
\subcaption{User-written input program $t$}
\label{fig:overview-input-program}
\end{minipage}
\begin{plainCodeBoxFramed}[after skip=0pt]{\linewidth}
\begin{algorithm}[H]
    $\codeLet~\mathit{kernel} = \lambda(\theta: \typeReal). \lambda(x: \typeInteger).$
    \Block{
        $\codeLet~p = \theta / (\theta + \min(25, x))$\;
        $b \codePGets \highlightBox{highlightColorPrimitive}{\codeFlipCRN}~p$\; \label{algline:overview-input-program-flip}
        $\codeLet~x' = \codeIfThenElse{b}{x+1}{x-1}$\;
        $\codePReturn~x'$ \;
    }
    \vspace{0.5em}
    $\lambda(\theta: \typeReal).$
    \Block{
        $x \codePGets {\codePIterate}_{n, \typeInteger}~
            (\mathit{kernel}~\theta)~
            0
            $\; \label{algline:overview-input-program-iterate}
        $\codePReturn~x$ \label{algline:overview-input-program-retval}
    }
\end{algorithm}
\end{plainCodeBoxFramed}
{\footnotesize {\hspace{4mm}\tikz{\node[minimum width=.5cm,fill=highlightColorPrimitive,draw=black,label={right:probabilistic primitive}]{};}}}

\bigskip

\begin{minipage}{\linewidth}
\subcaption{Synthesized coupled program $t_{\transformCoupling}$}
\label{fig:overview-coupled-program}
\end{minipage}
\begin{plainCodeBoxFramed}[after skip=0pt]{\linewidth}
\begin{algorithm}[H]
$\codeLet~\mathit{kernel} = \begin{aligned}[t]
    &\lambda(\theta_A: \typeReal, \theta_B : \typeReal).
    \\[-3pt]
    &\lambda(x_A: \typeInteger, x_B: \typeInteger).
    \end{aligned}$
    \Block{
      $\codeLet~p_A = \theta_A/(\theta_A + \min(25, x_A))$\; \label{algline:overview-coupled-program-pA}
      $\codeLet~p_B = \theta_B/(\theta_B + \min(25, x_B))$\; \label{algline:overview-coupled-program-pB}
      $\highlightBox{highlightColorCoupling}{\omega \codePGets \codeUniform~(0,1)}$\; \label{algline:overview-coupled-program-crn1}
      $\highlightBox{highlightColorCoupling}{\codeLet~b_A = \omega < p_A}$ \;         \label{algline:overview-coupled-program-crn2}
      $\highlightBox{highlightColorCoupling}{\codeLet~b_B = \omega < p_B}$            \label{algline:overview-coupled-program-crn3}
      \;
      $\codeLet~x'_A = \codeIf~b_A~\codeThen~x_A+1~\codeElse~x_A-1$\;
      $\codeLet~x'_B = \codeIf~b_B~\codeThen~x_B+1~\codeElse~x_B-1$\;
      $\codePReturn~(x'_A, x'_B)$
    }
\vspace{.5em}
$\lambda(\theta: \typeReal \times \typeReal).$
\Block{
  $(x_A, x_B) \codePGets {\codePIterate}_{n,\typeInteger^2}
      ~\begin{aligned}[t]
      &(\mathit{kernel}~\theta)~(0,0)
      \end{aligned}$
    \\
    $\codePReturn~(x_A, x_B)$ \label{algline:overview-coupled-program-retval}
}
\end{algorithm}
\end{plainCodeBoxFramed}
{\footnotesize {\hspace{4mm}\tikz{\node[minimum width=.5cm,fill=highlightColorPrimitive,draw=black,label={right:result of coupling}]{};}}}

\bigskip

\begin{minipage}{\linewidth}
\subcaption{Synthesized inference target $t_{\transformPPEval}$}
\label{fig:overview-factorized-program}
\end{minipage}
\begin{plainCodeBoxFramed}[after skip=0pt]{\linewidth}
\begin{algorithm}[H]
$\codeLet~\mathit{kernel} = \begin{aligned}[t]
    &\lambda(\theta_A: \typeReal, \theta_B : \typeReal).
    \\[-3pt]
    &\lambda(x_A: \typeInteger, x_B: \typeInteger,
    \highlightBox{highlightColorPartialEval}{s : \typeSeed}).
    \end{aligned}$
  \Block{
    $\codeLet~p_A = \theta_A/(\theta_A + \min(25, x_A))$\; \label{algline:overview-seeded-program-pA}
    $\codeLet~p_B = \theta_B/(\theta_B + \min(25, x_B))$\; \label{algline:overview-seeded-program-pB}
    $\highlightBox{highlightColorPartialEval}{\codeLet~(s_0, s_1) = \codeSplit~s}$\; \label{algline:overview-seeded-program-split}
    $\highlightBox{highlightColorPartialEval}{\codeLet~b_A = s_0 < p_A}$ $\triangleright$\textcolor{black}{\bfseries{\scriptsize ~fixed via seed}}
    \; \label{algline:overview-seeded-program-bA}
    $\highlightBox{highlightColorPartialEval}{
      b_B \codePGets
      \codeFlipEnum
      \left(
        \begin{aligned}
        &\codeIf~b_A~\codeThen \min\Big(\frac{p_B}{p_A}, 1\Big) \\
        &\codeElse~\max\Big(\frac{p_B - p_A}{1 - p_A}, 0\Big)
        \end{aligned}
        \right)
    }$
    \;
    \label{algline:overview-seeded-program-bB}
    $\codeLet~x'_A = \codeIf~b_A~\codeThen~x_A+1~\codeElse~x_A-1$\;
    $\codeLet~x'_B = \codeIf~b_B~\codeThen~x_B+1~\codeElse~x_B-1$\;
    $\codePReturn~(x'_A, x'_B, \highlightBox{highlightColorPartialEval}{s_1})$
  }
\vspace{.5em}
$\lambda(\theta: \typeReal \times \typeReal).
\highlightBox{highlightColorPartialEval}{\lambda(s : \typeSeed).}$
\label{algline:overview-seeded-program-lambda}
\Block{
    $(x_A, x_B) \codePGets
      \begin{aligned}[t]
      &{\codePIterate}_{
          n, \typeInteger \times \typeInteger \highlightBox{highlightColorPartialEval}{\scriptstyle\times \typeSeed}}
          \span\span\span
        \\[-2pt]
       &(\mathit{kernel}~\theta)~(0,0,\highlightBox{highlightColorPartialEval}{s})
       \end{aligned}$\;
    $\codePReturn~(x_A, x_B)$
}
\end{algorithm}
\end{plainCodeBoxFramed}
{\footnotesize {\hspace{4mm}\tikz{\node[minimum width=.5cm,fill=highlightColorPrimitive,draw=black,label={right:result of partial evaluation}]{};}}}
\vspace{-1.2pt}
\end{minipage}\hfill
\begin{minipage}[t]{.53\linewidth}
\footnotesize
\begin{minipage}[t]{\linewidth}
\subcaption{Gradient estimation problem}
\label{fig:overview-mathematical-model}
\vspace{2pt}
\end{minipage}
\begin{plainCodeBox}[left=0mm,right=0mm,colback=gray!5]{\textwidth}
\begin{minipage}{\linewidth}
\centering
$\begin{aligned}
&\mbox{\textbf{Generative Model:}}\\
&X_0 \defas 0 \quad c \defas 25 \\
&B_i \sim \mathrm{Bernoulli}\left({\theta}/{\left(\theta + \min(c,X_{i-1})\right)}\right) \\
&X_i \defas \begin{cases}
    X_{i-1} + 1 & \mbox{if } B_i = 1 \\
    X_{i-1} - 1 & \mbox{otherwise}
\end{cases}\\
&\mbox{\textbf{Goal:} Compute } \nabla_\theta \expect\left[X_n\right], \mbox{where}\\
&\expect\left[X_n\right] = \expect\left[\cdots\expect\left[\expect\left[X_n\mid X_{n-1}\right]\right]\cdots \mid X_0 \right]
\end{aligned}$
\end{minipage}
\end{plainCodeBox}

\smallskip

\begin{subfigure}{\linewidth}
\centering
\captionsetup{aboveskip=2pt,belowskip=2pt}
\subcaption{Sample paths of $(X_0, X_1, \dots, X_{50})$}
\label{fig:overview-sample-traces}
\includegraphics[width=\linewidth]{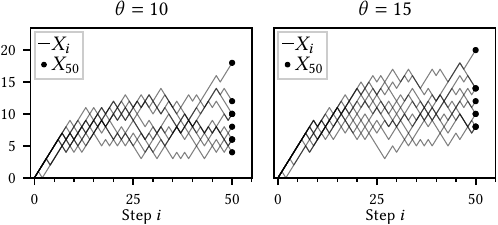}
\end{subfigure}

\begin{subfigure}[t]{\linewidth}
\centering
\captionsetup{aboveskip=-5pt,belowskip=3pt}
\caption{Applying two inference algorithms to $t_{\transformPPEval}$}
\label{fig:overview-inference}
\input{figures/tikz-inference.tex}
\end{subfigure}

\begin{subfigure}{\linewidth}
\captionsetup{aboveskip=-5pt,belowskip=-1pt}
\subcaption{Variance of gradient estimators}
\label{fig:overview-variance-asymptotics}
\centering
\includegraphics[width=\textwidth]{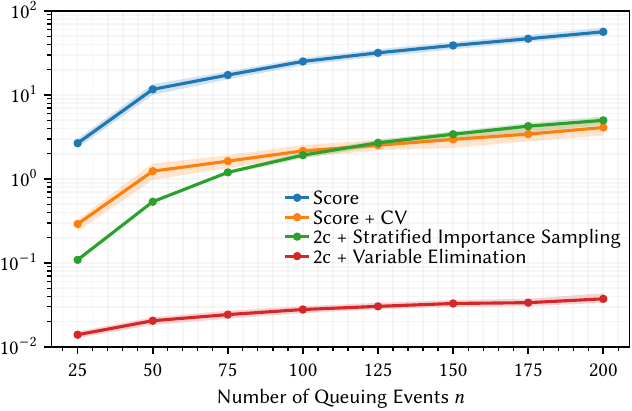}
\end{subfigure}

\end{minipage}

\captionsetup{skip=4pt}
\caption{Applying \systemname{} to a gradient estimation problem in a
discrete-time M/M/c queuing model.}
\label{fig:overview}
\end{figure}

\noindentparagraph{\normalfont \bfseries Probabilistic Program.}

\Cref{fig:overview-input-program} shows a user-written
probabilistic program $t$ in \systemname{}
that implements the model in \cref{fig:overview-mathematical-model}.
The program uses standard probabilistic constructs
such as
Monte Carlo sampling from a Bernoulli distribution
(\cref{algline:overview-input-program-flip})
and iteration of a probabilistic kernel
(\cref{algline:overview-input-program-iterate}).
The return value $x$ on \cref{algline:overview-input-program-retval} is the
final queue size $X_n$.
The semantics $\sem{t}: \setReal \to P(\setReal)$ of this program is precisely
the function $\mu: \setReal \to P(\setReal)$ from \cref{sec:introduction},
which maps the input parameter
$\theta$ to the \textit{marginal probability distribution}
$\sem{t}(\theta)$ of the return value $x$.
In this notation, we seek to compute
${\nabla_\theta\left[\mathbb{E}_{x \sim \sem{t}(\theta)}[x]\right]}
  = {\nabla_\theta \expect[X_n]}$, which
for $\theta : \setReal$ is just a scalar:
\begin{align}
  {\nabla_\theta \expect[X_n]}
  =
  \lim_{\varepsilon \to 0}
  \frac{{\expect\left[X(\theta+\varepsilon)\right] - \expect\left[X(\theta)\right]}}{\varepsilon}
  =
  \lim_{\varepsilon \to 0}
  \frac{{\expect_{x\sim \sem{t}(\theta+\varepsilon)}\left[x\right] - \expect_{x\sim\sem{t}(\theta)}\left[x\right]}}{\varepsilon}.
\label{eq:overview-gradient-limit}
\end{align}
Momentarily ignoring the $\varepsilon$-limit,
\cref{eq:overview-gradient-limit} captures the essence of the gradient estimation
challenge: the numerator requires an accurate estimate of the \textit{difference}
$\expect\left[X(\theta+\varepsilon)\right] - \expect\left[X(\theta)\right]$
of two closely-related expectations.
We will thus turn our attention to this subproblem by temporarily assuming
$\varepsilon$ is a small fixed perturbation, before taking the limit to eliminate
$\varepsilon$.

Given only the program $t$ in \cref{fig:overview-input-program}, we can obtain
independent samples using inputs $\theta$ and $\theta+\varepsilon$
to estimate
$\expect\left[X(\theta+\varepsilon)\right]-\expect\left[X(\theta)\right]$
via simple Monte Carlo.
However, the resulting
Monte Carlo estimator of $\left(\expect\left[X(\theta+\varepsilon)\right] - \expect\left[X(\theta)\right]\right)/\varepsilon$
using $k$ i.i.d.~samples
has unbounded variance
$(\variance(X(\theta)) + \variance(X(\theta+\varepsilon)))/(k\varepsilon^2)$
as $\varepsilon \to 0$, rendering this na\"ive approach ineffective.

\noindentparagraph{\normalfont\bfseries Step 1: Forming a Coupling.}
The variance explosion of simple Monte Carlo motivates forming a \emph{coupling}.
The key idea behind a coupling is that
for any pair of random variables
$(X_A, X_B)$ whose marginal distributions agree
with $(X(\theta), X(\theta+\varepsilon))$,
the linearity of expectation implies that
${\expect\left[X_B - X_A\right]
= \expect\left[X(\theta+\varepsilon)\right] - \expect\left[X(\theta)\right]}$.
If $X_A$ and $X_B$ have positive covariance, then coupling achieves a
(often substantial) variance reduction~\citep[\S8.6]{owen2013} compared to
using independent samples.

\Cref{fig:overview-coupled-program} shows a coupled program
$t_{\transformCoupling}$ synthesized from $t$, in which each original value
has been replaced with a \emph{pair} of values.
This term has semantics
$\sem{t_{\transformCoupling}}: (\setReal \times \setReal) \to P(\setReal \times \setReal)$,
which defines a joint distribution over pairs of outputs.
\Crefrange{algline:overview-coupled-program-crn1}{algline:overview-coupled-program-crn3}
of \cref{fig:overview-coupled-program} show how the coupled variables are generated
using a common random number (CRN) strategy, corresponding to the annotated
$\codeFlipCRN$ primitive on
\cref{algline:overview-input-program-flip} of \cref{fig:overview-input-program}.
The return value $(x_A, x_B)$ of
$\sem{t_{\transformCoupling}}(\theta,\theta+\varepsilon)$
on \cref{algline:overview-coupled-program-retval} of \cref{fig:overview-coupled-program}
is guaranteed to be a valid coupling of
$\sem{t}(\theta)$ and $\sem{t}(\theta+\varepsilon)$,
retaining the unbiasedness of simple Monte Carlo estimation
using $t$:
\begin{equation}
{\expect_{(x_A,x_B)\sim \sem{t_{\transformCoupling}}(\theta,\theta+\varepsilon)}[x_B - x_A]
=
\expect_{x\sim \sem{t}(\theta+\varepsilon)}\left[x\right] - \expect_{x\sim\sem{t}(\theta)}\left[x\right]}.
\label{eq:overview-coupled-expectation}
\end{equation}
\noindentparagraph{\normalfont\bfseries Step 2: Factorizing the Coupling.}
To form the inference target, \systemname{} synthesizes a
\emph{factorized} version of the coupling in \cref{fig:overview-coupled-program}.
As motivation, note that only samples where $x_B \neq x_A$ contribute to the expectation
in \cref{eq:overview-coupled-expectation}.
However, consider
\crefrange{algline:overview-coupled-program-pA}{algline:overview-coupled-program-crn3}
of \cref{fig:overview-coupled-program}: as $\varepsilon \to 0$ we have
$\theta_B \to \theta_A$ and thus
$p_B \to p_A$.
Hence, the samples $b_A$ and $b_B$
on \crefrange{algline:overview-coupled-program-crn2}{algline:overview-coupled-program-crn3}
will be equal with high probability, yielding an estimate of $x_B - x_A = 0$.
We can improve our estimator by using a probabilistic inference
algorithm, such as \textit{importance sampling}~\citep[\S9]{owen2013}, to
generate traces where $x_B \neq x_A$.

Let $X_{A,0:n}$ and $X_{B,0:n}$ be random variables corresponding to traces
of $x_A$ and $x_B$ across all $n$ steps of the coupled program $t_{\transformCoupling}$
in \cref{fig:overview-coupled-program}.
By the chain rule, their joint distribution can be factorized as
\begin{align}
P\left(X_{A,0:n}=x_{A,0:n},X_{B,0:n}=x_{B,0:n}\right)
= P\left(X_{A,0:n}=x_{A,0:n}\right) P\left(X_{B,0:n}=x_{B,0:n}\,\middle|\,X_{A,0:n}=x_{A,0:n}\right).
\label{eq:overview-coupled-disintegration}
\end{align}
We call $X_{A,0:n}$ the \emph{primal} trace and $X_{B,0:n}$ the \emph{residual} trace.
\Cref{fig:overview-factorized-program} shows a synthesized
program $t_{\transformPPEval}$ that implements the factorization
in \cref{eq:overview-coupled-disintegration}.
The program $t_{\transformPPEval}$ is constructed to allow a probabilistic form of \emph{partial
evaluation} of the coupled program $t_{\transformCoupling}$ in
\cref{fig:overview-coupled-program}:
all sampling steps
determining the primal trace can be fixed in $t_{\transformPPEval}$
by evaluating the program at a fixed random seed $s$.
This property enables the subsequent use of
generic probabilistic inference algorithms to infer a residual trace
$X_{B,0:n}(\theta+\varepsilon) \mid X_{A,0:n}(\theta) = x_{A,0:n}$
that is distinct from a fixed primal trace $x_{A,0:n}$ with high probability.

\noindentparagraph{\normalfont\bfseries Step 3: Applying Probabilistic Inference.}
Equipped with the \emph{inference target} program $t_{\transformPPEval}$ (\cref{fig:overview-factorized-program}),
the goal of probabilistic inference is to return an efficient estimator $Z$
of the conditional expectation
\begin{equation}
 \mathbb{E}\left[Z(\theta, \theta + \varepsilon, x_{A,0:n})\right] = \mathbb{E}\left[X_{B,n} - X_{A,n}\,\middle|\,X_{A,0:n} = x_{A,0:n} \right]
 = \mathbb{E}_{(x_A,x_B)\sim \semP{t_{\transformPPEval}}(\theta,\theta+\varepsilon)(s)}\left[x_B - x_A\right].
 \label{eq:overview-probabilistic-inference-estimator}
\end{equation}
This step of gradient inference is related to Rao-Blackwellization,
a variance reduction technique
that replaces an estimator by its conditional expectation given
a subset of random variables~\citep[\S8.7]{owen2013}.
Our criterion \cref{eq:overview-probabilistic-inference-estimator} is
a form of \emph{approximate} Rao-Blackwellization,
since the estimator $Z$ must be unbiased
but not necessarily exhibit zero variance.
\Cref{fig:overview-inference} shows two
inference algorithms applied to $t_{\transformPPEval}$
for $\theta = 15$ and $\varepsilon = 10^{-10}$.
We develop a custom estimator $Z_{\rm SIR}$ that uses stratified importance resampling
(SIR), where the algorithm stratifies (i.e., separately samples from)
the events $A_i$ where the trace of $x_B$ diverges from the fixed trace of $x_A$
at some iteration $i$.
For this queuing model, we can develop an improved estimator $Z_{\rm VE}$
that uses a
\emph{linear-time} variable elimination (VE) inference algorithm.
The efficiency of VE is enabled by coupling and factorization:
because the coupling is monotone, the value of $x_B$ is either
$x_A$ or $x_A + 2$ at each iteration (where $x_A$ is \emph{fixed} in the primal sample),
greatly simplifying the inference target.
Applying VE to the model in \cref{fig:overview-coupled-program} would instead scale
quadratically, highlighting the benefits of factorization.

\noindentparagraph{\normalfont\bfseries Step 4: Forming the Gradient Estimator.}
Having performed the variance-reduction steps,
we now eliminate $\varepsilon$ from \cref{eq:overview-gradient-limit}.
Given an estimator $Z$ satisfying \cref{eq:overview-probabilistic-inference-estimator},
the law of total expectation gives
\begin{align}
\expect\left[X(\theta + \varepsilon)\right] - \expect\left[X(\theta)\right]
=
\expect\left[Z\left(\theta, \theta+\varepsilon, X_{A,0:n}\right)\right].
\label{eq:overview-ad-diffE}
\end{align}
If the random choices made to simulate the estimator $Z$ during probabilistic inference
do not depend directly on the perturbation $\varepsilon$,
we may first apply the chain rule of differentiation
and then (under standard conditions) interchange the derivative and expectation:
\begin{align}
\nabla_\theta \expect\left[X(\theta)\right]
=
\nabla_\varepsilon\left[
\expect\left[Z\left(\theta, \theta+\varepsilon, X_{A,0:n}\right)\right]\right]_{\varepsilon=0}
=
\expect\left[\nabla_\varepsilon\left[Z\left(\theta, \theta+\varepsilon, X_{A,0:n}\right)\right]_{\varepsilon=0}\right].
\label{eq:overview-ad-Ediff}
\end{align}
\Cref{eq:overview-ad-Ediff} gives the final gradient estimator.
The inner derivative on the right-hand side can be obtained using
standard (forward- or reverse-mode) AD,
concluding the gradient inference workflow.

\noindentparagraph{\normalfont\bfseries Empirical Results.}
\Cref{fig:overview-variance-asymptotics} shows a comparison of the variance
of four gradient estimators in the M/M/c model (\cref{fig:overview-mathematical-model})
at $\theta = 15$,
with $n$ varied from 25 to 200.
The blue and orange curves show the standard score
function (REINFORCE~\citep{williams1992}) gradient estimator
without and with a control variate (CV).
The green curve corresponds to our gradient inference
scheme using SIR,
which turns out to recover the randomized phantom estimator
from the smoothed perturbation analysis literature~\citep{fu1997conditional}.
The red curve corresponds to our gradient inference scheme using VE,
which produces a novel gradient estimator that improves over the existing baselines
by up to two orders of magnitude.

\subsection{Soundness and Automation via Types and Transformations}
\label{sec:overview-pl}

\systemname{} facilitates the implementation of gradient inference
workflows for expressive, higher-order probabilistic programs
written in our core language $\langP$, which is formalized in \cref{sec:preliminaries}.
Implementing the \systemname{} workflow in
\cref{fig:gradient-inference-system} requires novel solutions to three key
challenges:
\begin{enumerate}[label={\bfseries(C\arabic*)},ref=(C\arabic*),wide=0pt,itemsep=5pt]
\item \label{challenge:couple}
\textit{How do we form couplings of arbitrary higher-order
  probabilistic programs?}
In \cref{sec:coupling}, adapting ideas from
relational probabilistic program logics~\citep{barthe2017coupling},
we present a programmable
coupling transformation $\transformCoupling{\cdot}$ for $\langP$ programs
and use a proof-relevant logical relations argument to prove that it synthesizes
sound couplings of entire programs (\cref{thm:coupling-correctness}).
\item \label{challenge:factorize}

\textit{How do we partially evaluate random choices in the coupled program?}
In \cref{sec:partial-evaluation}, we present an extended language $\langPP$
for factorizing probabilistic programs into \emph{primal} and \emph{residual} choices,
whose type system builds on the dependency core
calculus of \citet{abadi1999core}.
We present a transformation $\transformModCoupling{\cdot}$
for forming factorized couplings,
which targets $\langPP$
but simplifies to a sound coupling transformation
in $\langP$ when the $\langPP$-specific constructs are erased
using the transformation $\transformErasure{\cdot}$
(\cref{thm:factorized-coupling-correctness});
and a transformation $\transformPPEval{\cdot}$
from $\langPP$ to $\langP$
for probabilistic partial evaluation (in the sense described in \cref{sec:overview-model})
that produces sound factorizations of coupled programs
(\cref{thm:langPP-realization}).

\item \label{challenge:compose}
\emph{How do we ensure that the final synthesized gradient estimator is sound?}
In \cref{sec:gradient-inference}, we complete the gradient inference workflow by
explaining how the program synthesized by $\transformModCoupling{\cdot}$
and $\transformPPEval{\cdot}$
can be composed with existing rigorous formalisms of probabilistic inference~\citep{scibior2018denotational}
and AD transformations~\citep{huot2020correctness,krawiec2022provably}
to obtain an overall sound gradient estimator (\cref{thm:grad-inference-correctness}).
\end{enumerate}


\begin{figure}[t]

\begin{minipage}[t]{.55\linewidth}
\begin{adjustbox}{valign=t}
\begin{tikzpicture}[thick]
\tikzset{lbl/.style={align=center,font=\footnotesize}}

\node[name=lp,draw=black,label={[lbl]left:User\\Program}]{$\langP$};
\node[name=lpC,draw=black,right=1.5 of lp,label={[lbl]right:Coupled\\Program}]{$\langP$};
\node[name=lpD,draw=black,below=1 of lp,label={[lbl]left:Factorized\\Coupled Program}]{$\langPP$};
\node[name=lpE,draw=black,at=(lpD -| lpC),label={[lbl]right:Unannotated\\Coupled Program}]{$\langP$};
\node[name=lpS,draw=black,below=1 of lpD,label={[lbl]left:Partially\\Evaluated Program}]{$\langP$};

\node[name=I,right=1.3 of lpS,circle,fill=black,inner sep=1pt]{};
\node[name=A,right=1.3 of I,circle,fill=black,inner sep=1pt,label={[lbl]right:Gradient\\Estimator}]{};

\draw[thick,-latex] (lp) -- (lpC) node[pos=0.5,above]{$\transformCoupling$};
\draw[thick,dotted] (lp.north) to[bend left] node[pos=0.5,anchor=center,fill=highlightColorTheorem,yshift=.1cm]{\cref{thm:coupling-correctness}} (lpC.north);

\draw[thick,-latex] (lp) -- (lpD) node[pos=0.5,left]{$\transformModCoupling$};
\draw[thick,-latex] (lpD) -- (lpE) node[pos=0.5,above]{$\transformErasure$};
\draw[thick,dotted] (lpE) -- (lpC) node[pos=0.5,anchor=center,fill=highlightColorTheorem,xshift=.1cm]{\cref{thm:factorized-coupling-correctness}};

\draw[thick,-latex] (lpD) -- (lpS) node[pos=0.5,left]{$\transformPPEval$};
\draw[thick,dotted] (lpS) -- (lpE) node[pos=0.6,anchor=center,fill=highlightColorTheorem,xshift=.1cm]{\cref{thm:langPP-realization}};

\draw[thick,-latex] (lpS) -- node[pos=0.5,above,label={[lbl]below:Inference}]{$\transformInference$} (I);
\draw[thick,-latex] (I) -- node[pos=0.5,above,label={[lbl]below:AD}]{$\transformAD$} (A);

\draw[thick,dotted] (lpS.south) to[bend right,looseness=1.5] node[pos=0.5,anchor=center,fill=highlightColorTheorem,yshift=.1cm]{\cref{thm:grad-inference-correctness}} (A.south);
\end{tikzpicture}
\end{adjustbox}
\end{minipage}\hfill
\begin{minipage}[t]{.45\linewidth}
\captionsetup{aboveskip=-5pt,belowskip=0pt}
\caption{Overview of program transformations in \systemname{}.
Dotted lines indicate program equivalences, whose
properties are stated in each of the respective theorems.
The user program is written in $\langP$.
The probabilistic language $\langPP$ serves as an intermediate representation
for synthesized coupled programs that contain annotations for factorization.
Erasing these annotations yields a valid coupled $\langP$ program.
The annotations are used to synthesize an equivalent partially evaluated
$\langP$ program to target with inference and AD, which gives a final
sound gradient estimator.
}
\label{fig:theorems}
\end{minipage}
\vspace{-10pt}
\end{figure}
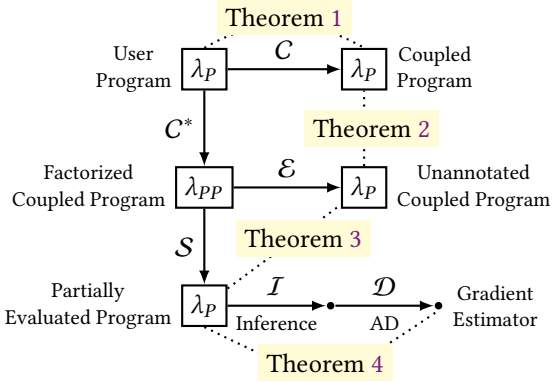

\Cref{fig:theorems} summarizes the program transformations used in
$\systemname{}$ and related theorems that together establish the soundness
of gradient inference.
\Crefrange{sec:preliminaries}{sec:gradient-inference}
unpack the details of \cref{fig:theorems}.

\section{Core Probabilistic Language}

\label{sec:preliminaries}

We begin our technical development by introducing the syntax and semantics of our core probabilistic language $\langP$.
We use \highlightProbabilityName to highlight language constructs related to probability.

\noindentparagraph{\normalfont\bfseries Syntax.}
\Cref{def:core-language} gives the syntax of $\langP$.
The deterministic core of $\langP$
is the standard simply-typed $\lambda$-calculus,
with functions, products, and built-in base
types for booleans, integers, reals, and unit.
For example, $\lambda (x\,{:}\,\tau). t$ and $t'~t$
represent function abstraction and application,
respectively, while $(t_1, t_2)$ creates a pair
and $\codeFst~t$ and $\codeSnd~t$ extract
its components.
To this core, we add the type $\highlightProbability{\typeProb~\type}$ for
measures over $\tau$.
Measures are created
by probabilistic primitives or the monadic return
$\highlightProbability{\codePReturn~t}$;
these measures are propagated using the
bind $\highlightProbability{x \codePGets t;~t'}$.
The set of primitives
$p$ is extensible to support different applications and gradient estimation strategies
(e.g., $\codeFlipCRN$ in \cref{sec:overview}).
The $\langP$ language is equipped with a typing judgement
$\Gamma \vdash t : \tau$.
\Cref{def:core-language} gives typing rules for
base types and the probabilistic constructs;
the full set of rules and details on syntactic sugar
(e.g., for tuple unpacking and iteration)
are given in \cref{appendix:deferred-preliminaries}.

\begin{listing}[t]
\small

\setlength{\abovedisplayshortskip}{0pt}
\setlength{\abovedisplayskip}{0pt}
\setlength{\belowdisplayskip}{0pt}

\begin{syntaxbox}[title=\textbf{Types}]{.475\linewidth}
\begin{align*}
  \base~(\in \setBaseTypes) \Coloneq\, & \typeBool \mid \typeInteger \mid \typeReal
  \mid \typeUnit
  &&
  \\
  \type~(\in \setTypes) \Coloneq\,     &
  \base
  \mid \type_1 \typeProduct \type_2
  \mid \type \typeFunction \type'
  \mid \highlightProbability{\typeProb~\type}
\end{align*}

\bigskip

\tcbsubtitle{\textbf{Contexts}}
\begin{align*}
\Gamma~(\in \setContexts{}) \Coloneq \emptyEnv \mid \Gamma, x:\tau
\end{align*}
\end{syntaxbox}\hfill
\begin{syntaxbox}[title=\textbf{Terms\vphantom{p}}]{.525\linewidth}
  \begin{align*}
    r \in\,                            & \setReal; \;
    n \in \setInteger; \;
    x \in \setVariables
    \\
    c~(\in \setConstants) \Coloneq\,  & r \mid n
    \mid \codeUnit \mid \codeTrue     \mid \codeFalse
    \\
    \ter~(\in \setTerms) \Coloneq\,   &
    \var
    \mid c
    \mid p
    \mid (\ter_1, \ter_2)
    \mid \lambda (\var : \type). \ter
    \mid \ter'~\ter
    \\
    &\mid \codeLet~\var = \ter;~\ter'
    \mid \codeFst
    \mid \codeSnd
    \\
    & \mid \codeIfThenElse{\ter}{\ter_1}{\ter_2}
    \\
    &\mid \highlightProbability{\codePReturn~t} \nonumber
    \mid \highlightProbability{\var \codePGets \ter;~\ter'}
    \\
    p~(\in \setPrimitives) \Coloneq\, &
    \langle \text{extensible} \rangle
  \end{align*}
\end{syntaxbox}

\begin{syntaxbox}[title={\textbf{Typing Rules (Selected)}}]{\linewidth}
 \begin{gather*}
    \Gamma \vdash \codeTrue : \typeBool
    \quad
    \Gamma \vdash \codeFalse : \typeBool
    \quad
    \Gamma \vdash n : \typeInteger
    \quad
    \Gamma \vdash r : \typeReal
    \quad
    \Gamma \vdash \codeUnit : \typeUnit
    \quad
    \Gamma \vdash p : \tau_p
    \\
    \highlightProbability{
       \begin{array}{c}
          \Gamma \vdash t: \type
          \\
          \hline
          \Gamma \vdash \codePReturn~\ter : \typeProb~\type
       \end{array}
    }
    \quad
    \highlightProbability{
       \begin{array}{c}
          \Gamma \vdash t: \typeProb~\type
          \quad \extend{\Gamma}{x}{\type} \vdash t': \typeProb~\type'
          \\
          \hline
          \Gamma \vdash x \codePGets t;~t' : \typeProb~\type'
       \end{array}
    }
 \end{gather*}
\end{syntaxbox}
\captionsetup{aboveskip=4pt}
\captionsetup{belowskip=0pt}
\caption{Types and terms for $\langP$, a core calculus for higher-order probabilistic programs.}
\label{def:core-language}
\vspace{-12pt}
\end{listing}

\noindentparagraph{\normalfont\bfseries Denotational Model.}
A denotational model
provides an interpretation of the syntactic types and well-typed terms
of a language
in a chosen ``semantic space'' category.
As in recent works with
higher-order probabilistic languages~\citep{scibior2018denotational,lew2020trace,becker2024probabilistic,lew2023adev},
the semantic space for $\langP$ is the category
$\catQBS$ of quasi-Borel spaces~\citep{heunen2017convenient}.
The semantics of a type $\tau$ is a $\catQBS$-object $\semP{\tau}$
and the semantics of a well-typed term $\Gamma \vdash t : \tau$
is a $\catQBS$-morphism
$\semP{\Gamma \vdash t : \tau}: \semP{\Gamma} \to \semP{\tau}$,
where a context $\Gamma = x_1 : \tau_1, \ldots, x_n : \tau_n$
is interpreted as a labeled product of $\semP{\tau_1}, \dots, \semP{\tau_n}$.

We make use of two notational shorthands:
(1) when the typing judgement is clear from context, we shorten
$\semP{\Gamma \vdash t: \tau}$ to $\semP{t}$,
and (2) in $\catQBS$, we write $u : U$ as a shorthand for
$u : \semUnit \to U$.
Using this notation, a closed
term
$\emptyEnv \vdash t : \tau$ is interpreted
as a \emph{value} $\semP{t} : \semP{\tau}$,
while a general open term
$\Gamma \vdash t : \tau$ is interpreted
as a map from an \emph{environment} $\gamma : \semP{\Gamma}$
to a value $\semP{t}(\gamma) : \semP{\tau}$.
We denote environment extension by $\gamma[x \setvar u]$.
In the following, we use $\mathbf{black~boldface}$ for semantic space
constructs mirroring language syntax
in $\colforsyntax{blue~boldface}$, e.g.,
$\semFst((x, y)) = x$.

A review of $\catQBS$ is given in \cref{appendix:deferred-preliminaries},
but in the main text
it will generally suffice to analogize to the usual category $\catMeas$
of measurable spaces~\citep{giry1982categorical}, so that
we can intuitively think of $\semP{\tau}$ as a measurable space
and $\semP{\Gamma \vdash t : \tau}$ as a measurable function.
The utility of $\catQBS$ is that the category both supports higher-order functions
and admits a monad $P$ of measures~\citep{scibior2018denotational}.
We can gain intuition for the key operations supported by $P$
by analogizing to the Giry monad in $\catMeas$:
(1) as a functor, $P$ sends a space $U$ to the space $P(U)$ of measures over $U$
and a function $f:U \to U'$ to the pushforward $P(f): P(U) \to P(U')$;
(2) $\semPReturn(u):P(U)$ constructs a Dirac distribution at $u : U$;
and (3) the bind $(\mu \semPBind f) : P(U')$ operation composes a measure $\mu\,{:}\,P(U)$
with a kernel $f:U \to P(U')$ in the usual way~\citep{ramsey2002}:
$(\mu \semPBind f)(A) = \int f(u)(A) \mu(\diff{u})$.

Using $P$, we can define the model $\semP{\cdot}$ on $\langP$.
The model interprets base types using their corresponding $\catQBS$ spaces
$\setBool$, $\setInteger$, $\setReal$, and $\setUnit$,
interprets products and function types using products and exponentials in $\catQBS$, and interprets
$\typeProb~\tau$ using the measure monad $P$:
\begin{gather}
   \semP{\typeBool} \defas \setBool
   \quad
   \semP{\typeInteger} \defas \setInteger
   \quad
   \semP{\typeReal} \defas \setReal
   \quad
   \semP{\typeUnit} \defas \setUnit
   \\
   \semP{\type_1 \typeProduct \type_2} \defas \semP{\type_1} \times \semP{\type_2}
   \quad
   \semP{\type_1 \typeFunction \type_2} \defas \semP{\type_1} \to \semP{\type_2}
   \quad
   \highlightProbability{\semP{\typeProb~\type} \defas P\left(\semP{\type}\right)}.
\end{gather}
The model acts in the standard way on deterministic terms
(e.g., the term $\lambda (x : \tau). t$ is interpreted as a function).
The terms $\codePReturn~t$ and $x \codePGets t;~t'$ are interpreted
using the return and bind of $P$:
$\highlightProbability{\semP{\codePReturn~t}(\gamma) \defas
      \semPReturn\left(\semP{t}(\gamma)\right)}$
and
$\highlightProbability{\semP{\{x \codePGets t;~t'\}}(\gamma) \defas
      \semP{t}(\gamma) \semPBind{}
      \left(u \mapsto \semP{t'}(\gamma[x \setvar u])\right)}$.
Since a primitive $p$ does not depend on its environment, we let
$\semP{p}(\gamma) \defas \semP{p}$ where
$\semP{p}: \semP{\tau_p}$ is given for each primitive.
\cref{appendix:deferred-preliminaries}
supplies the full definition of the model $\semP{\cdot}$
and further detail.

\noindentparagraph{\normalfont\bfseries Program Transformations.}
A program transformation $\transform{\cdot}$ between two languages $\lang_1$ and $\lang_2$
maps the types and well-typed terms of a language $\lang_1$ to those of a language $\lang_2$,
with the property that a well-typed term
$\Gamma \vdash t: \tau$ of $\lang_1$ transforms to a well-typed
term $\transform{\Gamma \vdash t : \tau}$ of $\lang_2$ with
context $\transform{\Gamma}$ and type $\transform{\tau}$.
As with the denotational model, we will abbreviate $\transform{\Gamma \vdash t : \tau}$
to $\transform{t}$ when the typing judgement is clear from context.

Our program transformations will
typically act in an ``interesting'' way on only
a few types and terms, while preserving the structure of the rest.
Accordingly, we say that $\transform{\cdot}$
acts \emph{homomorphically} on a
type or term when it recurses on its structure.
For example,
$\transform{\tau_1 \typeProduct \tau_2} \defas
\transform{\tau_1} \times \transform{\tau_2}$,
$\transform{\typeProb~\tau} \defas \typeProb~\transform{\tau}$,
and
$\transform{\codeLet~x = t;~t'} \defas \set{\codeLet~x = \transform{t};~\transform{t'}}$
specify homomorphic actions on products, measures, and the let bind.
The full definition is given in \cref{appendix:deferred-preliminaries}.

\noindentparagraph{\normalfont\bfseries Example Program.}
The extensible set of primitives in $\langP$
includes standard arithmetic operations,
listed in \cref{appendix:deferred-preliminaries}.
Let us register a probabilistic $\highlightProbability{\codeFlip}$ primitive in $\langP$ with
type $\tau_\codeFlip = \typeReal \to \typeProb~\typeBool$
and semantics $\semP{\codeFlip} \defas r \mapsto \mathrm{Ber}(r)$,
where $\mathrm{Ber}(r) : P(\setBool)$ is a Bernoulli distribution
with success probability $r$.
We can now write an example probabilistic program term in $\langP$,
\begin{equation}
   t \defas \set{b \codePGets \codeFlip~0.5;~\codePReturn~(\codeIfThenElse{b}{(3.5 + 1.2)}{8.2})}.
\end{equation}
Treating $t$ as a closed term, we have $\emptyEnv \vdash t : \typeProb~\typeReal$,
where $\semP{t} : P(\setReal)$ is the probability distribution that
assigns probability $0.5$ to each of the values $4.7$ and $8.2$.

\section{Programmable Couplings of Higher-Order Programs}
\label{sec:coupling}

To address challenge~\labelcref{challenge:couple},
we introduce a coupling transformation
$\transformCoupling{\cdot}$
on the syntax of $\langP$, as well as a logical relation $\relationCoupling$
that governs the soundness of this transformation.

\noindentparagraph{\normalfont\bfseries Coupling Transformation.}
\cref{def:coupling-transformation} defines the coupling
program transformation $\transformCoupling{\cdot}$.
The transformation replaces base types $\sigma$ with pairs thereof and duplicates constant
literals $c$.
Its action on $\codeIfThenElse{t}{t_1}{t_2}$ uses a helper macro
$\transformMerge_\tau\set{\cdot}\set{\cdot}$, where
$\transformMerge_\tau\set{s_1}\set{s_2}$ is a coupled term in which
each pair draws its first component from $s_1$ and its
second component from $s_2$.
On all other types and terms, the transformation acts homomorphically.
The programs produced by $\transformCoupling{\cdot}$ can be viewed as
probabilistic product programs, in the sense of \citet{barthe2017coupling}.
%
\begin{listing}[t]
\small

\setlength{\abovedisplayshortskip}{0pt}
\setlength{\belowdisplayshortskip}{0pt}
\setlength{\abovedisplayskip}{0pt}
\setlength{\belowdisplayskip}{0pt}

\begin{syntaxbox}[title=\textbf{Coupling Transformation $\transformCoupling{\cdot}$ on $\langP$}]{\linewidth}
 \begin{gather*}
      \transformCoupling{\base} \defas \base \typeProduct \base \qquad
     \transformCoupling{c} \defas (c, c) \qquad
     \transformCoupling{p} \defas \langle\text{defined per-primitive}\rangle \\
      \transformCoupling{\Gamma \vdash \codeIfThenElse{t}{t_1}{t_2} : \tau} \defas
     \begin{aligned}[t]
          & \codeLet~(b_A, b_B) = \transformCoupling{t};             \\
          & \transformMerge_\tau\set{\codeIfThenElse{b_A}{\transformCoupling{t_1}}{\transformCoupling{t_2}}}\set{\codeIfThenElse{b_B}{\transformCoupling{t_1}}{\transformCoupling{t_2}}},
     \end{aligned}
 \end{gather*}

 \tcbsubtitle{\bfseries Merge Macro $\transformMerge$ for Conditional Expressions}

 \begin{alignat*}{1}
 \transformMerge_\sigma\set{t_1}\set{t_2}  &\defas \codeLet~((x_A : \sigma, x_B : \sigma), (y_A : \sigma, y_B : \sigma))=(t_1,t_2);~(x_A, y_B)     \\
 \transformMerge_{\tau_1 \typeProduct \tau_2}\set{t_1}\set{t_2}  &\defas
(\transformMerge_{\tau_1}\set{\codeFst~t_1}\set{\codeFst~t_2}, \transformMerge_{\tau_2}\set{\codeSnd~t_1}\set{\codeSnd~t_2})                                             \\
 \transformMerge_{\tau_1 \typeFunction \tau_2}\set{t_1}\set{t_2}  &\defas
\lambda(x : \transformCoupling{\tau_1}).
\transformMerge_{\tau_2}\set{t_1~x}\set{t_2~x}                                                                                   \\
 \transformMerge_{\typeProb~\tau}\set{t_1}\set{t_2} &\defas
\set{x_A \codePGets t_1;~x_B \codePGets t_2;~\codePReturn~\transformMerge_\tau\set{x_A}\set{x_B}},
 \end{alignat*}
\end{syntaxbox}
\captionsetup{aboveskip=3pt,belowskip=-11pt}
\caption{The coupling transformation for $\langP$. The transformation acts homomorphically on all omitted types and terms.}
\label{def:coupling-transformation}
\end{listing}

\noindentparagraph{\normalfont\bfseries Coupling Logical Relation.}
To reason about the soundness of $\transformCoupling{\cdot}$ on terms of arbitrary type,
we introduce the coupling logical relation $\relationCoupling$.
The relation relates tuples $(x_A, x_B, x_{\transformCoupling})$,
where $x_A : \semP{\tau}$, $x_B : \semP{\tau}$, and $x_{\transformCoupling} : \semP{\transformCoupling{\tau}}$.
If $(x_A, x_B, x_{\transformCoupling}) \in \relationCoupling(\tau)$, then $x_{\transformCoupling}$
is a sound coupling of $x_A$ and $x_B$.

To specify $\relationCoupling$, we define a \emph{proof-relevant}
version $\relationCouplingProof$ of the relation.
The relation $\relationCouplingProof$
associates each type $\tau$ of $\langP$ with a quasi-Borel space $\relationCouplingProof(\tau)$,
consisting of tuples $(x_A, x_B, x_{\transformCoupling}, \tilde{x})$
where $x_A : \semP{\tau}$, $x_B : \semP{\tau}$,
$x_{\transformCoupling} : \semP{\transformCoupling{\tau}}$,
and $\tilde{x}$ is a $\catQBS$-value that \emph{proves} the relatedness
$(x_A, x_B, x_{\transformCoupling})$.
As in a usual logical relations argument%
~\citep{statman1985logical,ahmed2006stepindexed,barthe2020versatility,huot2020correctness},
we define $\relationCouplingProof$ inductively on types.
At a base type $\sigma$,
\begin{align}
    \relationCouplingProof(\sigma) \defas
    \set{(x_A, x_B, x_{\transformCoupling}, \semUnit) \mid x_{\transformCoupling} = (x_A, x_B)}.
    \label{eq:relation-coupled-base}
\end{align}
In words, a coupling of two values of base type is a pair thereof.
The definition of $\relationCouplingProof$ on product and function types is standard,
and supplied in \cref{appendix:deferred-coupling}.
To lift the relation to measures~\citep{barthe2015relational,larsen1991bisimulation},
we let $\relationCouplingProof(\typeProb~\tau)$ be the set of tuples
$(\mu_A, \mu_B, \mu_{\transformCoupling}, \tilde{\mu})$
for which the proof value $\tilde{\mu}$ is a probability distribution
whose first three marginals recover
$\mu_A, \mu_B$ and $\mu_{\transformCoupling}$:
\begin{equation}
    \smallint \tilde{\mu}(\diff x) = 1 \land
    P(\semProj_1)(\tilde{\mu}) = \mu_A \land
    P(\semProj_2)(\tilde{\mu}) = \mu_B \land
    P(\semProj_3)(\tilde{\mu}) = \mu_{\transformCoupling}.
\end{equation}
Finally, we define
$
\relationCoupling(\tau) \defas
\set{(x_A, x_B, x_{\transformCoupling}) \mid \exists \tilde{x}. (x_A, x_B, x_{\transformCoupling}, \tilde{x}) \in \relationCouplingProof(\tau)}.
$
The proof-relevance of $\relationCouplingProof$ is crucial to proving our soundness
theorem for couplings (\cref{thm:coupling-correctness}),
as it ensures that proof values can be constructed and propagated using
$\catQBS$-morphisms (cf. \citep{sato2019approximate}).
\noindentparagraph{\normalfont\bfseries Example Coupling Primitives.}
For $\Gamma \vdash t : \tau$,
we say that $\transformCoupling{t}$ is sound if it \emph{preserves} the coupling logical relation, i.e.,
$(\gamma_A, \gamma_B, \gamma_{\transformCoupling}) \in \relationCoupling(\Gamma)$
implies that
$\left(\semP{t}(\gamma_A), \semP{t}(\gamma_B), \semP{\transformCoupling{t}}(\gamma_{\transformCoupling})\right)
\in \relationCoupling(\tau)$.
Since a primitive $p$ is environment-independent,
$\transformCoupling{p}$ is sound
if $\left(\semP{p}, \semP{p}, \semP{\transformCoupling{p}}\right) \in \relationCoupling(\tau_p)$.

For the deterministic primitives of $\langP$, there is usually just one
sound coupling, e.g.,
$\transformCoupling{+} \defas \lambda((x_A, x_B), (y_A, y_B)). (x_A + y_A, x_B + y_B)$.
On the other hand, the probabilistic primitives can be soundly coupled
using many different strategies.
\Cref{lst:coupling-primitives} shows four examples of coupling primitives
in $\langP$, implementing an independent coupling for a Bernoulli,
common random number couplings for a Bernoulli and normal distribution,
and a maximal coupling with independent residuals~\citep{wang2021maximal} for a categorical distribution.
Here, the primitives for normal and categorical distributions have types
$\tau_\codeNormal = \typeReal^2 \to \typeProb~\typeReal$ and
$\tau_{\codeCategorical_{\alpha,n}} = \alpha^n \to \typeReal^n \to \typeProb~\alpha$
and semantics
$\semP{\codeNormal} \defas (\mu, \sigma) \mapsto \mathrm{N}(\mu, \sigma)$;
$\semP{\codeCategorical_{\alpha,n}} \defas (x_{1:n}, p_{1:n}) \mapsto \sum_i p_i \cdot \semPReturn(x_i)$.

As a concrete example, consider the $\codeFlipCRN$ primitive,
whose type is $\type_\codeFlipCRN = \typeReal \typeFunction \typeProb~\typeBool$.
Applying the coupling transform to this type gives
\begin{equation}
    \transformCoupling{\type_\codeFlipCRN}
    = \typeReal \typeProduct \typeReal \typeFunction \typeProb~(\typeBool \typeProduct \typeBool).
\end{equation}
The semantics of $\codeFlipCRN$ are defined to match
$\codeFlip$, i.e.,
$\semP{\codeFlipCRN} = \semP{\codeFlip}$.
However, the $\colforsyntax{\texttt{CRN}}$ annotation
specifies that the primitive should be coupled
using \emph{common random numbers}.
Accordingly,
the implementation
of $\transformCoupling{\codeFlipCRN}$
uses a shared uniform $\omega$ to generate both flips.
By unpacking the definition of $\relationCoupling(\typeReal \to \typeProb~\typeBool)$,
it can be seen
that the soundness requirement imposed on $\transformCoupling{\codeFlipCRN}$
is equivalent to the statement that for all $p_A : \setReal$
and $p_B : \setReal$,
\begin{align}
    P(\semFst)(\semP{\transformCoupling{\codeFlipCRN}}(p_A, p_B))
        & =
    \semP{\codeFlipCRN}(p_A),
    \\
    P(\semSnd)(\semP{\transformCoupling{\codeFlipCRN}}(p_A, p_B))
        & = \semP{\codeFlipCRN}(p_B),
\end{align}
which is precisely the condition we intend:
for all $p_A : \setReal$ and $p_B : \setReal$,
$\semP{\transformCoupling{\codeFlipCRN}}(p_A, p_B)$ is a joint
probability distribution with marginals
of $\semP{\codeFlipCRN}(p_A)$ and $\semP{\codeFlipCRN}(p_B)$.

\begin{listing}
\scriptsize
\captionsetup[sublisting]{skip=0pt}
\begin{sublisting}[t]{0.19\textwidth}
    \begin{plainCodeBox}[box align=top,left=-4mm,colframe=white]{\textwidth}
        \LinesNotNumbered
        \begin{algorithm}[H]
            \DontPrintSemicolon
            $\transformCoupling{\codeFlipIndep} \defas$ \;
            $\;\lambda (p_A : \typeReal, p_B : \typeReal).$
            \Block{
                $b_A \codePGets \codeFlip~p_A$\;
                $b_B \codePGets \codeFlip~p_B$\;
                $\codePReturn~(b_A, b_B)$ \;
            }
        \end{algorithm}
    \end{plainCodeBox}
\end{sublisting}\hfill
\begin{sublisting}[t]{0.19\textwidth}
    \begin{plainCodeBox}[box align=top,left=-4mm,colframe=white]{\textwidth}
        \LinesNotNumbered
        \begin{algorithm}[H]
            \DontPrintSemicolon
            $\transformCoupling{\codeFlipCRN} \defas$ \;
            $\;\lambda (p_A : \typeReal, p_B : \typeReal).$
            \Block{
                $\omega \codePGets \codeUniform~(0,1)$\;
                $\codeLet~b_A = \omega < p_A$\;
                $\codeLet~b_B = \omega < p_B$\;
                $\codePReturn~(b_A, b_B)$ \;
            }
        \end{algorithm}
    \end{plainCodeBox}
\end{sublisting}\hfill
\begin{sublisting}[t]{0.19\textwidth}
    \begin{plainCodeBox}[box align=top,left=-4mm,colframe=white]{\textwidth}
        \LinesNotNumbered
        \begin{algorithm}[H]
            \DontPrintSemicolon
            $\transformCoupling{\codeNormalCRN} \defas$ \;
            $\;\lambda \begin{aligned}[t]
                \big(&(\mu_A, \mu_B), (\sigma_A, \sigma_B))\\[-4pt]
                     &:(\typeReal^2\times\typeReal^2)
                     \big).
                \end{aligned}$
            \Block{
                $\omega \codePGets \codeNormal~(0, 1)$\;
                $\codeLet~x_A = \mu_A + \sigma_A \times \omega$\;
                $\codeLet~x_B = \mu_B + \sigma_B \times \omega$\;
                $\codePReturn~(x_A, x_B)$ \;
            }
        \end{algorithm}
    \end{plainCodeBox}
\end{sublisting}\hfill
\begin{sublisting}[t]{0.42\textwidth}
  \begin{plainCodeBox}[box align=top,left=-4mm,colframe=white]{\textwidth}
      \LinesNotNumbered
      \begin{algorithm}[H]
          \DontPrintSemicolon
          $\transformCoupling{\codeCategoricalMaxIndep_{\alpha,n}} \defas$\;
          $\;\lambda (z : (\alpha^2)^n, p : (\typeReal^2)^n).$
          \Block{
              $\codeLet~(p_A, p_B) = (\codeMap~\codeFst~p, \codeMap~\codeSnd~p)$\;
              $i_A \codePGets \codeCategorical~((1, \dots, n), p_A)$\;
              $\codeLet~x_A = \codeIndex~(\codeMap~\codeFst~z)~i_A$\;
              $\codeLet~x_A' = \codeIndex~(\codeMap~\codeSnd~z)~i_A$\;
              $b \codePGets \codeFlip~(\min(1, (\codeIndex~p_B~i_A) / (\codeIndex~p_A~i_A)))$\;
              $\codeIf~b~\codeThen~(\codePReturn~(x_A,x_A'))~\codeElse$
              \Block{
                  $\codeLet~r = \codeMap~(\lambda(p_A, p_B). \max(0, p_B - p_A))~p$\;
                  $x_B \codePGets \codeCategorical~(\codeMap~\codeSnd~z, \codeNormalize~r)$\;
                  $\codePReturn~(x_A, x_B)$\;
              }
          }
      \end{algorithm}
  \end{plainCodeBox}
\end{sublisting}

\captionsetup{aboveskip=1pt,belowskip=-7pt}
\caption{Example definitions of the coupling transformation for four probabilistic primitives.}
\label{lst:coupling-primitives}
\end{listing}

\noindentparagraph{\normalfont\bfseries Soundness of Program Couplings.}
The main theorem of this section states that if $\transformCoupling{\cdot}$
is soundly defined on each primitive, then the transformation
forms a sound coupling of any input program.
(All proofs are given in the appendices.)
\begin{restatable}[Soundness of Coupling Transformation]{theorem}{thmCouplingCorrectness}
    Let $\Gamma \vdash t : \tau$ be a well-typed term of $\langP$.
    If each primitive $p$ appearing in $t$ satisfies
    $(\semP{p}, \semP{p}, \semP{\transformCoupling{p}})
        \in \relationCoupling(\tau_p)$, then
    \begin{equation}
        \left(\gamma_A, \gamma_B, \gamma_{\transformCoupling}\right)
        \in \relationCoupling(\Gamma)
        \implies
        \left(
        \semP{t}(\gamma_A),
        \semP{t}(\gamma_B),
        \semP{\transformCoupling{t}}(\gamma_{\transformCoupling})
        \right)
        \in \relationCoupling(\tau).
    \end{equation}
    \label{thm:coupling-correctness}
\end{restatable}

\vspace{-1.5em}
\section{Partially Evaluating a Coupled Probabilistic Program}
\label{sec:partial-evaluation}

Having defined a sound coupling transformation,
we now address challenge~\labelcref{challenge:factorize}:
how can we soundly factorize a coupled probabilistic program,
so that it can be partially evaluated in the sense explained
in \cref{sec:overview}?
In this section, we present a factorized coupling transformation
$\transformModCoupling{\cdot}$ that partitions a coupled probabilistic program
into \emph{primal} and \emph{residual} random choices, and a
\emph{partial probability evaluation} transformation $\transformPPEval{\cdot}$
that allows the primal
choices to be fixed ahead of time, ensuring that only residual choices are targeted
by probabilistic inference.

\subsection{Forming Factorized Couplings}
\label{sec:partial-evaluation-factorized-coupling}

To ensure sound factorizations, we desire that residual choices in a program do not affect primal choices:
a \emph{noninterference} property.
To enforce this property, we introduce a new intermediate probabilistic language $\langPP$
based on the classic dependency core calculus (DCC) of \citet{abadi1999core},
that will be the target of the transformation $\transformModCoupling{\cdot}$.

\begin{listing}[t]
\small

\setlength{\abovedisplayshortskip}{0pt}
\setlength{\abovedisplayskip}{0pt}
\setlength{\belowdisplayskip}{0pt}

\begin{syntaxbox}[title=\textbf{Types and Terms}]{\linewidth}
\begin{alignat*}{2}
  \type~(\in \setTypes)& \Coloneq&&
  \begin{aligned}[t]
  ~\cdots
  &\mid \highlightDCC{\typeResidual~\type}
  \mid \highlightFactorization{\typePProb~\type}
  \end{aligned}
  \\
  \ter~(\in \setTerms)& \Coloneq &&
  \begin{aligned}[t]
  ~\cdots
  &\mid \highlightDCC{\codeRHide~t}
  \mid \highlightDCC{x \codeRGets t;~t'}
  \mid \highlightFactorization{\codePPPrimal~t}
  \mid \highlightFactorization{\codePPResidual~t}
  \mid \highlightFactorization{\codePPReturn~t}
  \mid \highlightFactorization{x \codePPGets t;~t'}
  \end{aligned}
\end{alignat*}
\end{syntaxbox}

\smallskip

\begin{syntaxbox}[title={\textbf{Typing Rules}}]{\linewidth}
  \begin{gather*}
     \highlightDCC{\begin{array}{c}
             \Gamma \vdash t: \type
             \\
             \hline
             \Gamma \vdash \codeRHide~t: \typeResidual~\type
         \end{array}}
     \quad
     \highlightDCC{\begin{array}{c}
             \Gamma \vdash t :\typeResidual~\type
             \quad
             \extend{\Gamma}{x}{\type} \vdash t' : \type'
             \quad
             \judgeHidden~\type'
             \\
             \hline
             \Gamma \vdash x \codeRGets t;~t': \type'
         \end{array}}
     \\
     \highlightFactorization{
         \begin{array}{c}
             \Gamma \vdash t: \typeProb~\type
             \quad
             \judgeTransparent~\type
             \\
             \hline
             \Gamma \vdash \codePPPrimal~t: \typePProb~\type
         \end{array}
     }
     \quad
     \highlightFactorization{
         \begin{array}{c}
             \Gamma \vdash t: \typeResidual~(\typeProb~\type)
             \quad
             \judgeHidden~\type
             \\
             \hline
             \Gamma \vdash \codePPResidual~t: \typePProb~\type
         \end{array}
     }
     \\
     \highlightFactorization{
         \begin{array}{c}
             \Gamma \vdash t: \type
             \\
             \hline
             \Gamma \vdash \codePPReturn~t: \typePProb~\type
         \end{array}
     }
     \quad
     \highlightFactorization{
         \begin{array}{c}
             \Gamma \vdash t: \typePProb~\type
             \quad
             \extend{\Gamma}{x}{\type} \vdash t': \typePProb~\type'
             \\
             \hline
             \Gamma \vdash x \codePPGets t;~t':
             \typePProb~\type'
         \end{array}
     }
  \end{gather*}
\end{syntaxbox}

\smallskip

\begin{syntaxbox}[title={\textbf{Hidden and Transparent Typing Judgements}}]{\linewidth}
\[
\begin{array}{@{}c@{\quad}c@{\quad}c@{\quad}c@{}}
  \begin{array}{c} \\ \hline \judgeTransparent~\sigma \end{array}
  &
  \begin{array}{c} \judgeTransparent~\tau_1 \quad \judgeTransparent~\tau_2 \\ \hline \judgeTransparent~(\tau_1 \times \tau_2) \end{array}
  &
  \begin{array}{c} \judgeTransparent~\tau_2 \\ \hline \judgeTransparent~(\tau_1 \to \tau_2) \end{array}
  \\[1.25em]
  \begin{array}{c} \judgeTransparent~\tau \\ \hline \judgeTransparent~(\typeProb~\tau) \end{array}
  &
  \begin{array}{c} \\ \hline \judgeHidden~(\typeResidual~\type) \end{array}
  &
  \begin{array}{c} \judgeHidden~\tau_1 \quad \judgeHidden~\tau_2 \\ \hline \judgeHidden~(\tau_1 \times \tau_2) \end{array}
  \\[1.25em]
  \begin{array}{c} \judgeHidden~\tau_2 \\ \hline \judgeHidden~(\tau_1 \to \tau_2) \end{array}
  &
  \begin{array}{c} \judgeHidden~\tau \\ \hline \judgeHidden~(\typeProb~\tau) \end{array}
  &
  \begin{array}{c} \judgeHidden~\tau \\ \hline \judgeHidden~(\typePProb~\tau) \end{array}
\end{array}
\]
\end{syntaxbox}
\captionsetup{aboveskip=4pt,belowskip=-6pt}
\caption{Types and terms for $\langPP$, a core calculus for factorized probabilistic programs.}
\label{def:phase-separated-language}
\end{listing}

\noindentparagraph{\normalfont\bfseries New Syntax for Factorization.}
The language $\langPP$ is an extension of $\langP$ from \cref{sec:preliminaries}.
\Cref{def:phase-separated-language} gives new syntactic constructs,
using \highlightDCCName
to highlight constructs imported from the DCC
and
\highlightFactorizationName
to highlight new probabilistic constructs.
The first new type is
$\highlightDCC{\typeResidual~\type}$,
whose role is to represent residual values
that are prohibited from influencing primal random choices.
The type is equipped with an introduction form $\highlightDCC{\codeRHide~t}$
and a bind $\highlightDCC{x \codeRGets t;~t'}$.
These constructs
are derived from specializing the DCC to the case of a two-point lattice.
The second new type in $\langPP$ is
$\highlightFactorization{\typePProb~\type}$.
As with $\typeProb~\type$, the type $\typePProb~\type$ is equipped
with a monadic return $\highlightFactorization{\codePPReturn~t}$ and bind
$\highlightFactorization{x \codePPGets t;~t'}$.
The role of $\typePProb~\type$ is to represent probability distributions that
\emph{distinguish} primal random choices from residual ones.
Accordingly, the type is equipped with two additional introduction forms
$\highlightFactorization{\codePPPrimal~t}$ and $\highlightFactorization{\codePPResidual~t}$
for lifting a term of type $\typeProb~\type$ into a term of type $\typePProb~\type$,
marking the term as primal and residual, respectively.
Finally, the $\judgeTransparent$ and $\judgeHidden$ judgements
identify types that carry only primal values and only residual values, respectively.

\noindentparagraph{\normalfont\bfseries Factorized Coupling Transformation.}
The factorized coupling program transformation $\transformModCoupling{\cdot}$,
defined in \cref{def:factorized-coupling},
acts similarly to $\transformCoupling{\cdot}$ from
\cref{def:coupling-transformation}, with two main differences:
\begin{enumerate*}[label=(\roman*)]
\item the second component of each coupled pair is marked as residual; and
\item the return and bind for $\typeProb~\tau$ are replaced with the return and bind for $\typePProb~\tau$.
\end{enumerate*}
The macro $\transformModMerge_\tau\set{\cdot}\set{\cdot}$ (defined in
\cref{appendix:deferred-partial-evaluation}) used in the
$\codeIfThenElse{t}{t_1}{t_2}$ rule is a straightforward extension of the
merge macro in \cref{def:coupling-transformation}.

\begin{listing}[t]
\footnotesize

\setlength{\abovedisplayshortskip}{0pt}
\setlength{\belowdisplayshortskip}{0pt}
\setlength{\abovedisplayskip}{0pt}
\setlength{\belowdisplayskip}{0pt}

\begin{syntaxbox}[title={\bfseries Factorized Coupling Transformation $\transformModCoupling{\cdot}$ from $\langP$ to $\langPP$}]{\linewidth}
 \begin{gather*}
  \transformModCoupling{\base} \defas \base \typeProduct \typeResidual~\base
  \quad
  \transformModCoupling{\typeProb~\tau} \defas \typePProb~(\transformModCoupling{\tau})
  \quad
  \transformModCoupling{c} \defas (c, \codeRHide~c)
  \quad
  \transformModCoupling{p} \defas  \langle \text{defined per-primitive} \rangle
  \\
  \transformModCoupling{\codePReturn~t} \defas \codePPReturn~\transformModCoupling{t}
  \quad
  \transformModCoupling{x \codePGets t;~t'}  \defas
  \set{x \codePPGets \transformModCoupling{t};~\transformModCoupling{t'}}
  \\
  \transformModCoupling{\Gamma \vdash \codeIfThenElse{t}{t_1}{t_2} : \tau} \defas
  \begin{aligned}[t]
       \codeLet~(b_A, \bar{b}_B) = \transformModCoupling{t};
       \transformModMerge_{\tau}\begin{aligned}[t]
        &\set{\codeIfThenElse{b_A}{\transformModCoupling{t_1}}{\transformModCoupling{t_2}}}\\
        &\set{b_B \codeRGets \bar{b}_B;~\codeRHide~(\codeIfThenElse{b_B}{\transformModCoupling{t_1}}{\transformModCoupling{t_2}})}
        \end{aligned}
  \end{aligned}
 \end{gather*}
 \end{syntaxbox}

 \begin{syntaxbox}[title={\bfseries Erasure Transformation $\transformErasure{\cdot}$ from $\langPP$ to $\langP$}]{\linewidth}
 \begin{gather*}
    \transformErasure{\sigma} \defas \sigma
    \qquad
    \transformErasure{\typeResidual~\type} \defas
    \transformErasure{\type}
    \qquad
    \transformErasure{\typePProb~\type} \defas
    \typeProb~\transformErasure{\type}
    \qquad
    \transformErasure{p} \defas p
    \\
     \transformErasure{\codeRHide~t} \defas \transformErasure{t}
    \quad
    \transformErasure{\codePPReturn~t} \defas \transformErasure{\codePReturn~t}
    \quad
    \transformErasure{\codePPPrimal~t} \defas \transformErasure{t}
    \quad
    \transformErasure{\codePPResidual~t} \defas \transformErasure{t}
    \\
      \transformErasure{x \codeRGets t;~t'} \defas
    \set{\codeLet~x = \transformErasure{t};~\transformErasure{t'}}
    \qquad
    \transformErasure{x \codePPGets t;~t'} \defas
    \set{x \codePGets \transformErasure{t};~\transformErasure{t'}},
\end{gather*}
\end{syntaxbox}

\begin{syntaxbox}[title={\bfseries Partial Probability Evaluation Transformation $\transformPPEval{\cdot}$ from $\langPP$ to $\langP$}]{\linewidth}
\begin{gather*}
\transformPPEval{\sigma} \defas \sigma
\qquad
\transformPPEval{\typeResidual~\type} \defas \transformPPEval{\type}
\\
\transformPPEval{\typeProb~\type} \defas \typeProb~\transformPPEval{\type} \typeProduct \left(\typeSeed \to \transformPPEval{\type}\right)
\qquad
\transformPPEval{\typePProb~\type} \defas \typeSeed \to \typeProb~\transformPPEval{\type}
\\
\transformPPEval{\codeRHide~t} \defas \transformPPEval{t}
\qquad
\transformPPEval{x \codeRGets t;~t'} \defas \set{\codeLet~x = \transformPPEval{t};~\transformPPEval{t'}}
\\
\transformPPEval{\codePPPrimal~t} \defas
\lambda s.\{\codePReturn~((\codeSnd~\transformPPEval{t})~s)\}
\qquad
\transformPPEval{\codePPResidual~t} \defas \lambda\_.\set{\codeFst~\transformPPEval{t}}
\\
\transformPPEval{\codePReturn~t} \defas \left( \codePReturn~\transformPPEval{t}, \lambda\_. \transformPPEval{t} \right)
\qquad
\transformPPEval{\codePPReturn~t} \defas \lambda\_.\set{\codePReturn~\transformPPEval{t}}
\\
\transformPPEval{x \codePGets t;~t'} \defas
  \big(
    \set{x \codePGets \codeFst~\transformPPEval{t};~\codeFst~\transformPPEval{t'}},
    \lambda s.\set{\codeLet~(s_0, s_1) = \codeSplit~s;
    ~\codeLet~x = (\codeSnd~\transformPPEval{t})~s_0;
    ~(\codeSnd~\transformPPEval{t'})~s_1}
  \big)
\\
\transformPPEval{x \codePPGets t;~t'} \defas \lambda s. \set*{\codeLet~(s_0, s_1) = \codeSplit~s;~\set*{x \codePGets \transformPPEval{t}~s_0;~\transformPPEval{t'}~s_1}}
\end{gather*}
\end{syntaxbox}

\caption{The factorized coupling, erasure, and partial probability transformations for $\langP$ and $\langPP$.
The program transformations act homomorphically on all omitted types and terms.}
\label{def:factorized-coupling}
\end{listing}

The soundness of $\transformModCoupling{\cdot}$
can be justified by viewing it as an \emph{annotated} version of the
coupling transformation $\transformCoupling{\cdot}$.
To formalize this idea, we also introduce the erasure program
transformation $\transformErasure{\cdot}$
in \cref{def:factorized-coupling}, which erases all type and term
constructs responsible for separating primal and residual values.
We can now state the following analogue of \cref{thm:coupling-correctness}
for $\transformModCoupling{\cdot}$.
\begin{restatable}[Soundness of Factorized Coupling Transformation]{theorem}{thmFactorizedCouplingCorrectness}
    \label{thm:factorized-coupling-correctness}
    Let
    $\Gamma \vdash t : \tau$ be a well-typed term of $\langP$.
    If each primitive $p$ appearing in $t$
    satisfies
    $\left(\semP{p}, \semP{p}, \semP{\transformErasure{\transformModCoupling{p}}}\right)
        \in \relationCoupling(\tau_p)$,
    then
    \begin{equation}
        \left(\gamma_A, \gamma_B, \gamma_{\transformCoupling}\right) \in \relationCoupling(\Gamma)
        \implies
        \left(\semP{t}(\gamma_A), \semP{t}(\gamma_B), \semP{\transformErasure{\transformModCoupling{t}}}(\gamma_{\transformCoupling})\right)
        \in
        \relationCoupling(\tau).
    \end{equation}
\end{restatable}

\noindentparagraph{\normalfont\bfseries Example Factorized Coupling Primitives.}
As with $\transformCoupling{\cdot}$, there is usually one clear definition for deterministic primitives,
e.g.,
$\transformModCoupling{+} \defas \lambda((x_A, \bar{x}_B), (y_A, \bar{y}_B)).
(x_A + y_A, \{x_B \codeRGets \bar{x}_B;~y_B \codeRGets \bar{y}_B;~\codeRHide~(x_B + y_B)\})$.
\Cref{lst:factorized-coupling-primitives} defines factorized couplings
for the $\codeFlipCRN$ and $\codeCategoricalMaxIndep$ primitives from
\cref{sec:coupling}.
Let us unpack the case of $\transformModCoupling{\codeFlipCRN}$,
whose type is given by
\begin{equation}
    \transformModCoupling{\tau_\codeFlipCRN}
     = \typeReal \typeProduct \typeResidual~\typeReal \typeFunction \typePProb~(\typeBool \typeProduct \typeResidual~\typeBool).
\end{equation}
In the implementation of $\transformModCoupling{\codeFlipCRN}$,
the flip operation that samples $b_A$ is wrapped in $\codePPPrimal$.
That the typing
rule for $\codePPPrimal~t$ in \cref{def:phase-separated-language}
expects a \textit{transparent} type ensures that this sample
does not depend on any residual values.
In contrast, the operations that sample $\bar{b}_B$
are wrapped in $\codePPResidual~t$.
That the typing rule for $\codePPResidual~t$ in \cref{def:phase-separated-language}
expects a \textit{hidden} type ensures that this
sample does not determine any primal values.
Recall that this residual sampling step will be targeted by inference: the
$\colforsyntax{\texttt{ENUM}}$ annotation is a hint to the inference algorithm
(explained in \cref{sec:gradient-inference}).

\Cref{lst:factorized-coupling-primitives} also shows the erasure of
$\transformModCoupling{\codeFlipCRN}$.
It is readily verified that
$\transformErasure{\transformModCoupling{\codeFlipCRN}}$
has the same denotation as $\transformCoupling{\codeFlipCRN}$ in \cref{lst:coupling-primitives},
and thus $\transformModCoupling{\codeFlipCRN}$ is a sound factorized coupling.

\subsection{Partial Probability Evaluation}
\label{sec:partial-evaluation-transformation}

To make use of the factorized couplings produced by
$\transformModCoupling{\cdot}$ in gradient inference,
we need a way to compile them into $\langP$ programs
in which only the residual choices remain probabilistic.

\noindentparagraph{\normalfont\bfseries Partial Probability Evaluation Transformation.}
The \textit{partial probability
evaluation} program transformation $\transformPPEval{\cdot}$
from $\langPP$ to $\langP$ is defined in \cref{def:factorized-coupling}.
The key idea is to use a shared seed to determine the output of all
primal choices, but to ignore this seed when sampling residual choices.
The definition uses a $\typeSeed$ type
and helper primitives
$\codeSeed$ and $\codeSplit$ for sampling and splitting seeds,
where
$\tau_\codeSeed = \typeProb~\typeSeed$
and $\tau_\codeSplit = \typeSeed \typeFunction \typeSeed \typeProduct \typeSeed$.
Concretely, we set $\semP{\typeSeed} = \setReal$, $\semP{\codeSeed} = \mathrm{Unif}(0, 1)$,
and $\semP{\codeSplit}$ to any measure-preserving
map from $\setReal$ to $\setReal \times \setReal$.

The transformation $\transformPPEval{\cdot}$
behaves similarly to the erasure transformation on nonprobabilistic terms.
However, the key difference is that
$\transformPPEval{\typePProb~\tau}$ is defined as a
function \emph{from seeds to probability distributions}.
The translation of the introduction form
$\codePPResidual~t$ \emph{ignores} the seed, preserving the original distribution.
The translation of $\codePPPrimal~t$ uses \emph{only} the seed, forming
a Dirac probability distribution given any fixed seed.
To enable this behavior,
the transformation on $\typeProb~\tau$ is modified
to create and propagate both ordinary and seeded
versions of each distribution:
the former is of type $\typeProb~\tau$ and used in $\transformPPEval{\codePPResidual~t}$,
and the latter is of type $\typeSeed \to \tau$ and used in $\transformPPEval{\codePPPrimal~t}$.

\begin{listing}
\footnotesize
\captionsetup[sublisting]{skip=0pt}
\begin{sublisting}{.5\linewidth}
  \begin{plainCodeBox}[left=-3mm,box align=top]{\textwidth}
    \LinesNotNumbered
    \begin{algorithm}[H]
      \DontPrintSemicolon
      $\transformModCoupling{\codeFlipCRN} \defas
        \lambda (p_A: \typeReal, \bar{p}_B: \typeResidual~\typeReal).$
      \Block{
          $b_A \codePPGets \highlightFactorizationDeep{\codePPPrimal}~(\codeFlip~p_A)$\;
          $\bar{b}_B \codePPGets \highlightFactorizationDeep{\codePPResidual}$
          \Block{
              $p_B \codeRGets \bar{p}_B$\;
              $\codeLet~r = \codeIfThenElse{b_A}
                  {\min\left(\frac{p_B}{p_A}, 1\right)}
                  {\max\left(\frac{p_B - p_A}{1 - p_A}, 0\right)}$\;
              $\codeRHide~\left(
                  \codeFlipEnum
                  ~r\right)$
          }
          $\highlightFactorizationDeep{\codePPReturn}~(b_A, \bar{b}_B)$ \;
      }
      \vspace{1em}
      $\transformErasure{\transformModCoupling{\codeFlipCRN}} \defas
        \lambda (p_A: \typeReal, p_B: \typeReal).$
      \Block{
          $b_A \codePGets \codeFlip~p_A$\;
          $\codeLet~r = \codeIfThenElse{b_A}
              {\min\left(\frac{p_B}{p_A}, 1\right)}
              {\max\left(\frac{p_B - p_A}{1 - p_A}, 0\right)}$\;
          $b_B \codePGets \codeFlipEnum~r$\;
          $\codePReturn~(b_A, b_B)$ \;
      }
    \end{algorithm}
  \end{plainCodeBox}
\end{sublisting}\hfill%
\begin{sublisting}{.5\linewidth}
  \begin{plainCodeBox}[left=-3mm, box align=top]{\textwidth}
    \LinesNotNumbered
    \begin{algorithm}[H]
        $\transformModCoupling{\codeCategoricalMaxIndep_{\alpha,n}} \defas$ \;
        $\;\lambda (z : (\alpha \typeProduct \typeResidual~\alpha)^n, p : (\typeReal \typeProduct \typeResidual~\typeReal)^n).$
        \Block{
        $\codeLet~(z_A,p_A) = (\codeMap~\codeFst~z, \codeMap~\codeFst~p)$\;
        $i_A \codePPGets \highlightFactorizationDeep{\codePPPrimal}~\left(\codeCategorical~((1, \dots, n), p_A)\right)$\;
        $\codeLet~x_A = \codeIndex~z_A~i_A$\;
        $\bar{x}_B \codePPGets \highlightFactorizationDeep{\codePPResidual}$
        \Block{
        $p_{B} \codeRGets \codeSequenceR~(\codeMap~\codeSnd~p)$\;
        $z_B \codeRGets \codeSequenceR~(\codeMap~\codeSnd~z)$\;
        $b \codePGets \codeFlip~(\min(1, (\codeIndex~p_B~i_A) / (\codeIndex~p_A~i_A)))$\;
        $\codeIf~b~\codeThen~(\codeRHide~(\codePReturn~(\codeIndex~z_B~i_A)))~\codeElse$
        \Block{
            $\codeLet{}~f = (\lambda(p_A, p_B). \max(0, p_B - p_A))$\;
            $\codeLet~r = \codeMap~f~(\codeZip~p_A~p_B)$\;
            $\codeRHide~(\codeCategoricalENUM~(z_B, \codeNormalize~r))$\;
        }
        }
        $\highlightFactorizationDeep{\codePPReturn}~(x_A, \bar{x}_B)$ \;
        }
    \end{algorithm}
  \end{plainCodeBox}
\end{sublisting}
\captionsetup{aboveskip=2pt}
\captionsetup{belowskip=2pt}
\caption{Example definitions of the factorized coupling transformation for two primitives from \cref{lst:coupling-primitives}.}
\label{lst:factorized-coupling-primitives}
\vspace{-10pt}
\end{listing}

\noindentparagraph{\normalfont\bfseries Soundness of Partial Probability Evaluation.}
The following theorem
guarantees that if $\transformPPEval{\cdot}$ is soundly implemented on each primitive,
then the transformation forms a sound factorization of any input program.
In particular, sampling a seed (which fixes the primal choices)
followed by sampling the residual choices in the transformed program $\transformPPEval{t}$ is
equivalent (in distribution) to sampling all choices jointly using the unannotated original
program $\transformErasure{t}$.
The theorem is established using a logical relation $\relationPPEval$
(defined in \cref{appendix:deferred-partial-evaluation}),
where $(x_{\transformErasure}, x_{\transformPPEval}) \in \relationPPEval(\tau)$
when $x_{\transformPPEval} : \semP{\transformPPEval{\tau}}$
is a sound partial probability evaluation of $x_{\transformErasure} : \semP{\transformErasure{\tau}}$.
\begin{restatable}[Soundness of Partial Probability Evaluation]{theorem}{thmLangPPRealization}
    \label{thm:langPP-realization}
    Let $\Gamma \vdash t : \typePProb~\tau$ be a well-typed term of $\langPP$,
    where the type $\tau$ is constructed from base types, products, functions,
    and the $\typeResidual~\tau$ constructor.
    If each primitive $p$ appearing in $t$ satisfies
    $\left(\semP{p}, \semP{\transformPPEval{p}}\right)
    \in \relationPPEval(\tau_p)$,
    then
    \begin{equation}
        \left(\gamma_{\transformErasure}, \gamma_{\transformPPEval}\right) \in \relationPPEval(\Gamma)
        \implies
        \semP{\transformErasure{t}}(\gamma_{\transformErasure})
        =
        \semSeed \semPBind{} \semP{\transformPPEval{t}}(\gamma_{\transformPPEval}).
    \end{equation}
\end{restatable}

\section{Gradient Inference via Composition with Inference and AD}
\label{sec:gradient-inference}

We complete the gradient inference workflow by addressing
challenge \labelcref{challenge:compose},
formalizing how the transformations $\transformModCoupling{\cdot}$ and
$\transformPPEval{\cdot}$ from \cref{sec:partial-evaluation}
can be composed with inference and automatic differentiation (AD)
to yield a sound gradient estimator.
\noindentparagraph{\normalfont\bfseries Soundness of Inference.}
Define the map $\semExpect : P(\setReal) \to \setReal$ by
$\semExpect(\mu) \defas \int_{\setReal} x \mu(\diff x)$.
An \textit{inference macro}
$\transformInference{\cdot}$
acts on a term
$\Gamma \vdash t : \typeProb~\typeReal$
of $\langP$
and produces a term
$\Gamma \vdash \transformInference{t} : \typeProb~\typeReal$.
We say that $\transformInference{\cdot}$ is \textit{sound} on $t$ if
it preserves its expectation, i.e.,
for all $\gamma : \semP{\Gamma}$,
\begin{align}
\semExpect \left(\semP{\transformInference{t}}(\gamma)\right)
= \semExpect \left(\semP{t}(\gamma)\right).
\label{eq:inference-soundness}
\end{align}
The partial evaluation transformation $\transformPPEval{\cdot}$
ensures that the inference target
program does not contain any intractable conditioning or ``observe'' operations.
Thus, \cref{eq:inference-soundness} requires an
inference algorithm that soundly estimates the expectation of the return
value of a probabilistic program, rather than a \textit{conditional} expectation
involving an intractable normalization constant.
This flexibility lets us reuse many standard inference algorithms to produce
provably unbiased gradient estimators.
We briefly review three such inference algorithms below, employed in our applications.
\begin{itemize}[wide=0pt, leftmargin=*]
    \item \textit{Variable elimination}~\citep[\S9]{koller2009} is an exact inference algorithm
    based on successively marginalizing latent variables.
    When a variable $x$ is eliminated, the joint distribution over
    its dependencies $\mathbf{y}$ is updated by summing out all
    possible values of $x$:
    $p(\mathbf{y}) = \sum_{x} p(\mathbf{y}, x)$.
    The final result is the marginal distribution over the program's return value,
    from which the expectation is computed.
    \item \textit{Stratified importance resampling}~\citep[\S8]{owen2013}
    partitions the probability space into events $A_0, A_1, \dots, A_n$,
    known as \emph{strata}.
    The expectation $\mathbb{E}[X]$ of the program's return value
    can be unbiasedly estimated by
    separately sampling $X_i \sim X \mid A_i$ from each stratum and computing
    $\sum_{i=0}^n X_i \mathbb{P}(A_i)$.
    The computational cost can be reduced by resampling, wherein a random subset
    of the summands are evaluated and reweighted to preserve unbiasedness.
    \item \textit{Sequential Monte Carlo}~\citep{chopin2020introduction} propagates
    a set of weighted particles
    through the program, performing a combination of transition and resampling steps.
    These moves are designed to preserve \emph{proper weightedness}~\citep[\S2]{liu2004monte}, which
    is a property that enables unbiased estimation of integrals.
    Once a set of properly weighted particles $(x_1, w_1), \dots, (x_k, w_k)$
    for the program's return value has been obtained,
    the expectation can be unbiasedly estimated by
    $\sum_{i=1}^k x_i w_i$.
\end{itemize}
To support the gradient inference applications in \cref{sec:applications}, we use
the \textsf{monad-bayes} library
developed by \citet{scibior2018functional}, which is
backed by a rigorous denotational semantics for
functional and higher-order probabilistic programs~\citep{scibior2018denotational}.
We also include inference annotations in
the output of the $\systemname{}$ transformations that can
be used by \textsf{monad-bayes} inference algorithms.
For example, the primitive $\codeFlipEnum$
indicates to enumerate (i.e., stratify)
the two possible outputs of $\codeFlip$.

\noindentparagraph{\normalfont\bfseries Soundness of AD.}
After probabilistic inference is applied, the final step of gradient inference is to apply AD.
An \textit{AD macro}
$\transformAD{\cdot}$
acts on a term
$\Gamma \vdash t : \typeReal^n \typeFunction \typeProb~\typeReal$
of $\langP$ and produces a term
$\Gamma \vdash \transformAD{t} : \typeReal^n \typeFunction \typeProb~\typeReal^n$.
We say that $\transformAD{\cdot}$ is \textit{sound} on $t$ if for all
$\gamma : \semP{\Gamma}$ and $\theta : \setReal^n$,
\begin{align}
    \semExpect\left(\semP{\transformAD{t}}(\gamma)(\theta)\right)
    = \nabla_\theta \left[\semExpect\left(\semP{t}(\gamma)(\theta)\right)\right].
    \label{eq:ad-soundness}
\end{align}
Gradient inference workflows rely only on \emph{standard} AD
algorithms.
Several works formalize and prove AD transformations correct, including for
higher-order programs~\citep{huot2020correctness,krawiec2022provably}.
We use the Haskell libraries \textsf{ad}~\citep{ad2025} for forward-mode
and \textsf{backprop}~\citep{backprop2025} for reverse-mode.
In our setting, these AD transformations can safely ignore
the random choices made by the probabilistic inference algorithm
(i.e., treat their outputs as constants)
so long as they do
not depend on the finite perturbation $\varepsilon$ being differentiated
with respect to (cf.~\cref{eq:overview-ad-Ediff}).
Prior work has developed soundness guarantees for AD
in such settings~\citep{lew2023adev,lee2023smoothness,lee2020correctness,khajwal2023fast},
e.g., the smoothness analysis system of \citet{lee2023smoothness}.
\noindentparagraph{\normalfont\bfseries Gradient Inference Macro.}
The overall gradient inference workflow is enacted by the macro
$\transformGradInfAD{\cdot}$.
Given an inference macro $\transformInference{\cdot}$,
we first define a helper \textit{difference macro}
$\transformGradInfFD{\cdot}$,
acting on a term $\emptyEnv \vdash t : \typeReal^n \typeFunction \typeProb~\typeReal$
and producing a term $\emptyEnv \vdash \transformGradInfFD{t} : (\typeReal \typeProduct \typeReal)^n \typeFunction \typeProb~\typeReal$,
\begin{align}
    \transformGradInfFD{t} \defas
    {\lambda \theta. \set{s \codePGets \codeSeed;~\transformInference{(x_A, x_B) \codePGets \transformPPEval{\transformModCoupling{t}}~\theta~s;~\codePReturn~(x_B - x_A)}}}.
    \label{eq:grad-inf-fd}
\end{align}
Given also an AD macro $\transformAD{\cdot}$,
we now define the \textit{gradient inference macro}
$\transformGradInfAD{\cdot}$,
acting on a term $\emptyEnv \vdash t : \typeReal^n \typeFunction \typeProb~\typeReal$
and producing a term $\emptyEnv \vdash \transformGradInfAD{t} : \typeReal^n \typeFunction \typeProb~\typeReal^n$,
\begin{align}
    \transformGradInfAD{t} \defas
    \lambda(\theta_1, \dots, \theta_n).
    \big(\transformAD{\lambda(\varepsilon_1, \dots, \varepsilon_n). \transformGradInfFD{t}~\left((\theta_1, \theta_1+\varepsilon_1), \dots, (\theta_n, \theta_n+\varepsilon_n)\right)}
    \,(0, \dots, 0)\big).
    \label{eq:grad-inf-ad}
\end{align}
%
\begin{listing}
    \footnotesize
    \centering
    \begin{sublisting}[t]{0.45\linewidth}
        \caption{Original Program $t$}
        \begin{plainCodeBoxFramed}{\textwidth}
            \LinesNotNumbered
            \begin{algorithm}[H]
                $\lambda(\theta : \typeReal).$
                \Block{
                    $b \codePGets \codeFlipCRN~\theta$\;
                    $\codePReturn~(\codeIfThenElse{b}{0.5}{0.3})$
                }
            \end{algorithm}
        \end{plainCodeBoxFramed}
        \label{fig:inference-and-ad-original}
    \end{sublisting}\hfill
    \begin{sublisting}[t]{0.525\linewidth}
        \caption{Inference Target $t_{\transformPPEval} \defas \transformPPEval{\transformModCoupling{t}}$}
        \begin{plainCodeBoxFramed}{\linewidth}
            \LinesNotNumbered
            \begin{algorithm}[H]
                $\lambda(\theta_A : \typeReal, \theta_B : \typeReal). \lambda(s : \typeSeed).$
                \Block{
                    $\codeLet~b_A = s < \theta_A$\;
                    $\codeLet~r = \begin{aligned}[t]
                        &\codeIf~b_A~
                            \codeThen~
                            \min({\theta_B}/{\theta_A}, 1)
                        \\[-3pt]
                        &\codeElse~
                        \max\left({(\theta_B - \theta_A)}/{(1 - \theta_A)}, 0\right)
                        \end{aligned}$\;
                    $b_B \codePGets \codeFlipEnum~r$\;
                    $\codePReturn~(\codeIfThenElse{b_A}{0.5}{0.3}, \codeIfThenElse{b_B}{0.5}{0.3})$
                }
            \end{algorithm}
        \end{plainCodeBoxFramed}
        \label{fig:inference-and-ad-target}
    \end{sublisting}

    \bigskip

    \begin{sublisting}[t]{0.45\linewidth}
        \caption{Difference Estimator $\transformGradInfFD{t}$}
        \begin{plainCodeBoxFramed}{\linewidth}
            \LinesNotNumbered
            \begin{algorithm}[H]
                $\lambda(\theta_A : \typeReal, \theta_B : \typeReal).$
                \Block{
                    $s \codePGets \codeSeed$\;
                    $\codeLet~b_A = s < \theta_A$\;
                    $\codeLet~r = \begin{aligned}[t]
                        &\codeIf~b_A~
                            \codeThen~
                            \min({\theta_B}/{\theta_A}, 1)
                        \\[-3pt]
                        &\codeElse~
                        \max\left({(\theta_B - \theta_A)}/{(1 - \theta_A)}, 0\right)
                        \end{aligned}$\;
                    $\codeLet~(b_{B, 1}, w_1) = (\codeTrue, r)$\;
                    $\codeLet~(b_{B, 2}, w_2) = (\codeFalse, 1 - r)$\;
                    $\codePReturn~
                    \begin{aligned}[t]
                    &w_1 \times (\codeIfThenElse{b_{B, 1}}{0.5}{0.3})\\[-3pt]
                    &+ w_2 \times (\codeIfThenElse{b_{B, 2}}{0.5}{0.3})\\[-3pt]
                    &- (w_1 + w_2) \times (\codeIfThenElse{b_A}{0.5}{0.3})
                    \end{aligned}$
                }
            \end{algorithm}
        \end{plainCodeBoxFramed}
        \label{fig:inference-and-ad-difference}
    \end{sublisting}\hfill
    \begin{sublisting}[t]{0.525\linewidth}
        \caption{Gradient Estimator $\transformGradInfAD{t}$}
        \begin{plainCodeBoxFramed}{\linewidth}
            \LinesNotNumbered
            \begin{algorithm}[H]
                $\lambda(\theta : \typeReal).$
                \Block{
                    $s \codePGets \codeSeed$\;
                    $\codeLet~b_A = s < \theta$\;
                    $\codeLet~\left(r, \dot{r}\right) =
                    \left(\begin{aligned}
                    &\codeIfThenElse{b_A}{1}{0}, \\[-5pt]
                    &\codeIfThenElse{b_A}{0}{{1}/{(1-\theta)}}
                    \end{aligned}\right)$\;
                    $\codeLet~(b_{B, 1}, (w_1, \dot{w}_1)) = (\codeTrue, (r, \dot{r}))$\;
                    $\codeLet~(b_{B, 2}, (w_2, \dot{w}_2)) = (\codeFalse, (1 - r, -\dot{r}))$\;
                    $\codePReturn~
                    \begin{aligned}[t]
                    &\dot{w}_1 \times (\codeIfThenElse{b_{B, 1}}{0.5}{0.3}) \\[-3pt]
                    &+ \dot{w}_2 \times (\codeIfThenElse{b_{B, 2}}{0.5}{0.3}) \\[-3pt]
                    &- (\dot{w}_1 + \dot{w}_2) \times (\codeIfThenElse{b_A}{0.5}{0.3})
                    \end{aligned}$
                }
            \end{algorithm}
        \end{plainCodeBoxFramed}
        \label{fig:inference-and-ad-grad}
    \end{sublisting}
    \captionsetup{aboveskip=2pt}
    \caption{Example derivation of a gradient estimator via gradient inference,
    for a simple input program $t$ that flips a biased coin and returns a value based on the outcome.}
    \label{fig:inference-and-ad}
    \vspace{-12pt}
\end{listing}

The term $\transformGradInfFD{t}$ in \cref{eq:grad-inf-fd}
first samples a seed $s$
to fix the values of the primal random choices in $t_{\transformPPEval}
= \transformPPEval{\transformModCoupling{t}}$.
Then, inference is applied to the residual choices in $t_{\transformPPEval}$
to produce an estimator of the expected difference $x_B - x_A$ of the coupled pair of
return values $(x_A, x_B)$.
To form the gradient estimator, the term $\transformGradInfAD{t}$
in \cref{eq:grad-inf-ad} constructs an anonymous function that accepts
perturbations $\varepsilon_1, \dots, \varepsilon_n$ to the components of
the parameter $\theta$.
AD is then applied to this function at $\varepsilon_1 = \dots = \varepsilon_n = 0$.
\cref{fig:inference-and-ad} illustrates each step of gradient inference on a toy
program, using forward-mode AD via dual numbers.
The inference step (\cref{fig:inference-and-ad-target} to \cref{fig:inference-and-ad-difference})
enumerates the two possible outcomes of the sample from $\codeFlipEnum$ in the inference target program
$t_{\transformPPEval}$,
while the AD step (\cref{fig:inference-and-ad-difference} to \cref{fig:inference-and-ad-grad})
differentiates the enumerated weights to produce a gradient estimator.
The soundness theorem for the full \systemname{} workflow is as follows.
\begin{restatable}[Soundness of Gradient Inference Macro]{theorem}{thmGradInferenceCorrectness}
    \label{thm:grad-inference-correctness}
    Let $\emptyEnv \vdash t : \typeReal^n \typeFunction \typeProb~\typeReal$
    be a well-typed term of $\langP$, where
    each primitive $p$ appearing in $t$
    satisfies
    $(\semP{p}, \semP{p}, \semP{\transformErasure{\transformModCoupling{p}}})
        \in \relationCoupling(\tau_p)$
    and
    $(\semP{\transformErasure{p}}, \semP{\transformPPEval{p}})
        \in \relationPPEval(\tau_p)$.
    If the inference macro applied in \cref{eq:grad-inf-fd}
    is sound (per \cref{eq:inference-soundness})
    and the AD macro applied in \cref{eq:grad-inf-ad}
    is sound (per \cref{eq:ad-soundness}),
    then for any $\theta : \setReal^n$,
    \begin{equation}
        \semExpect\left(\semP{\transformGradInfAD{t}}(\theta)\right) =
        \nabla_\theta \left[\semExpect(\semP{t}(\theta))\right].
    \end{equation}
\end{restatable}


\section{Applications}
\label{sec:applications}

We investigate the extent to which \systemname{} achieves its design goals
\labelcref{aim:modular,aim:novel} from \cref{sec:introduction-our-approach},
by considering the following research questions:
\begin{enumerate}[label=\textbf{(Q\arabic*)},wide=\parindent,leftmargin=*]
\item\label{question:variance}
Can \systemname{} express novel estimators that deliver lower variance than
state-of-the-art baseline estimators in the literature?

\item\label{question:bias}
Do the estimators in \systemname{} empirically exhibit zero bias, as the
theory predicts?

\item\label{question:runtime}
How large is the constant-factor runtime overhead in the prototype
implementation of \systemname{} (Haskell) compared to ADEV~\citep{lew2023adev} (Haskell)
and Storchastic~\citep{krieken2021storchastic} (PyTorch)?
\end{enumerate}
We study these questions in the context of three challenging case studies:
the M/M/c queuing model (\cref{sec:applications-queuing});
an option pricing model (\cref{sec:applications-option-pricing});
and
a gene transcription model (\cref{sec:applications-gene-transcription}).
The results confirm that \systemname{} meets its design goals, where
\systemname{} is able to express state-of-the-art existing gradient estimators
and improve upon them with novel estimators that achieve up to 370x
reductions in time-adjusted variance.
\Cref{appendix:gradient-inference} presents a number of additional examples.

\subsection{Experimental Setup}

We implemented a prototype of \systemname{} in Haskell
(\url{https://github.com/probsys/grad-inf}).
The workflow begins with a user-written probabilistic program $t$ similar to
the one in \cref{fig:overview-input-program}, which uses \systemname{}
probabilistic primitives
to define a probability distribution over a real output.
\systemname{} synthesizes an inference target program
$t_{\transformPPEval}$ from $t$, to which we apply probabilistic inference
algorithms using the \textsf{monad-bayes} library~\citep{scibior2018denotational}.
The final automatic differentiation step uses the \textsf{ad}
library~\citep{ad2025} (for forward-mode AD) or the \textsf{backprop}
library~\citep{backprop2025} (for reverse-mode AD).
To measure the performance of gradient estimators, we empirically measure
both
per-sample variance and the product of variance and wall-clock time.
Because variance decreases linearly with the number of samples,
the latter metric measures the computational work needed to obtain an
estimate with a desired Monte Carlo error.
Error bands with 95\% confidence intervals are provided for variance plots,
which are computed via bootstrapping.
\Cref{appendix:deferred-applications} gives a full list of experiment parameters.

\subsection{\texorpdfstring{\hyperref[question:variance]{\ref*{question:variance}}}{\ref{question:variance}} Developing Novel Estimators with Low Variance}
\label{sec:evaluation-efficiency}

\subsubsection{Sensitivity Analysis of Queuing Model (Overview Example)}
\label{sec:applications-queuing}

\begin{figure}[t]
\captionsetup{aboveskip=6pt}
\footnotesize
\centering
\includegraphics[width=\linewidth]{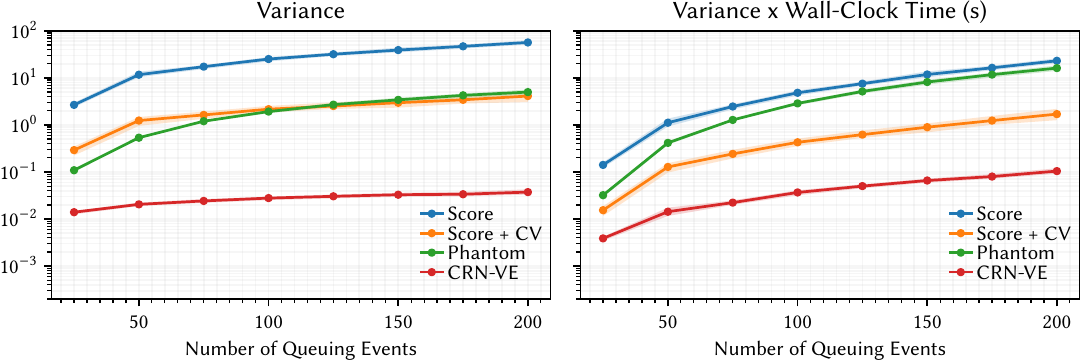}
\caption{Variance scaling plots for gradient estimators expressed by \systemname{} on M/M/c queue model.}
\label{fig:queue-performance}
\vspace{-10pt}
\end{figure}

We first revisit the application of \systemname{} to the M/M/c queue model
from \cref{sec:overview-model}.
By combining a $\codeFlipCRN$ primitive
with different inference algorithms,
\systemname{} can express the randomized phantom estimator
~\citep{fu1997conditional}
and enable a novel low-variance estimator (\GradInfA) based on variable elimination.

\Cref{fig:queue-performance} provides variance and time-adjusted variance asymptotics
against the number of queuing events $n$
for four estimators expressed by \systemname{}:
two score-function based estimators,
with and without a control variate (CV) set equal to the empirical mean of the objective;
the randomized phantom estimator; and the novel estimator \GradInfA.
The results show that \GradInfA{} exhibits reduced variance
relative to existing methods (e.g., 110x reduction at $n = 200$).
Since all four estimators have runtime linear in the number of queuing events $n$,
this variance reduction translates into improvements in time-adjusted variance (e.g.,
16x reduction at $n = 200$).

\subsubsection{Estimating Greeks of Stochastic Option Pricing Models}
\label{sec:applications-option-pricing}

\begin{figure}[t]
\footnotesize
\begin{subfigure}[b]{0.5\textwidth}
    \captionsetup{skip=10pt,belowskip=4pt}
    \begin{plainCodeBox}[left=0mm,right=0mm,colback=gray!5]{\linewidth}
    \begin{tikzpicture}[grow=right,sibling distance=2.6mm,level distance=18mm]
        \node[anchor=north,draw] at (2.12, 0.35) {$S_{i+1}$};

        \node[anchor=north west] (tree) at (0, -0.05) {
        \Tree
        [.\node[draw] (tree-top) {$S_i$};
        \edge node[auto=right,font=\tiny]{$p_d$};
        [.\node (tree-left) {$S_i \cdot d$};]
        \edge node[auto=left,xshift=1.5mm,yshift=-0.2mm,font=\tiny]{$p_s$};
        [.\node{$S_i$};]
        \edge node[auto=left,font=\tiny]{$p_u$};
        [.\node{$S_i \cdot u$};]
        ]
        };

        \node[anchor=north] at (1.5, -2.05) {$(i = 0,\dots,T/\Delta T-1)$};

        \node[anchor=north west] (B) at (3.0, 0.4) {
            $\begin{aligned}
                    u   & = e^{\sigma\sqrt{2 \Delta t}} \quad
                    d       = e^{-\sigma\sqrt{2 \Delta t}}                                                                                    \\
                    p_u & = \left(\frac{e^{r\Delta t/2} - e^{-\sigma\sqrt{\Delta t/2}}}{e^{\sigma\sqrt{\Delta t/2}} - e^{-\sigma\sqrt{\Delta t/2}}}\right)^{\!2} \\
                    p_d & = \left(\frac{e^{\sigma\sqrt{\Delta t/2}} - e^{r\Delta t/2}}{e^{\sigma\sqrt{\Delta t/2}} - e^{-\sigma\sqrt{\Delta t/2}}}\right)^{\!2} \\
                    p_s &= 1 - (p_u + p_d)
                \end{aligned}$
        };
    \end{tikzpicture}
    \end{plainCodeBox}
\captionsetup{skip=2pt}
\caption{Generative model}
\label{fig:option-pricing-tree-model}
\end{subfigure}
\hfill
\begin{subfigure}[b]{0.48\textwidth}
\centering
\includegraphics[width=\textwidth]{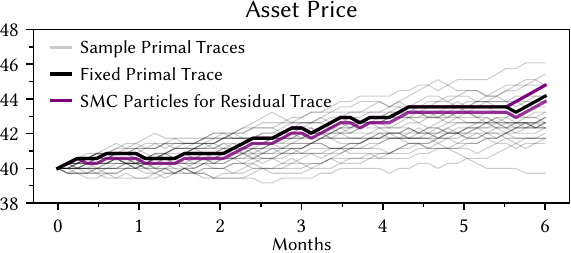}
\captionsetup{skip=10pt,belowskip=4pt}
\caption{Twisted SMC estimator via gradient inference}
\label{fig:option-pricing-smc-visualize}
\end{subfigure}

\bigskip

\begin{subfigure}[b]{\textwidth}
    \centering
    \begin{subfigure}[b]{.495\linewidth}
    \includegraphics[width=\linewidth]{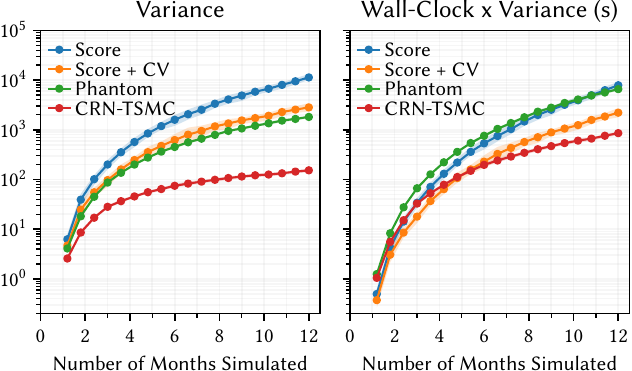}%
    \caption{Variance scaling asymptotics for $\kappa$}
    \label{fig:option-pricing-performance-kappa}
    \end{subfigure}\hfill
    \begin{subfigure}[b]{.495\linewidth}
    \includegraphics[width=\linewidth]{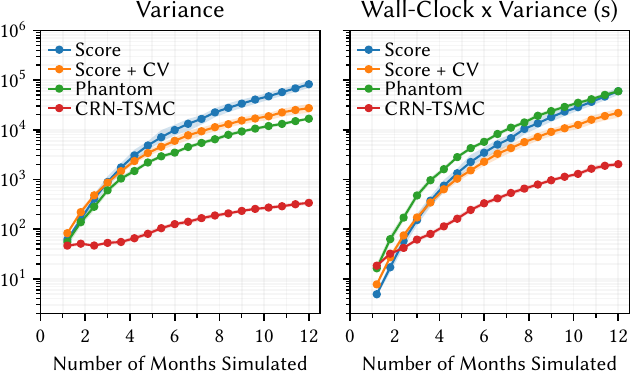}%
    \caption{Variance scaling asymptotics for $\rho$}
    \label{fig:option-pricing-performance-rho}
    \end{subfigure}
\end{subfigure}
\captionsetup{belowskip=-10pt, aboveskip=5pt}
\caption{Applying $\systemname$ to an option pricing model from computational finance.}
\label{fig:option-pricing}

\end{figure}

Sensitivities of the expected payoff of a financial option
are of fundamental interest in mathematical finance,
where they are known as ``Greeks''~\citep[\S7]{glasserman2003monte}.
\Cref{fig:option-pricing-tree-model} illustrates a widely used
option pricing model
known as the trinomial model~\citep{boyle1988lattice}.
The model evolves the price $S_i$
of an asset at discrete
time steps $i \Delta T$ ($i=1,2,\dots$) via a multiplicative random walk,
whose kernel depends on the interest rate $r$ and price
volatility $\sigma$.
We consider a European-style option, in which
an asset can be bought at a fixed time $T$ at price
$K$.
The quantity
$\ell \defas \mathbb{E}[e^{-rT} \max(S_{T/\Delta T} - K, 0)]$
gives the expected payoff of the option.
We apply \systemname{} to the
trinomial model
to estimate the Greeks $\kappa \defas \partial \ell/\partial \sigma$
and $\rho \defas \partial \ell/\partial r$,
with parameters $r = 15\% / \mathrm{year}$,
$\sigma = 5\% / \mathrm{year}$, $S_0 = 40$, and $K = 41$.

In \systemname{},
we express the kernel of the random walk
using a $\codeCategoricalCRN$ primitive, which uses a
CRN-based factorized coupling
(see \cref{appendix:gradient-inference}).
Using the same stratification-based scheme as for the M/M/c queue from
\cref{sec:applications-queuing}, \systemname{} produces an SPA phantom estimator,
which is state-of-the-art for discrete option pricing models~\citep{glasserman2003monte}.
However, as \cref{fig:option-pricing-performance-kappa,fig:option-pricing-performance-rho}
show, the phantom estimator gives limited performance improvements
compared to a simpler score-based method with control variates.
We thus swap in a more sophisticated inference algorithm based on
twisted SMC~\citep[\S10.3.3]{chopin2020introduction}
(details in \cref{appendix:gradient-inference-gradinfb}).
\Cref{fig:option-pricing-smc-visualize} depicts the final SMC particles
from a particular inference run, with opacities corresponding to particle weights.
\Cref{fig:option-pricing-performance-kappa,fig:option-pricing-performance-rho} show
that the resulting novel estimator (\GradInfB{})
achieves significant reductions in time-adjusted variance
as the time period $T$ grows (e.g.,
11x reduction for $\rho$ at $T = 1.0$ year).

\subsubsection{Parameter Inference of Gene Transcription Model}
\label{sec:applications-gene-transcription}

We apply \systemname{} to the task of parameter inference of a chemical reaction network (RN) model
of gene transcription.
RN models have seen broad adoption in systems
biology for simulating the dynamics of biochemical reactions~\citep{chandran2008mathematical}.
Following \citet{anderson2012efficient},
we consider a chemical RN model of the copy numbers
of mRNA $M$ and protein $P$, with four reactions
$
    \emptyset \to^\alpha M,
$
$
    M \to^{\beta} M + P,
$
$
    M \to^{\gamma} \emptyset,
$
and
$
    P \to^{\delta} \emptyset,
$
corresponding to gene transcription, mRNA translation, mRNA degradation,
and protein degradation, respectively.
Each reaction determines a Poisson process in time,
with rate given as the product
of the rate parameter ($\alpha, \beta, \gamma$ or $\delta$) and
the current copy number of the reactant species.
We simulate the RN using the
Gillespie algorithm~\citep{gillespie1977exact}
shown in \cref{fig:gene-transcription-gillespie-algorithm},
which relies on a \emph{discrete} sampling
step from a categorical distribution over the possible reactions.
Our goal is to estimate the
four parameters $(\alpha, \beta, \gamma, \delta)$
that reproduce observed
mean copy numbers $\hat{M}$ and $\hat{P}$
of mRNA and protein
(\cref{fig:gene-transcription-model-traces}, dotted lines),
produced at hidden ground truth parameters.
We minimize a relative mean squared error (MSE) loss
$\ell \defas \mathbb{E}\left[(\tilde{M} - \hat{M})^2 / \hat{M}^2 +
    (\tilde{P} - \hat{P})^2 / \hat{P}^2\right]$
between $k$-sample Monte Carlo estimates $\tilde{M}$ and $\tilde{P}$
of the mean copy numbers and their ground truth values
$\hat{M}$ and $\hat{P}$.

\begin{figure}[t]
\begin{subfigure}[b]{0.27\textwidth}
    \footnotesize
    \setlength{\abovedisplayskip}{0pt}
    \setlength{\belowdisplayskip}{0pt}
    \begin{plainCodeBox}[left=0mm,right=0mm,colback=gray!5]{\linewidth}
    \begin{align*}
        &S \defas \begin{bmatrix}
            1 & 0 & -1 & 0 \\
            0 & 1 & 0 & -1
        \end{bmatrix} \\
        &(M_0, P_0) \defas \mbox{<initial values>}\\
        &t_0 \defas 0 \\
        &\mbox{\underline{for $i=1,2,3,\dots$}} \\
        &\mathbf{r}_i \defas \left(\alpha, \beta M_{i-1}, \gamma M_{i-1}, \delta P_{i-1}\right) \\
        &s \defas \mathbf{r}_{i,0} + \mathbf{r}_{i,1} + \mathbf{r}_{i,2} + \mathbf{r}_{i,3}\\
        &k_i \sim \mathrm{Cat}\left(\mathbf{r}_i\right) \quad \Delta t_i \sim \mathrm{Exp}(s) \\
        &t_i \defas t_{i-1} + \Delta t_i \\
        &M_i \defas M_{i-1} + S_{0, k_i} \\
        &P_i \defas P_{i-1} + S_{1, k_i}
    \end{align*}
    \end{plainCodeBox}
    \caption{Generative model}
    \label{fig:gene-transcription-gillespie-algorithm}
\end{subfigure}\hfill
\begin{subfigure}[b]{0.225\textwidth}
    \includegraphics[width=\textwidth]{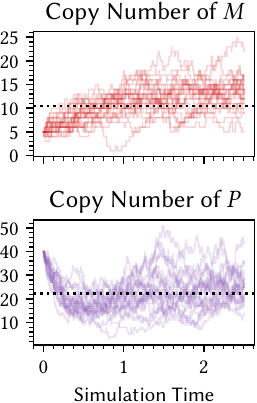}
    \caption{Model traces}
    \label{fig:gene-transcription-model-traces}
\end{subfigure}\hfill
\begin{subfigure}[b]{0.4\textwidth}
    \includegraphics[width=\textwidth]{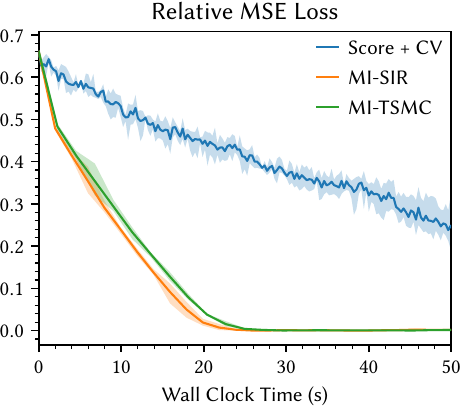}
    \caption{Stochastic gradient descent traces}
    \label{fig:gene-transcription-sgd}
\end{subfigure}
\captionsetup{aboveskip=4pt,belowskip=0pt}
\caption{Applying \systemname{} to infer the hidden parameters of a probabilistic gene transcription model.}
\label{fig:gene-transcription}
\end{figure}
\begin{table}[t]
\captionsetup{aboveskip=0pt,belowskip=-5pt}
\caption{Variance of gradient estimators on gene transcription model at
$\alpha = \beta = \gamma = \delta = 1$ with $k = 100$.
Reductions in variance and time-adjusted variance are given
with respect to the Score + CV baseline.}
\label{tab:gene-transcription-performance}
\footnotesize

\sisetup{
  round-mode=places,
  round-precision=3,
  exponent-product=\times,
}

\begin{tabular*}{\linewidth}{r@{\hspace{0.5em}\extracolsep{\fill}}rrrrrrr}
\toprule
& \multicolumn{3}{c}{Variance} & \multicolumn{2}{c}{Variance Reduction} & \multicolumn{2}{c}{Variance Reduction (Time Adj.)}
\\ \cmidrule{2-4} \cmidrule{5-6} \cmidrule{7-8}
{Rate} & Score + CV & \textsf{MI-SIR} & \textsf{MI-TSMC} & of \textsf{MI-SIR} & of \textsf{MI-TSMC} & of \textsf{MI-SIR} & of \textsf{MI-TSMC}
\\ \midrule
$\alpha$ & $1.7 \cdot 10^{-1}$  & $1.6 \cdot 10^{-4}$ & $5.8 \cdot 10^{-5}$ & \textbf{1100x} & \textbf{2900x} & \textbf{160x} & \textbf{370x} \\
$\beta$  & $2.0 \cdot 10^{-1}$ & $5.2 \cdot 10^{-4}$ & $4.1 \cdot 10^{-4}$ & \textbf{380x} & \textbf{490x} & \textbf{61x} & \textbf{67x} \\
$\gamma$ & $2.0 \cdot 10^{-1}$ & $1.0 \cdot 10^{-3}$ & $4.4 \cdot 10^{-4}$ & \textbf{190x} & \textbf{450x} & \textbf{30x} & \textbf{61x} \\
$\delta$ & $3.9 \cdot 10^{-1}$ & $3.1 \cdot 10^{-3}$ & $2.6 \cdot 10^{-3}$ & \textbf{130x} & \textbf{150x} & \textbf{19x} & \textbf{20x} \\
\bottomrule
\end{tabular*}
\vspace{-12pt}
\end{table}

The score function estimator, called the ``Girsanov transformation method''
in the context of chemical RNs,
is the state-of-the-art unbiased
estimator in this setting but still exhibits high variance~\citep{wang2016efficiency}.
We apply \systemname{} to contribute two novel estimators to the literature.
The first combines the $\codeCategoricalMaxIndep$ primitive with stratified
sampling (\GradInfC).
The second employs a twisted sequential Monte Carlo algorithm (\GradInfD),
whose details are given in \cref{appendix:gradient-inference-gradinfd}.
\Cref{tab:gene-transcription-performance} shows the variance of these
gradient estimation schemes when differentiating the objective $\ell$
at fixed parameters, using synthetic observation data
generated at hidden parameters $\alpha = 18$, $\beta = 8$, $\gamma = 1.5$,
and $\delta = 4$ with simulation time period $T = 2.5$.
\GradInfC{}
improves efficiency by 19x--160x,
while \GradInfD{} improves efficiency by 20x--370x.
These variance reductions lead to better optimization performance using
stochastic gradient descent (SGD) on the loss $\ell$ (\cref{fig:gene-transcription-sgd}).
With a fixed computational budget,
both \GradInfC{} and \GradInfD{} find parameters that
reproduce the observed data (relative MSEs of 0.2\% and 0.1\%, respectively)
more accurately than the parameters found using the score-based estimator
(relative MSE of 26\%).

\subsection{\texorpdfstring{\hyperref[question:bias]{\ref*{question:bias}}}{\ref{question:bias}} Verifying Unbiased Estimation}
\label{sec:evaluation-bias}

We empirically verify that the novel estimators, while lower variance,
remain unbiased
as guaranteed by \cref{thm:grad-inference-correctness}.
\Cref{fig:verify-unbiased} shows box plots of samples of gradient estimates
from the novel estimators and standard baselines.
The case study parameters are scaled down to ensure that
the baselines are sufficiently low variance to serve as
accurate ground truths.
The empirical means of the estimates are labeled in each plot.
A paired equivalence test~\citep{lakens2018equivalence}
verifies that the empirical means of the novel estimators
\GradInfA{}, \GradInfB{}, \GradInfC{}, and \GradInfD{}
are all statistically indistinguishable from those of the baselines
($p < 1\%$, tolerance $\Delta = 5\%$).

\vspace{-5pt}
\begin{figure}[H]
\footnotesize
\centering
\begin{subfigure}[b]{0.32\textwidth}
    \centering
    \includegraphics[width=\textwidth]{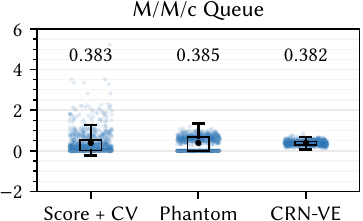}
\end{subfigure}
\hfill
\begin{subfigure}[b]{0.32\textwidth}
    \centering
    \includegraphics[width=\textwidth]{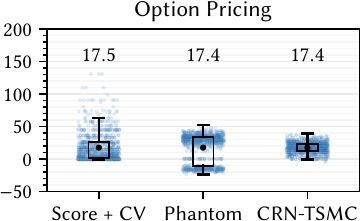}
\end{subfigure}
\hfill
\begin{subfigure}[b]{0.32\textwidth}
    \centering
    \includegraphics[width=\textwidth]{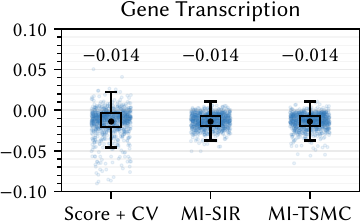}
\end{subfigure}

\captionsetup{belowskip=-8pt, aboveskip=8pt}
\caption{Gradient estimate samples from the estimators in each case study,
with their empirical means reported.}
\label{fig:verify-unbiased}
\end{figure}

\subsection{\texorpdfstring{\hyperref[question:runtime]{\ref*{question:runtime}}}{\ref{question:runtime}} Quantifying Constant-Factor Runtime Overhead}
\label{sec:evaluation-runtime}

We quantify the runtime overhead of our prototype implementation of \systemname{} in Haskell
by analyzing the score estimator, which is
a standard baseline that is expressible in systems such as
ADEV~\citep{lew2023adev} and Storchastic~\citep{krieken2021storchastic}.
\Cref{fig:runtime-overhead} compares the runtime of the score estimator across the systems
in two settings: the M/M/c queue model
(\cref{sec:applications-queuing})
and the option pricing model (\cref{sec:applications-option-pricing}).
Both models feature a large number of loop iterations, leading to deep computation graphs.
The runtime performance is considered
as a function of problem size.
\systemname{}
has improved runtime scaling relative to Storchastic and is
competitive with ADEV (within 3x runtime overhead for the queue model, and
6x for the option pricing model).
%

\vspace{-5pt}
\begin{figure}[H]
\footnotesize
\centering

\includegraphics[height=0.24\textwidth]{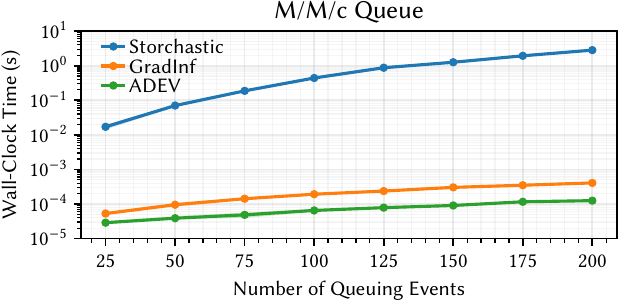}
\includegraphics[height=0.24\textwidth,trim=5mm 0 0 0,clip]{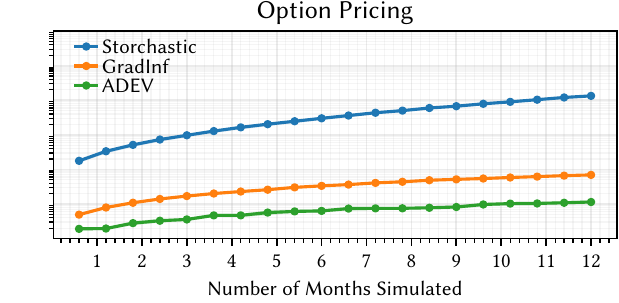}

\captionsetup{belowskip=-8pt, aboveskip=4pt}
\caption{Runtime scaling of the score estimator for different gradient estimation systems.}
\label{fig:runtime-overhead}
\end{figure}

\section{Related Work}
\label{sec:related}

\noindentparagraph{\normalfont\bfseries Differentiable Programming Systems.}
Widely used libraries for probabilistic and differentiable programming
such as
Pyro~\citep{bingham2019}
and
TensorFlow Probability~\citep{dillon2017}
provide users with predefined gradient estimators based on REINFORCE or
the reparameterization trick, but do not support the design or application
of most estimators in \cref{tab:gradient-estimation-schemes}.
Research libraries such as
ADEV~\citep{lew2023adev,becker2024probabilistic},
StochasticAD.jl~\citep{arya2022automatic},
and
Storchastic~\citep{krieken2021storchastic} provide modular
approaches to deriving gradient estimators of probabilistic programs.
These libraries build gradient estimators for whole programs by composing
gradient estimators attached to each primitive.
However, many estimators of interest do not arise as compositions of per-primitive
estimators (e.g., the randomized phantom estimator~\citep{fu1997conditional}).
A key insight in \systemname{} is to attach to each primitive a factorized coupling, rather
than an entire gradient estimator.
This modular design permits the subsequent application of probabilistic inference
and the construction of novel estimators as in \cref{sec:applications}.

\noindentparagraph{\normalfont\bfseries Gradient Estimation Methods.}
Gradient estimation has a vast literature, across which the design,
practice, and even naming of different gradient estimation schemes vary
widely~\citep{mohamed2020monte}.
In this literature, coupling, factorization (or conditioning more generally), and
probabilistic inference have been recognized as important ideas in varying capacities.
For instance, the application of REINFORCE to stochastic computation graphs~\citep{schulman2015gradient}
has incorporated
inference methods based on control variates~\citep{weber2019credit},
importance sampling~\citep{bornschein2015reweighted,le2020revisiting},
and Rao-Blackwellization~\citep{liu2019raoblackwellized}.
However, these applications are restricted to REINFORCE-based
estimators, preventing the use of inference to improve upon the
gradient estimation schemes shown in \cref{tab:gradient-estimation-schemes}.
Smoothed perturbation analysis (SPA) uses coupling and conditioning
for gradient estimation~\citep{fu1997conditional,dai2000perturbation,gong1987}.
However, SPA uses simple Monte Carlo or subsampling-based methods.
It has not been recognized that richer probabilistic inference schemes can be
modularly inserted into the gradient estimation workflow.

\noindentparagraph{\normalfont\bfseries Couplings of Probabilistic Programs.}
Coupling is a central technique in Monte Carlo methods~\citep{lindvall2002lectures},
and has been used by the programming languages community for
relational reasoning about probabilistic programs~\citep{barthe2015relational,barthe2017coupling},
including in the higher-order context~\citep{gregersen2024asynchronous,haselwarter2025approximate,sato2019formal,tassarotti2019separation}.
In \systemname{} we use coupling for a different purpose, namely programmable
variance reduction.
Specifically, rather than proving a relational property between programs,
we aim to form low-variance estimators of the difference in expectation of a single probabilistic
program evaluated at two different inputs.
The insights from couplings used in the relational context could prove
complementary to our goal, where techniques in the literature such as shift
couplings~\citep{aguirre2018relational} may be repurposed to produce
lower variance gradient estimators.

\noindentparagraph{\normalfont\bfseries Differentiating Probabilistic Inference Algorithms.}
AD techniques~\citep{tjoa2025,pfanschilling2022} have been
applied to optimize the likelihoods from probabilistic programming
languages with exact inference~\citep{saad2021,holtzen2020}.
Other works apply gradient estimation to optimize the parameters of proposal
distributions specified by inference algorithms such as sequential Monte
Carlo (SMC)~\citep{le2018autoencoding,maddison2017,naesseth2018variational}.
In contrast, our work is about using probabilistic inference to construct
new gradient estimation schemes.
Our inference algorithms, which operate on coupled and
factorized programs, are designed to admit a simple interchange
of expectation and derivative with respect to the finite perturbation
$\varepsilon$ in the final AD step~\cref{eq:overview-ad-Ediff}.
We do not aim to learn the parameters of the probabilistic inference
schemes themselves.

\noindentparagraph{\normalfont\bfseries Noninterference in Probabilistic Programs.}
Prior work has used information-flow typing
to automatically detect conditional independence in user models and exploit it
for improved posterior inference in Bayesian models, e.g., the imperative
while-loop language SlicStan~\citep{gorinova2019probabilistic,gorinova2022conditional},
and the higher-order language MaPPL~\citep{li2024compiling}.
MaPPL performs automated variable elimination (VE) via a language based on security labels,
in which a candidate variable to be eliminated is labeled ``high'' and the labels of
all other variables are inferred accordingly.
In contrast, $\langPP$ builds on the syntax of the dependency core calculus~\citep{abadi1999core}
to provide constructs for
a probabilistic version
of \emph{partial evaluation}, in which a subset of random variables in a coupled program
are sampled and fixed so that the remainder can be targeted by a generic probabilistic
inference algorithm.
While VE is one such algorithm, its use is as a generic inference algorithm
applied subsequent to partial evaluation,
of which \cref{tab:gradient-estimation-schemes} gives several other examples.


\section{Limitations}
\label{sec:limitations}

\noindentparagraph{\normalfont\bfseries System Expressiveness.}
\systemname{} supports programs featuring higher-order functions, finite iteration,
and both discrete and continuous probability distributions, matching the expressiveness of
prominent techniques in the PL literature such as ADEV~\citep{lew2023adev}.
These programs already cover many challenging problems;
\cref{tab:gradient-estimation-schemes} covers several key classes of gradient
estimators~\citep{mohamed2020monte},
including specialized gradient estimators
used for applications such as discrete variational autoencoders.
However, certain applications still pose a challenge to gradient inference:
\begin{itemize}[wide=0pt, leftmargin=*, itemsep=5pt]
    \item Differentiating programs featuring \emph{unbounded recursion} or \emph{infinite data structures},
    which arise in simulations for particle showers in high-energy physics~\citep{kagan2023branches}.
    A challenge for such programs is
    forming statistically effective couplings that handle
    variable-size loops and data structures.
    Addressing this challenge would require extending our coupling transformation
    $\transformCoupling{\cdot}$ with more sophisticated coupling strategies
    for stochastic control flow.
    \item Handling \emph{parametric discontinuities} involving continuous random variables~\citep{michel2024distributions},
    which arise in queuing theory~\citep{fu1997conditional} and
    differentiable rendering~\citep{li2018differentiable,bangaru2022differentiable}.
    The challenge here is that there is no direct combination of coupling,
    factorization, and probabilistic inference algorithm to which standard AD
    can be correctly applied (cf.~\cref{fig:inference-and-ad}).
    \item Expressing gradient estimation strategies based on
    \emph{sophisticated control variates}, which are common in
    reinforcement learning~\citep{mohamed2020monte}.
    These control variates are often learned during optimization using approaches such as
    actor-critic methods~\citep{sutton1998reinforcement},
    which \systemname{} does not support.
\end{itemize}
\noindentparagraph{\normalfont\bfseries Automating Primitives.}
\systemname{} requires users to annotate
each primitive in their program to specify a suitable factorized coupling.
While this design enables flexibility, automating the selection of factorized couplings,
or even the design of factorized coupling primitives themselves
(e.g., by leveraging work on automatic symbolic disintegration~\citep{narayanan2020}),
could improve automation.
\noindentparagraph{\normalfont\bfseries GPU Acceleration and Batching.}
The prototype implementation of \systemname{} does not support
GPU-accelerated execution and batched evaluation.
It would be useful to implement support for these features,
following the approach of systems such as Storchastic~\citep{krieken2021storchastic} and
ADEV JAX~\citep{becker2024probabilistic}.

\noindentparagraph{\normalfont\bfseries Higher-Order Gradients.}
This work has focused on estimating first-order gradients.
The problem of estimating higher-order gradients
has been studied
for several classes of gradient estimators~\citep{mohamed2020monte}.
Generalizing these ideas to the gradient inference setting
is a promising future direction.


\section{Conclusion}
We present \textit{gradient inference}, a new approach to gradient estimation
that enables the modular construction of sound and efficient gradient
estimators.
Our technique converts a gradient estimation problem into a
probabilistic inference problem, to which
powerful probabilistic inference algorithms can be applied, followed by
standard automatic differentiation.
This reduction rests on synthesizing an intermediate probabilistic program
that serves as an inference target, leveraging provably sound program
transformations for coupling and factorization to reduce variance.
We evaluate our system, $\systemname{}$, on challenging case studies,
showing it can express diverse estimators and enables the construction of
new estimators that outperform state-of-the-art baselines.

Designing effective couplings, factorizations, and probabilistic inference
algorithms is key to our approach to gradient estimation.
In this paper, we demonstrate several examples of how these techniques
can be combined to arrive at effective estimators.
Because gradient inference is compositional, future
improvements to any one of these three components will automatically
deliver benefits for the entire workflow.

\section*{Data-Availability Statement}
A Haskell library with all the algorithms described in this article
is available at \url{https://github.com/probsys/grad-inf}.
A reproducible artifact for the evaluation in \cref{sec:applications}
is available on Zenodo~\citep{arya_2026_19673572}.

\begin{acks}
The authors thank the referees for helpful feedback.
This material is based upon work supported by the
\grantsponsor{NSF}{National Science Foundation}{https://doi.org/10.13039/100000001}
under Grant No.~\grantnum{NSF}{2311983}.
Any opinions, findings, and conclusions or recommendations
in this material are those of the authors and do not necessarily
reflect the views of the NSF.
\end{acks}

\bibliographystyle{ACM-Reference-Format}
\citestyle{acmnumeric}
\bibliography{paper}

\AtEndDocument{\clearpage

\appendix

\section*{Supplementary Material}

The supplementary material contains the following appendices.

\begin{itemize}[topsep=0pt, wide=0pt]
  \item \Cref{appendix:gradinf-usage}:
    Demonstration of the \systemname{} Haskell prototype on an example problem.
  \item \Cref{appendix:gradient-inference}:
    Explanation of how each estimator in \cref{tab:gradient-estimation-schemes} is expressed
    using gradient inference.
  \item \Cref{appendix:deferred-preliminaries,appendix:deferred-coupling,appendix:deferred-partial-evaluation,appendix:deferred-gradient-inference}:
    Full definitions and deferred proofs for the technical core.
  \item \Cref{appendix:deferred-applications}:
    Full list of numerical parameters for the case studies.
\end{itemize}


\section{End-to-End Example Usage of the \systemname{} Haskell Prototype}
\label{appendix:gradinf-usage}

In this section, we provide an end-to-end walkthrough of using the \systemname{} Haskell prototype
to perform gradient estimation via gradient inference.

\subsection{High-Level API}

The main entry point to our \systemname{} Haskell prototype is the function \texttt{gradInfAD}
for computing gradient estimates.
A typical usage of \texttt{gradInfAD} takes the following shape:

\noindent
\begin{minipage}{\linewidth}
\begin{lstlisting}
do
    -- define the model program...
    gradientEstimates <- replicateM numSamples $
                         gradInfAD program inferenceAlg theta
    -- use the gradient estimates...
\end{lstlisting}
\end{minipage}

The three inputs to the \texttt{gradInfAD} function are:
\begin{itemize}[wide,leftmargin=*,topsep=5pt,itemsep=5pt]
    \item \emph{An Input Probabilistic Program.}
    The user-written probabilistic program, with Haskell type of the form
    \texttt{f d -> m d}. Here, \texttt{d} represents the type of real numbers,
    \texttt{f} is a functor that allows for more flexible differentiation parameters,
    and \texttt{m} is the probability monad.
    An important consideration is that the input program must be written polymorphically
    in the monad \texttt{m} and the numerical type \texttt{d}.
    The \systemname{} Haskell prototype leverages this polymorphism
    to implement \systemname{}'s source-to-source transformations \emph{within} the language.
    \item
    \begin{sloppypar}
    \emph{An Inference Algorithm.}
    The inference algorithm to be used for gradient estimation.
    The \systemname{} Haskell prototype uses the monad-bayes library~\citep{scibior2018functional}
    to provide predefined inference algorithms,
    such as \texttt{stratifiedImportanceResamplingInferenceAlg}, \texttt{variableEliminationInferenceAlg}, and
    \texttt{twistedSMCInferenceAlg}.
    \end{sloppypar}
    \item \emph{A Differentiation Parameter.}
    The input point at which to compute the gradient, of type \texttt{f Double}.
    Here, the functor \texttt{f} can be any functor supporting the Haskell \texttt{Traversable} and \texttt{Applicative} classes
    (e.g., \texttt{Identity} for a single parameter, \texttt{ZipList} for lists of parameters, and so forth).
\end{itemize}
The $\texttt{gradInfAD}$ function employs the \textsf{ad}~\citep{ad2025} library
for forward-mode automatic differentiation.
A separate function \texttt{gradInfReverseAD} of similar signature is provided for reverse-mode automatic differentiation
using the \textsf{backprop}~\citep{backprop2025} library.

\subsection{Writing the Input Probabilistic Program}

Let us implement the queuing model from \cref{sec:overview}.
The code below is complete runnable Haskell
for the version of the \systemname{} prototype provided
in the artifact~\citep{arya_2026_19673572},
making use of the following imports.

\noindent
\begin{minipage}{\linewidth}
\begin{lstlisting}
{-# LANGUAGE RebindableSyntax #-}
{-# LANGUAGE ScopedTypeVariables #-}

import Prelude
import Numeric.GradInf.Primitives.DeterministicPrimitives
import Numeric.GradInf.Primitives.FlipCRN
import Numeric.GradInf.Primitives.IterateP
import Numeric.GradInf
import Control.Monad
import Control.Monad.Bayes.Sampler.Lazy (Sampler)
import Data.Functor.Identity
\end{lstlisting}
\end{minipage}

We begin by defining the Markov kernel for the model, which simulates a single queuing event.

\noindent
\begin{minipage}{\linewidth}
\begin{lstlisting}
queueKernel :: forall m d i b mat.
    (DeterministicPrimitives d i b mat, FlipCRN m d b)
    => d -> i -> m i
queueKernel theta x = do
    let p = theta /
            (theta + if (isGreater x 25) :: b then 25 else toDouble x)
    b :: b <- flipCRN p
    let x' = if b then x + 1 else x - 1
    return x'
\end{lstlisting}
\end{minipage}

Here, note that
the program is polymorphic in the type variables
\texttt{m}, \texttt{d}, \texttt{i}, \texttt{b}, and \texttt{mat},
which are generic stand-ins for the probability monad, doubles, integers,
booleans, and matrices, respectively.
The typeclass \texttt{DeterministicPrimitives}
provides a number of standard numerical operations in Haskell,
as well as new primitives in cases where the existing Haskell
options are not sufficiently generic.
The typeclass \texttt{FlipCRN} provides the $\codeFlipCRN$ primitive
from \cref{sec:overview}.
We now complete the model by iterating the kernel to
simulate $n$ queuing events.

\noindent
\begin{minipage}{\linewidth}
\begin{lstlisting}
queueModel :: forall m d i b mat.
    (DeterministicPrimitives d i b mat, FlipCRN m d b, IterateP m i)
    => Int -> d -> m d
queueModel n theta = do
    x <- iterateP (queueKernel theta) 0 !\LstBangBang! n
    return (toDouble x)
\end{lstlisting}
\end{minipage}

We can sample the model at $\theta = 15$ and $n = 50$,
corresponding
to the right panel of \cref{fig:overview-sample-traces}.

\noindent
\begin{minipage}{\linewidth}
\begin{lstlisting}
getSamples :: Sampler [Double]
getSamples = replicateM 5 (queueModel (50 :: Int) (15.0 :: Double))

-- >>> sampler getSamples
-- [12.0, 10.0, 10.0, 18.0, 18.0]
\end{lstlisting}
\end{minipage}

\subsection{Performing Gradient Estimation}

We are now ready to perform gradient estimation with \texttt{gradInfAD}.
First, we call \texttt{gradInfAD} with \texttt{stratifiedImportanceResamplingInferenceAlg},
corresponding to the estimator $Z_{\mathrm{SIR}}$ in \cref{sec:overview}.
We report the sample mean and variance of $100$ Monte Carlo gradient estimates at $n = 50$ and $\theta = 15$.

\noindent
\begin{minipage}{\linewidth}
\begin{lstlisting}
getGradientEstimateSIR :: Int -> Double -> Sampler Double
getGradientEstimateSIR n theta = fmap runIdentity $
      gradInfAD (queueModel n . runIdentity)
      stratifiedImportanceResamplingInferenceAlg (Identity theta)

getMeanAndVarianceSIR :: IO ()
getMeanAndVarianceSIR = do
    samples <- sampler (replicateM 100 (getGradientEstimateSIR 50 15.0))
    let n = length samples
    let mu = sum samples / fromIntegral n
    let var = sum (map (\x -> (x - mu) * (x - mu)) samples)
              / fromIntegral (n - 1)
    print (mu, var)

-- >>> getMeanAndVarianceSIR
-- (0.6459758088306011,0.5114738156744449)
\end{lstlisting}
\end{minipage}

We can swap out the inference algorithm for
\texttt{variableEliminationInferenceAlg},
corresponding to the estimator $Z_{\mathrm{VE}}$ in \cref{sec:overview}.

\noindent
\begin{minipage}{\linewidth}
\begin{lstlisting}
getGradientEstimateVE :: Int -> Double -> Sampler Double
getGradientEstimateVE n theta = fmap runIdentity $
      gradInfAD (queueModel n . runIdentity)
      variableEliminationInferenceAlg (Identity theta)

getMeanAndVarianceVE :: IO ()
getMeanAndVarianceVE = do
    samples <- sampler (replicateM 100 (getGradientEstimateVE 50 15.0))
    let n = length samples
    let mu = sum samples / fromIntegral n
    let var = sum (map (\x -> (x - mu) * (x - mu)) samples)
              / fromIntegral (n - 1)
    print (mu, var)

-- >>> getMeanAndVarianceVE
-- (0.6532304786017948,2.211323360216621e-2)
\end{lstlisting}
\end{minipage}
The estimator remains unbiased, but the variance is significantly reduced.

\clearpage
\section{Expressing Gradient Estimation Strategies in \systemname{}}
\label{appendix:gradient-inference}

In this appendix, we give further details on how the existing
and novel estimators in \cref{tab:gradient-estimation-schemes}
are realized via gradient inference.
For each estimator, we define its factorized coupling
primitive
and explain the inference algorithm that, when differentiated with AD
in the manner formalized in \cref{sec:gradient-inference},
yields the desired estimator.

\subsection{Pathwise Estimator}
\label{appendix:gradient-inference-pathwise}

The pathwise estimator is a widely-used gradient estimator
for continuous probability distributions.
When it is applicable,
it has often been observed to be low-variance~\citep{mohamed2020monte}.
(The estimator is not applicable in any of the challenging cases
considered in \cref{sec:applications}.)
Existing AD tools can readily express the pathwise estimator
by differentiating a probabilistic program at a fixed random seed.

In \systemname{}, the pathwise estimator is expressed using
a common-random-number coupling for each continuous
probabilistic primitive, with the selection of the common random
number $\omega$ marked as primal.
The $\codeNormalCRN$ primitive below gives an example.
\noindent
\begin{plainCodeBox}[left=-2mm,box align=top]{\textwidth}
\small
\begin{algorithm}[H]
    $\transformModCoupling{\codeNormalCRN} \defas
    \lambda \left(((\mu_A, \bar{\mu}_B), (\sigma_A, \bar{\sigma}_B)) : (\typeReal \typeProduct \typeResidual~\typeReal)^2\right).$
    \Block{
        $\omega \codePPGets \codePPPrimal~(\codeNormal~(0,1))$\;
        $\codeLet~x_A = \mu_A + \sigma_A \cdot \omega$\;
        $\codeLet~\bar{x}_B = \{\mu_B \codeRGets \bar{\mu}_B;~\sigma_B \codeRGets \bar{\sigma}_B;~\codeRHide~(\mu_B + \sigma_B \cdot \omega)\}$\;
        $\codePPReturn~(x_A, \bar{x}_B)$ \;
    }
\end{algorithm}
\end{plainCodeBox}
In an arbitrary $\langP$ program composed of such primitives,
the inference target deterministically evaluates the program
on a perturbed input $\theta + \mathbf{\varepsilon}$, using the same
underlying random numbers as for the original input $\theta$.
Applying AD to this target (with the inference transformation set to a trivial no-op)
yields the pathwise estimator for the program.
Informally, we may say that the pathwise estimator is expressed in \systemname{}
as a \emph{differentiated coupling}.

As an example, we consider the continuous Black-Scholes option pricing model,
based on a geometric Brownian motion
$\diff S_t  = S_t \left[r \diff t + \sigma \diff W_t \right]$
that is sampled from using the
Euler-Maruyama time-discretization.
The same parameters and setup are used
as for the discrete trinomial option pricing model
in \cref{sec:applications-option-pricing}.
\cref{fig:option-pricing-pathwise-visualize}
shows a sample from the coupled program produced by \systemname{},
where we make a finite perturbation $\varepsilon = 10^{-2}$ to the parameter $\sigma$.
The coupled residual trace (blue) closely tracks the primal trace (black):
the pathwise estimator can be understood as differentiating the final value of the
residual trace with respect to the perturbation $\varepsilon$, conditional upon
the fixed primal trace.

\begin{figure}[H]
    \footnotesize
    \centering
    \includegraphics[width=\textwidth]{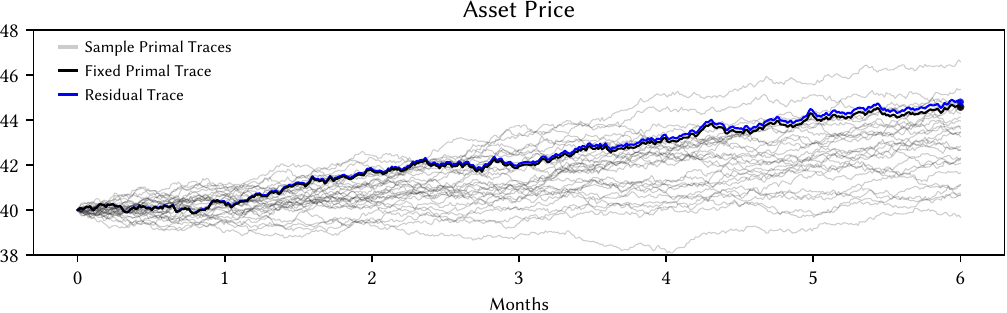}
    \caption{The coupled inference target synthesized by \systemname{} for a Black-Scholes option pricing model.
    Applying AD to this target yields the pathwise estimator for the program.}
    \label{fig:option-pricing-pathwise-visualize}
\end{figure}

\subsection{Phantom Estimators from Smoothed Perturbation Analysis (SPA)}
\label{appendix:gradient-inference-spa}

The phantom estimator from SPA~\citep{fu1997conditional}
is expressed in \systemname{}
using a common-random-number coupling for each discrete
probabilistic primitive in the program.
In these couplings, all sampling decisions that determine
the first components of the returned coupled pairs are
marked as primal.
Other sampling decisions can be marked
as either primal or residual:
this design freedom mirrors the
manual choice of a $\sigma$-algebra to condition upon
in the conventional formation of an SPA estimator~\citep{fu1997conditional}.
The $\codeFlipCRN$ and $\codeCategoricalCRN$ primitives below give
two examples.
\noindent
\begin{plainCodeBox}[left=-2mm,box align=top]{\textwidth}
\small
\begin{algorithm}[H]
    $\transformModCoupling{\codeFlipCRN} \defas
    \lambda (p_A: \typeReal, \bar{p}_B: \typeResidual~\typeReal).$
    \Block{
        $b_A \codePPGets \codePPPrimal~(\codeFlip~p_A)$\;
        $\bar{b}_B \codePPGets \codePPResidual$
        \Block{
            $p_B \codeRGets \bar{p}_B$\;
            $\codeLet~r = \codeIfThenElse{b_A}
                {\min\left(\frac{p_B}{p_A}, 1\right)}
                {\max\left(\frac{p_B - p_A}{1 - p_A}, 0\right)}$\;
            $\codeRHide~\left(
                \codeFlipEnum
                ~r\right)$
        }
        $\codePPReturn~(b_A, \bar{b}_B)$ \;
    }
\end{algorithm}
\end{plainCodeBox}
\label{appendix:gradient-inference-categorical-crn}
\noindent
\begin{plainCodeBox}[left=-2mm,box align=top]{\textwidth}
\small
\begin{algorithm}[H]
    $\transformModCoupling{\codeCategoricalCRN} \defas$ \;
    $\;\lambda (z : (\alpha \typeProduct \typeResidual~\alpha)^n, p : (\typeReal \typeProduct \typeResidual~\typeReal)^n).$
    \Block{
    $\codeLet~z_A = \codeMap~\codeFst~z$\;
    $\codeLet~p_A = \codeMap~\codeFst~p$\;
    $i_A \codePPGets \codePPPrimal~\left(\codeCategorical~((1, \dots, n), p_A)\right)$\;
    $\codeLet~x_A = \codeIndex~z_A~i_A$\;
    $\codeLet~s_A = \codeScanl~(+)~0~p_A$\;
    $\codeLet~s_{A,l} = \codeIndex~s_A~(i_A-1);~
        \codeLet~s_{A,h} = \codeIndex~s_A~i_A$\;
    $\bar{x}_B \codePPGets \codePPResidual$
    \Block{
        $p_{B} \codeRGets \codeSequenceR~(\codeMap~\codeSnd~p)$\;
        $z_B \codeRGets \codeSequenceR~(\codeMap~\codeSnd~z)$\;
        $\codeLet~s_B = \codeScanl~(+)~0~p_B$\;
        $\codeLet~f = \lambda(s_{B,l}, s_{B,h}). \max\big(0, (\min(s_{A,h}, s_{B,h}) - \max(s_{A,l}, s_{B,l}))/(s_{A,h} - s_{A,l})\big)$\;
        $\codeLet~r = \codeMap~(\lambda i. f~(\codeIndex~s_{B}~(i-1), \codeIndex~s_B~i))~(1, \dots, n)$\;
        $\codeRHide~(\codeCategoricalENUM~(z_B, r))$\;
    }
    $\codePPReturn~(x_A, \bar{x}_B)$\;
    }
\end{algorithm}
\end{plainCodeBox}
These primitives make use of syntactic sugar for
standard functional operations on $n$-tuples, which are defined in \cref{notation:syntactic-sugar}.
Observe that, rather than sampling a precise seed $\omega$ as in the implementations
of $\transformCoupling{\codeFlipCRN}$ and $\transformCoupling{\codeCategoricalCRN}$,
the factorized couplings
$\transformModCoupling{\codeFlipCRN}$ and $\transformModCoupling{\codeCategoricalCRN}$
sample only from the first marginal in the primal,
so that the conditional distribution of the second component is formed in the residual.

Equipped with these factorized couplings,
the phantom estimator can be expressed
by applying stratified importance (re)sampling (SIR)
to the inference target synthesized by \systemname{}.
As described in \cref{sec:gradient-inference}, SIR
partitions the probability space into events $A_0, A_1, \dots, A_n$,
and estimates an expectation $\mathbb{E}[X]$
by sampling $X_i \sim X \mid A_i$ and computing
$\sum_{i=0}^n X_i \mathbb{P}(A_i)$.
To produce the phantom estimator,
we let $A_0$ be the event where the residual trace
is equal to the primal trace, while $A_1, A_2, \dots$
correspond to the events where the residual trace first diverges from
the primal trace at a particular step of the program (i.e.,
the coupled pair(s) returned by a primitive have differing first
and second components).

In the SPA literature, the
additional samples $X_1 \sim X \mid A_1, X_2 \sim X \mid A_2, \dots$
are known as \emph{phantoms},
and the event $A_1 \cup A_2 \cup \dots$ is known as
the \emph{critical event}~\citep{fu1997conditional,gong1987}.
Unlike in SPA, these objects are formed independently
of the differentiation logic: the inference algorithm is valid
for any finite choice of perturbation $\varepsilon$.
\systemname{} recovers the phantom estimator from SPA
by applying standard AD to the inference algorithm itself.

A \emph{randomized} phantom estimator has also been proposed in the SPA literature~\citep[Corollary 3.1]{fu1997conditional},
in which only one stratified event is sampled, significantly reducing computational cost.
The randomized phantom estimator can be expressed by performing a Monte Carlo subsampling
step across the $n$ phantom events.
Specifically, we can sample an index $i \sim q(\cdot)$ from a proposal distribution $q$,
and replace the estimate $\sum_{i=0}^n X_i \mathbb{P}(A_i)$ with
the (also unbiased, but usually higher variance) estimate
$X_0 \mathbb{P}(A_0) + \frac{1}{q(i)} X_i \mathbb{P}(A_i)$.
This approach is illustrated in \cref{fig:overview-inference} for the M/M/c queue model,
and also applied in \cref{sec:applications-option-pricing} for the trinomial option pricing model.

\subsection{Measure-Valued Differentiation}
\label{appendix:gradient-inference-mvd}

The approach of measure-valued differentiation (MVD)~\citep{heidergott2008measurevalued}
is traditionally viewed as distinct from SPA~\citep{mohamed2020monte}.
However, in \systemname{}, we can actually express MVD schemes using a similar workflow to SPA,
but with a different choice of coupling, whose soundness can be directly verified.
The $\codeFlipAlwaysFlip$ and $\codePoissonFixedProp$ primitives below are factorized couplings
that reproduce MVDs for Bernoulli and Poisson distributions, respectively.
\noindent
\begin{plainCodeBox}[left=-2mm,box align=top]{\textwidth}
\small
\begin{algorithm}[H]
    $\transformModCoupling{\codeFlipAlwaysFlip} \defas
    \lambda (p_A: \typeReal, \bar{p}_B: \typeResidual~\typeReal).$
    \Block{
        $b_A \codePPGets \codePPPrimal~\left(\codeFlip~p_A\right)$\;
        $\codeLet~b^* = \neg~b_A$\;
        $\bar{b}_B \codePPGets \codePPResidual$
        \Block{
            $p_B \codeRGets \bar{p}_B$\;
            $\codeLet~\Delta p = \codeIfThenElse{b_A}{p_B - p_A}{p_A - p_B}$\;
            $\codeRHide~\left(\codeCategoricalENUM~\left(
                (b_A, b^*), (1 + \Delta p, -\Delta p)\right)
                \right)$\;
        }
        $\codePPReturn~(b_A, \bar{b}_B)$\;
    }
\end{algorithm}
\end{plainCodeBox}
\noindent
\begin{plainCodeBox}[left=-2mm,box align=top]{\textwidth}
\small
\begin{algorithm}[H]
    $\transformModCoupling{\codePoissonFixedProp} \defas
    \lambda (\lambda_A: \typeReal, \bar{\lambda}_B: \typeResidual~\typeReal).$
    \Block{
        $x_A \codePPGets \codePPPrimal~\left(\codePoisson~\lambda_A\right)$\;
        $\codeLet~x^* = x_A + 1$\;
        $\bar{x}_B \codePPGets \codePPResidual$
        \Block{
            $\lambda_B \codeRGets \bar{\lambda}_B$\;
            $\codeLet~w = \bigl(\textstyle\sum_{i=0}^{x_A}\left(\mathrm{Poisson}(\lambda_B)(i) - \mathrm{Poisson}(\lambda_A)(i)\right)\bigr) \div \mathrm{Poisson}(\lambda_A)(x_A)$\;
            $\codeRHide~\left(\codeCategoricalENUM~\left(
                (x_A, x^*), (1+w, -w)
            \right)\right)$\;
        }
        $\codePPReturn~(x_A, \bar{x}_B)$\;
    }
\end{algorithm}
\end{plainCodeBox}
Note that the calls to $\codeCategoricalENUM$ in these factorized coupling primitives
construct measures that may potentially be signed (see \cref{remark:measure-monad-on-qbs}).
However, these primitives nevertheless represent
sound factorized couplings of the original probability distributions, according
to the condition in \cref{thm:factorized-coupling-correctness}.
Let us understand the soundness of the $\codeFlipAlwaysFlip$ primitive,
in which a primal sample $b_A$ is drawn from $\codeFlip~p_A$,
while the residual sample $\bar{b}_B$ is drawn from a measure with two atoms:
$\bar{b}_B = b_A$, with weight $1 + \Delta p$, and $\bar{b}_B = \neg b_A$, with weight $-\Delta p$.
We verify that $\transformErasure{\transformModCoupling{\codeFlipAlwaysFlip}}$ is a valid coupling:
the first marginal is correct
since $b_A$ is drawn from $\codeFlip~p_A$, while the second marginal is correct because
$\bar{b}_B$ is set to true with expected weight
\begin{equation*}
    p_A \times (1 + (p_B - p_A)) + (1 - p_A) \times (p_B - p_A) = p_A + p_A p_B - p_A^2 + p_B - p_A - p_A p_B + p_A^2 = p_B.
\end{equation*}
Phantom estimators based on MVDs have been proposed in the literature~\citep{heidergott2008measurevalued}.
These gradient estimators are expressed in \systemname{}
using the same SIR inference algorithm as for SPA (\cref{appendix:gradient-inference-spa}), followed
by standard AD.

\subsection{Score Method}
\label{appendix:gradient-inference-score}

While \systemname{} estimators are based on couplings,
the vanilla score-function estimator only takes a \emph{single} sample from the model,
which is scored by a weight.
Expressing the score-function estimator in \systemname{} requires a lightweight extension
to the system.
Specifically, we extend the grammar of $\langP$ with a new constant $\colforsyntax{0}$,
i.e. $c \defas \cdots \mid \colforsyntax{0}$,
and add a typing rule $\Gamma \vdash \colforsyntax{0} : \tau$ for each base type $\tau$.
To adjust the denotational model, we augment the interpretation of each
base type $\sigma$ with the constant $\semP{\colforsyntax{0}} = \mathbf{0}$,
e.g., $\semP{\typeReal} \defas \mathbb{R} \cup \{\mathbf{0}\}$,
and adjust the semantics of terms to preserve the constant $\mathbf{0}$.
In this manner, $\colforsyntax{0}$ behaves
similarly to a \emph{bottom element}.
We modify the final subtraction map in \cref{eq:grad-inf-fd} to replace this constant with an arbitrary constant
control variate $\mathit{CV}$,
i.e., we replace ``$x_B - x_A$'' with
$
(\codeIfThenElse{x_B = \colforsyntax{0}}{\mathit{CV}}{x_B}) - (\codeIfThenElse{x_A = \colforsyntax{0}}{\mathit{CV}}{x_A}),
$
where $\mathit{CV}$ is to be chosen as part of the inference algorithm.
That gradient inference with this extension remains sound follows from the fact that, in expectation, the events
$x_A = \mathbf{0}$ and $x_B = \mathbf{0}$ must both have zero weight for
the factorized and coupled program to be sound.
The primitives $\codeNormalWeighted$ and $\codeFlipWeighted$ below give \systemname{} primitives implementing the
score method for Gaussian and Bernoulli distributions.
\noindent
\begin{plainCodeBox}[left=-2mm,box align=top]{\textwidth}
\small
\begin{algorithm}[H]
    $\transformModCoupling\set{\codeNormalWeighted} \defas
    \lambda ((\mu_A, \mu_B) : \typeReal^2, (\sigma_A, \sigma_B): \typeReal^2).$
    \Block{
        $x_A \codePPGets \codePPPrimal~(\codeNormal~(\mu_A, \sigma_A))$\;
        $x_B \codePPGets \codePPResidual$
        \Block{
            $\codeRHide~\bigg(\codeCategoricalENUM~\bigg(((x_A, x_A), (\colforsyntax{0}, x_A)), \left(1, \frac{N(\mu_B, \sigma_B)(x_A)}{N(\mu_A, \sigma_A)(x_A)} - 1\right)\bigg)\bigg)$\;
        }
        $\codePPReturn~(x_A, x_B)$\;
    }
\end{algorithm}
\end{plainCodeBox}
\noindent
\begin{plainCodeBox}[left=-2mm,box align=top]{\textwidth}
\small
\begin{algorithm}[H]
    $\transformModCoupling\set{\codeFlipWeighted} \defas
    \lambda (p_A : \typeReal, p_B: \typeReal).$
    \Block{
        $b_A \codePPGets \codePPPrimal~(\codeFlip~p_A)$\;
        $b_B \codePPGets \codePPResidual$
        \Block{
            $\codeRHide~\bigg(\codeCategoricalENUM~\bigg(((b_A, b_A), (\colforsyntax{0}, b_A)), \left(1, \frac{\codeIfThenElse{b_A}{p_B}{1-p_B}}{\codeIfThenElse{b_A}{p_A}{1-p_A}} - 1\right)\bigg)\bigg)$\;
        }
        $\codePPReturn~(b_A, b_B)$\;
    }
\end{algorithm}
\end{plainCodeBox}
The soundness of these (somewhat unconventional)
factorized coupling primitives can be directly verified.
Considering the $\codeFlipWeighted$ primitive for concreteness, note that the primitive
takes a single sample $b_A$ from the first marginal of the output
coupled pair, which is then used to form the pairs
$(b_A, b_A)$ and $(\colforsyntax{0}, b_A)$.
The first pair $(b_A, b_A)$, with weight 1, ensures that the first marginal is correct.
To ensure the correctness of the second marginal, the second pair $(\colforsyntax{0}, b_A)$ is weighted
(i.e., \emph{scored})
by
$\frac{\mathbb{P}_{b \sim \mathrm{Ber}(p_B)}[b = b_A]}{\mathbb{P}_{b \sim \mathrm{Ber}(p_A)}[b = b_A]} - 1$,
so that the total weight assigned to the sample $b_A$ for the second component is precisely
the likelihood ratio
$\frac{\mathbb{P}_{b \sim \mathrm{Ber}(p_B)}[b = b_A]}{\mathbb{P}_{b \sim \mathrm{Ber}(p_A)}[b = b_A]}$.

We use a stratified IS inference algorithm that maintains
two unique weighted particles, one a pair of the form $(x, x)$
and the other a pair of the form $(\colforsyntax{0}, x)$.
Since all sampled values in the score-based primitives
take this form, this inference algorithm is \emph{deterministic}, and
hence amenable to standard AD.
The result reproduces the single-sample REINFORCE estimator:
for instance, in the case of $\codeFlipWeighted$,
observe that the derivative with respect to $p_B$ of the weight
$\frac{\mathbb{P}_{b \sim \mathrm{Ber}(p_B)}[b = b_A]}{\mathbb{P}_{b \sim \mathrm{Ber}(p_A)}[b = b_A]} - 1$,
taken at $p_B = p_A$,
is precisely the score $\frac{\mathrm{d}}{\mathrm{d} p_A} \log \mathbb{P}_{b \sim \mathrm{Ber}(p_A)}[b = b_A]$.

\subsection{Discrete Variational Autoencoder Estimators}
\label{appendix:gradient-inference-vae}

We consider expressing four prominent gradient estimators from the literature
on discrete variational autoencoders (VAEs)~\citepapp{rolfe2017discrete}:
REINFORCE Leave-One-Out
(RLOO), DisARM~\citep{dong2020disarm}, BitFlip1~\citep{kunes2023gradient}, and
straight-through~\citep{bengio2013estimating}.
These estimators apply to models that
generate ``layers'' of $n$ independent Bernoulli variables.
\Cref{fig:bernoulli-latents-toy-model,fig:bernoulli-latents-linear-vae}
show two such models.
\Cref{fig:bernoulli-latents-toy-model} specifies a challenge
problem designed by \citet{kunes2023gradient} using a logistic model.
\Cref{fig:bernoulli-latents-linear-vae} specifies a full discrete VAE model
with a linear encoder and decoder~\citep{dong2020disarm,kunes2023gradient}.

To express the gradient estimation schemes,
we register new factorizing coupling primitives
for sampling a Bernoulli layer with $n$ variables.
The $\codeFlipLayerSameOrIndep{,n}$ primitive
used for 2-sample RLOO is defined below.
\noindent
\begin{plainCodeBox}[left=-2mm,box align=top]{\textwidth}
\small
\begin{algorithm}[H]
    $\transformModCoupling{\codeFlipLayerSameOrIndep{,n}} \defas
        \lambda (p: (\typeReal \times \typeResidual~\typeReal)^n).$
    \Block{
        $\codeLet~p_A = \codeMap~\codeFst~p$\;
        $b_1 \codePPGets \codePPPrimal~(\codeFlipLayerSameOrIndep~p_A)$\;
        $b_2 \codePPGets \codePPPrimal~(\codeFlipLayerSameOrIndep~p_A)$\;
        $\codePPResidual$
        \Block{
            $p_B \codeRGets \codeSequenceR~(\codeMap~\codeSnd~p)$\;
            $\codeLet~f = \lambda ((b_1, b_2), (p_A, p_B)).$
            \Block{
                $\codeLet~(p, \Delta p) = \codeIfThenElse{b_1}{(p_A, p_B - p_A)}{(1 - p_A, p_A - p_B)}$\;
                $\left(\frac{\Delta p}{2p}, -\frac{\Delta p \cdot [b_1 \neq b_2]}{2p(1-p)}\right)$\;
            }
            $\codeLet~w_1 = \codeSum~(\codeMap~\codeFst~(\codeMap~f~(\codeZip~(\codeZip~b_1~b_2)~(\codeZip~p_A~p_B))))$\;
            $\codeLet~w_2 = \codeSum~(\codeMap~\codeSnd~(\codeMap~f~(\codeZip~(\codeZip~b_1~b_2)~(\codeZip~p_A~p_B))))$\;
            $\codeLet~b_{11} = \codeZip~b_1~b_1$\;
            $\codeLet~b_{12} = \codeZip~b_1~b_2$\;
            $\codeRHide~\left(\codeCategoricalENUM~\left(\left(b_{11}, b_{12}\right),
            (1 + w_1, w_2)\right)\right)$\;
        }
    }
\end{algorithm}
\end{plainCodeBox}
The primitive $\transformModCoupling{\codeFlipLayerSameOrAnti{,n}}$ for DisARM is defined below.
\noindent
\begin{plainCodeBox}[left=-2mm,box align=top]{\textwidth}
\small
\begin{algorithm}[H]
    $\transformModCoupling{\codeFlipLayerSameOrAnti{,n}} \defas
        \lambda (p: (\typeReal \times \typeResidual~\typeReal)^n).$
    \Block{
        $\codeLet~p_A = \codeMap~\codeFst~p$\;
        $\omega \codePPGets \codePPPrimal~(\codeSequenceP~(\codeMap~(\lambda \_. \codeUniform)~(1, \dots, n)))$\;
        $\codeLet~b_1 = \codeMap~(\lambda (p_A, \omega). \omega < p_A)~(\codeZip~p_A~\omega)$\;
        $\codeLet~b_2 = \codeMap~(\lambda (p_A, \omega). \omega > 1 - p_A)~(\codeZip~p_A~\omega)$\;
        $\codePPResidual$
        \Block{
            $p_B \codeRGets \codeSequenceR~(\codeMap~\codeSnd~p)$\;
            $\codeLet~f = \lambda ((b_1, b_2), (p_A, p_B)).$
            \Block{
                $\codeLet~(p, \Delta p) = \codeIfThenElse{b_1}{(p_A, p_B - p_A)}{(1 - p_A, p_A - p_B)}$\;
                $\left(\frac{\Delta p}{2p}, -\frac{\Delta p \cdot [b_1 \neq b_2]}{2\min(p, 1-p)}\right)$\;
            }
            $\codeLet~w_1 = \codeSum~(\codeMap~\codeFst~(\codeMap~f~(\codeZip~(\codeZip~b_1~b_2)~(\codeZip~p_A~p_B))))$\;
            $\codeLet~w_2 = \codeSum~(\codeMap~\codeSnd~(\codeMap~f~(\codeZip~(\codeZip~b_1~b_2)~(\codeZip~p_A~p_B))))$\;
            $\codeLet~b_{11} = \codeZip~b_1~b_1$\;
            $\codeLet~b_{12} = \codeZip~b_1~b_2$\;
            $\codeRHide~\left(\codeCategoricalENUM~\left(\left(b_{11}, b_{12}\right),
            (1 + w_1, w_2)\right)\right)$\;
        }
    }
\end{algorithm}
\end{plainCodeBox}
These factorized couplings (which do not appear directly in the original works)
may be proven sound in \systemname{} by verifying, via a direct computation
of the marginals, that
they denote sound $\langP$ couplings once the factorization annotations are erased.
As for the primitive $\codeFlipLayerAlwaysFlip{,n}$,
we can define it compositionally using the $\codeFlipAlwaysFlip$
primitive that we used to express the MVD estimator in \cref{appendix:gradient-inference-mvd}.
In particular,
we define $\codeFlipLayerAlwaysFlip{,n}$
as \emph{sugar} for
$
\lambda (p : \typeReal^n). \codeSequenceP~(\codeMap~\codeFlipAlwaysFlip~p).
$
(Here, $\codeSequenceP$ sequences a collection of measures in the standard way,
see \cref{notation:syntactic-sugar}.)

For the primitives $\codeFlipLayerSameOrIndep{,n}$ and $\codeFlipLayerSameOrAnti{,n}$,
we employ
a stratified IS probabilistic inference algorithm, followed
by standard reverse-mode AD.
We recover the 2-sample RLOO estimator
from $\codeFlipLayerSameOrIndep{,n}$ and
the DisARM estimator from $\codeFlipLayerSameOrAnti{,n}$.
To express the BitFlip1 estimator of \citet{kunes2023gradient}, we pair the $\codeFlipLayerAlwaysFlip{,n}$
primitive with stratified importance resampling, i.e., the same inference scheme that is used to express
the randomized phantom estimator (\cref{sec:overview,appendix:gradient-inference-spa}).

Finally, we express the straight-through estimator~\citep{bengio2013estimating},
which is biased but has exhibited good empirical performance in many
settings~\citepapp{shekhovtsov2021reintroducing}.
The estimator arises from the same $\codeFlipLayerAlwaysFlip{,n}$ primitive,
but a different
\emph{approximate}
inference algorithm based on a first-order moment closure~\citepapp{schnoerr2017approximation},
wherein each primitive distribution is replaced by its mean
(under composition, we approximate $\expect[XY] \approx \expect[X] \expect[Y]$ for random variables $X$ and $Y$,
and $\expect[f(X)] \approx f(\expect[X])$ for a random variable $X$ and function $f$).
When applying this inference scheme to the residual distribution constructed in
$\codeFlipAlwaysFlip$,
the expectation of the second marginal of the
$\codeCategoricalENUM$ primitive when $b_A = \semFalse$ is
computed as
$
   0 \cdot (1 + p_A - p_B) + 1 \cdot (p_B - p_A)
   = p_B - p_A
$
and when $b_A = \semTrue$ is computed as
$
   1 \cdot (1 + p_B - p_A) + 0 \cdot (p_A - p_B)
    = 1 + p_B - p_A,
$
which in both cases is equal to $(\semIfThenElse{b_A}{1}{0}) + (p_B - p_A)$.
Applying reverse-mode AD to this inference scheme yields the straight-through estimator
of \citet{bengio2013estimating}.
Thus, in gradient inference, the \emph{biased} straight-through estimator can be understood as the composition of
an \emph{exact} factorized coupling with an \emph{approximate} inference algorithm.

\cref{fig:bernoulli-latents-optimization} shows training loss curves
for the two models (where the \texttt{digits} dataset from scikit-learn~\citepapp{pedregosa2011scikit}
is used for the VAE model),
using the implementation of the four estimators in \systemname{}.

\begin{figure}[H]
\captionsetup[subfigure]{skip=-5pt}
\footnotesize

\begin{minipage}[t]{0.44\textwidth}
\begin{subfigure}[t]{\linewidth}
\caption{Logistic model}
\begin{plainCodeBox}[colback=gray!5]{\linewidth}
$\begin{aligned}
    &p_i \defas \sigma(\eta_i) \defas 1/(1 + e^{-\eta_i})\\
    &\highlightBox{highlightColorBernoulliDeep}{b_i  \sim \mathrm{Ber}(p_i)} \\
    &\quad (i = 1, \ldots, n) \\
    &\mathrm{loss} \defas \left(\textstyle\sum_i b_i - 0.499\right)^2 \\
\end{aligned}$
\end{plainCodeBox}
\captionsetup{skip=0pt}
\label{fig:bernoulli-latents-toy-model}
\end{subfigure}

\bigskip

\begin{subfigure}[t]{\linewidth}
\caption{Discrete VAE model}
\begin{plainCodeBox}[colback=gray!5]{\linewidth}
$\begin{aligned}
    &\mbox{\underline{training data} } (x_1, \dots, x_M) : ([0,1]^n)^M \\
    &\mbox{\underline{encoder} } E_\theta: \setReal^{d \times n},
        \mbox{\underline{decoder} } D_\theta : \setReal^{n \times d}\\
    &k \sim \mathrm{Unif}(0,M),\; x \defas x_k,\; \eta \defas E_\theta x \\
    &p_i \defas \sigma(\eta_i) \defas 1/(1 + e^{-\eta_i}),
        \highlightBox{highlightColorBernoulliDeep}{b_i  \sim \mathrm{Ber}(p_i)} \\
    &\quad (i = 1, \ldots, n); \mu \defas D_\theta b \\
    &\mathrm{loss} \begin{aligned}[t]
        &\defas \textstyle\sum_j \log \left(1/\mathrm{Ber}(\sigma(\mu_j))(x)\right) + \\
        &\textstyle\sum_{i=1}^{n}
            \left(\log \left(\mathrm{Ber}(0.5)(b_i)/\mathrm{Ber}(p_i)(b_i)\right)\right)
        \end{aligned}
\end{aligned}$
\end{plainCodeBox}
\captionsetup{skip=0pt}
\label{fig:bernoulli-latents-linear-vae}
\end{subfigure}
\end{minipage}%
\hfill
\begin{minipage}[t]{0.54\textwidth}
\begin{subfigure}[t]{\linewidth}
\caption{Training loss curves}
\includegraphics[width=\linewidth]{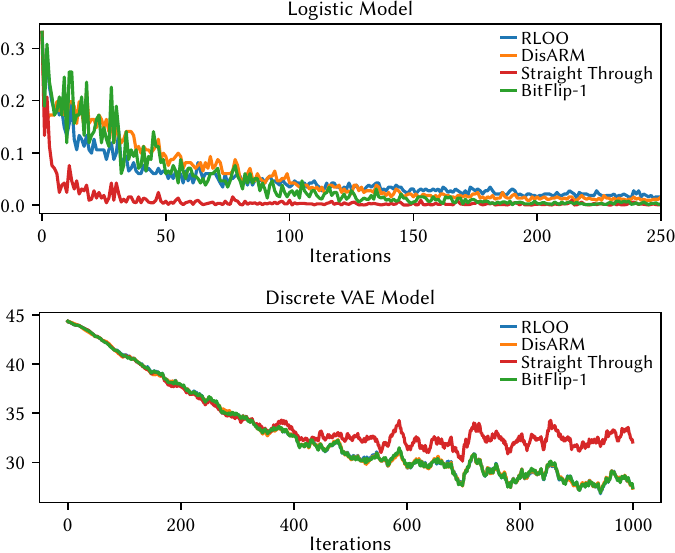}
\label{fig:bernoulli-latents-optimization}
\end{subfigure}
\end{minipage}

\caption{Discrete-latent benchmark models and training loss curves for VAE gradient inference experiments.
The sample from the Bernoulli latent layer in each model is highlighted using \highlightBernoulliName,
which is implemented using \systemname{} probabilistic primitives.}
\label{fig:bernoulli-latents}
\end{figure}

\subsection{Estimator \GradInfA{}}

The \GradInfA{} estimator makes use of the same $\codeFlipCRN$ primitive
that is used for SPA phantom estimators (\cref{sec:overview,appendix:gradient-inference-spa}).
Variable elimination is used to exactly compute the expected value of the
residual distribution.
Applied to the Markovian M/M/c queue model (\cref{fig:overview-inference}),
this inference algorithm enumerates
all possible values of the residual variable $x_B$ at each stage of the program.
The algorithm runs in linear time via dynamic programming
whenever the number of possible values of $x_B$ conditional on a
fixed primal trace is bounded by a constant
(in the M/M/c queue model, this constant is $2$).

\subsection{Estimator \GradInfB{}}
\label{appendix:gradient-inference-gradinfb}

The \GradInfB{} estimator makes use of the $\codeCategoricalCRN$ primitive
(\cref{appendix:gradient-inference-spa}).
For this estimator, we leverage a
twisted Sequential Monte Carlo (twisted SMC) algorithm~\citep[\S10.3.3]{chopin2020introduction}.
This inference algorithm is a variant of SMC that incorporates
\emph{twist} functions $\gamma_k : U_k \to \setReal$,
where $\gamma_k$ is used to scale the weights of particles used in the $k$th resampling step of SMC,
and where $U_k$ is the space of possible particle values at the $k$th step.
For the \systemname{}-synthesized inference target for the
trinomial option pricing model, which is a Markov chain on the state space $\setInteger \times \setInteger$,
we use the twist function $\gamma_k(x_{A,1:k}, x_{B,1:k}) = |x_{B,k} - x_{A,k}|$,
which increases the probability of resampling residual traces that diverge by a large amount from the primal trace.

\subsection{Estimator \GradInfC{}}
\label{appendix:gradient-inference-gradinfc}

The \GradInfC{} estimator makes use of the
$\codeCategoricalMaxIndep$ primitive below, which implements a factorized
maximal coupling with independent residuals~\citep{wang2021maximal}
for categorical distributions.
\noindent
\begin{plainCodeBox}[left=-2mm,box align=top]{\textwidth}
\small
\begin{algorithm}[H]
    $\transformModCoupling{\codeCategoricalMaxIndep} \defas$ \;
    $\;\lambda (z : (\alpha \typeProduct \typeResidual~\alpha)^n, p : (\typeReal \typeProduct \typeResidual~\typeReal)^n).$
    \Block{
    $\codeLet~(z_A,p_A) = (\codeMap~\codeFst~z, \codeMap~\codeFst~p)$\;
    $i_A \codePPGets \codePPPrimal~\left(\codeCategorical~((1, \dots, n), p_A)\right)$\;
    $\codeLet~x_A = \codeIndex~z_A~i_A$\;
    $\bar{x}_B \codePPGets \codePPResidual$
    \Block{
    $p_{B} \codeRGets \codeSequenceR~(\codeMap~\codeSnd~p)$\;
    $z_B \codeRGets \codeSequenceR~(\codeMap~\codeSnd~z)$\;
    $b \codePGets \codeFlip~(\min(1, (\codeIndex~p_B~i_A) / (\codeIndex~p_A~i_A)))$\;
    $\codeIf~b~\codeThen~(\codeRHide~(\codePReturn~(\codeIndex~z_B~i_A)))~\codeElse$
    \Block{
        $\codeLet{}~f = (\lambda(p_A, p_B). \max(0, p_B - p_A))$\;
        $\codeLet~r = \codeMap~f~(\codeZip~p_A~p_B)$\;
        $\codeRHide~(\codeCategoricalENUM~(z_B, \codeNormalize~r))$\;
    }
    }
    $\codePPReturn~(x_A, \bar{x}_B)$ \;
    }
\end{algorithm}
\end{plainCodeBox}

Stratified importance resampling is used for inference, as in the randomized phantom estimator
(\cref{sec:overview,appendix:gradient-inference-spa}).

\subsection{Estimator \GradInfD{}}
\label{appendix:gradient-inference-gradinfd}

As with \GradInfC{}, the \GradInfD{} estimator makes use of a
$\codeCategoricalMaxIndep$ primitive.
For inference, we employ
twisted SMC, as described in \cref{appendix:gradient-inference-gradinfb} for \GradInfB{}.
For the \systemname{}-synthesized inference target for the gene transcription model,
note that each sampled trace
is a Markov chain on the state space $\setInteger^2
\times \setInteger^2 \times \setReal$.
We use a linear twist function of the form
\begin{align}
&\gamma_k((M_{A,1:k}, M_{B,1:k}), (P_{A,1:k}, P_{B,1:k}), (t_{A,1:k}, t_{B,1:k}))
\\
&\quad =
\begin{aligned}[t]
&\frac{1}{\hat{M}}\left(\sum_{i=1}^{k-1} |M_{B,i} - M_{A,i}| \cdot (t_{B,i} - t_{A,i}) + c_1(M_{B,k} - M_{A,k})\right)
\\
&+ \frac{1}{\hat{P}}\left(\sum_{i=1}^{k-1} |P_{B,i} - P_{A,i}| \cdot (t_{B,i} - t_{A,i}) + c_1(P_{B,k} - P_{A,k})\right)
\\
&+ c_2(t_{B,k} - t_{A,k}).
\end{aligned}
\end{align}
The idea is that $\gamma_k((M_{A,1:k}, M_{B,1:k}), (P_{A,1:k}, P_{B,1:k}), (t_{A,1:k}, t_{B,1:k}))$
serves as a first-order approximation of the expected final perturbation to the mean mRNA and protein counts,
conditional upon the partial trace sampled thus far.
Traces with larger expected perturbations are thus more likely to be resampled, which reduces
the variance of the final estimate.
Note that although the twist function is based on a heuristic,
the twist function is only used to improve SMC \emph{performance}.
The overall inference algorithm and thus the gradient estimator remain unbiased.

\clearpage
\section{Deferred Material for \texorpdfstring{\hyperref[sec:preliminaries]{\S\ref*{sec:preliminaries}}}{\S\ref{sec:preliminaries}}: Full Language Definition, Syntactic Sugar, and \texorpdfstring{$\catQBS$}{QBS} Details}
\label{appendix:deferred-preliminaries}

In this appendix, we provide formal definitions and additional category-theoretic detail
for the objects introduced in \cref{sec:preliminaries}.

\begin{definition}[Full Definition of $\langP$]
    We define $\langP$ as the language generated by the grammar rules in
    \cref{def:core-language}.
    The full typing rules for $\langP$ are given as follows:
    \begin{gather}
        \Gamma \vdash \codeTrue : \typeBool
        \quad
        \Gamma \vdash \codeFalse : \typeBool
        \quad
        \Gamma \vdash n : \typeInteger
        \quad
        \Gamma \vdash r : \typeReal
        \quad
        \Gamma \vdash \codeUnit : \typeUnit
        \quad
        \Gamma \vdash p : \tau_p
        \\
        \begin{array}{c}
            x:\tau \in \Gamma \\
            \hline
            \Gamma \vdash x: \type
        \end{array}
        \quad
        \begin{array}{c}
            \extend{\Gamma}{x}{\type} \vdash \ter : \type'
            \\
            \hline
            \Gamma \vdash \lambda (x : \type). \ter : \type \typeFunction \type'
        \end{array}
        \quad
        \begin{array}{c}
            \Gamma \vdash t_1: \type_1
            \quad \Gamma \vdash t_2 : \type_2
            \\
            \hline
            \Gamma \vdash (t_1, t_2) : \type_1 \typeProduct \type_2
        \end{array}
        \\
        \begin{array}{c}
            \Gamma \vdash t': \type \typeFunction \type'
            \quad \Gamma \vdash t: \type
            \\
            \hline
            \Gamma \vdash t'~t : \type'
        \end{array}
        \quad
        \begin{array}{c}
            \Gamma \vdash t: \type
            \quad \extend{\Gamma}{x}{\type} \vdash t' : \type'
            \\
            \hline
            \Gamma \vdash \codeLet~x = t;~t' : \type'
        \end{array}
        \\
        \begin{array}{c}
            \Gamma \vdash t: \type_1 \typeProduct \type_2
            \\
            \hline
            \Gamma \vdash \codeFst~t: \type_1
        \end{array}
        \quad
        \begin{array}{c}
            \Gamma \vdash t: \type_1 \typeProduct \type_2
            \\
            \hline
            \Gamma \vdash \codeSnd~t: \type_2
        \end{array}
        \quad
        \begin{array}{c}
            \Gamma \vdash t: \typeBool
            \quad \Gamma \vdash t_1 : \type
            \quad \Gamma \vdash t_2: \type
            \\
            \hline
            \Gamma \vdash \codeIfThenElse{t}{t_1}{t_2} : \type
        \end{array}
        \\
        \begin{array}{c}
            \Gamma \vdash t: \type
            \\
            \hline
            \Gamma \vdash \codePReturn~\ter : \typeProb~\type
        \end{array}
        \quad
        \begin{array}{c}
            \Gamma \vdash t: \typeProb~\type
            \quad \extend{\Gamma}{x}{\type} \vdash t': \typeProb~\type'
            \\
            \hline
            \Gamma \vdash x \codePGets t;~t' : \typeProb~\type'
        \end{array}.
    \end{gather}
\end{definition}
\begin{remark}[Review of $\catQBS$]
    A quasi-Borel space is a pair $(U, M_U)$,
    consisting of a carrier set $U$ and a set of functions $M_U
    \subset [\setReal \to U]$
    representing all possible random variables $\alpha \in [\setReal \to U]$
    that can be defined on $U$.
    The set $M_U$ must satisfy the following three conditions:
    \begin{itemize}
        \item $\alpha \in M_U$ for all constant functions $\alpha : \setReal \to U$;
        \item $\alpha \circ f \in M_U$ for all $\alpha \in M_U$ and measurable
              $f : \setReal \to \setReal$; and
        \item $\beta \in M_U$ for all countable partitions $\setReal = S_1 \sqcup S_2 \sqcup \dots$
              into Borel sets and functions $\alpha_1, \alpha_2, \dots \in M_U$,
              where $\beta(u) = \alpha_i(u)$ for $u \in S_i$.
    \end{itemize}
    A morphism $f : (U, M_U) \to (V, M_V)$
    of $\catQBS$ is a set-theoretic function $f \in  [U \to V]$, which in the following
    sense respects the sets of random variables:
    \begin{equation}
        \alpha \in M_U \implies f \circ \alpha \in M_V.
    \end{equation}
    When there is no ambiguity,
    we refer to a quasi-Borel space $(U, M_U)$ by
    just its carrier set $U$.
    Additionally, since the morphisms of $\catQBS$ are
    functions, for $U : \catQBS$
    we identify a \emph{value} $u : U$
    (i.e., a morphism $u : \semUnit \to U$)
    with an element of the set $U$.
    The category $\catQBS$ is Cartesian-closed: for $U, V : \catQBS$,
    we denote their product by $U \times V : \catQBS$, and the exponential
    of $V$ by $U$ by $U \to V : \catQBS$.
    We direct the reader to \citet{heunen2017convenient} for the full treatment of $\catQBS$.
\end{remark}
\begin{remark}[Measure Monad on $\catQBS$]
    \label{remark:measure-monad-on-qbs}
    \citet{scibior2018denotational} define a monad of measures on $\catQBS$, denoted
    by $P$ in the main text.
    Our use of the measure monad, rather than the probability monad,
    allows more flexible coupling primitives that, e.g., make use of $\codeCategorical$ primitives
    with unnormalized weights (see \cref{appendix:gradient-inference}).
    In some cases, these weights may even be signed, which can be justified by augmenting the
    monad of \citet{scibior2018denotational} with a sign bit $b \in \set{-1, 1}$ via the writer monad transformer.
    Here these sign bits are accumulated by multiplication, and the integration of a function $f$ against a signed
    measure $\mu: P(\setReal \times \{-1, 1\})$ is defined as $\int (f(x) \cdot b) \mu(\diff{(x, b)})$.
\end{remark}
\begin{definition}[Quasi-Borel Spaces for Reals, Booleans, Integers, and Unit]
    We define quasi-Borel spaces for the carrier sets
    $\setReal$, $\setBool = \{\semTrue, \semFalse\}$,
    $\setInteger$, and $\setUnit = \{\semUnit\}$.
    For $S \in \set{\setReal, \setBool, \setInteger, \setUnit}$,
    we define the random variables $M_S$
    as the measurable functions from $\setReal$ to $S$,
    where $\setReal$ is endowed with the Borel $\sigma$-algebra
    $\mathcal{B}(\setReal)$,
    and countable sets are endowed with the power set $\sigma$-algebra.
    Precisely,
    we define
    \begin{align}
        M_{\setReal} & \defas \set{ \alpha \in [\setReal \to \setReal] \mid \forall A \in \mathcal{B}(\setReal).~\alpha^{-1}(A) \in \mathcal{B}(\setReal)} \\
        M_{\setBool} & \defas \set{ \alpha \in [\setReal \to \setBool] \mid \forall A \subseteq \setBool.~\alpha^{-1}(A) \in \mathcal{B}(\setReal)} \\
        M_{\setInteger} & \defas \set{ \alpha \in [\setReal \to \setInteger] \mid \forall A \subseteq \setInteger.~\alpha^{-1}(A) \in \mathcal{B}(\setReal)} \\
        M_{\setUnit} & \defas \set{ \_ \mapsto \semUnit }.
    \end{align}
\end{definition}
\begin{definition}[Quasi-Borel Spaces on Subsets]
    \label{def:subspace}
    We allow use of set-builder notation
    to construct quasi-Borel spaces on subsets,
    e.g. $\{u : \setReal \mid u \leq 3\}$
    constructs a quasi-Borel space on
    $(-\infty, 3] \subset \setReal$
    whose random variables are induced
    by those of $\setReal$.
    Formally, for $U : \catQBS$
    and a subset $U' \subset U$,
    we define the $\catQBS$ space
    induced on $U'$ by
    \begin{equation}
        M_{U'} \defas
        \set{
           \alpha \in M_U \mid
           \alpha(\setReal) \subseteq U'
        }.
    \end{equation}
\end{definition}
The following elementary lemma will prove useful to our logical relation proofs.
\begin{lemma}[$\catQBS$-Morphisms into Subsets]
    \label{lemma:morphism-into-subspace}
    Suppose $f \in [U \to V]$ is a $\catQBS$-morphism
    $f : U \to V$, and suppose $V' \subset V$
    with $M_{V'}$ defined as in \cref{def:subspace}.
    Then, if $f(U) \subseteq V'$, it holds that $f$ is also a valid
    $\catQBS$-morphism $f : U \to V'$.
\end{lemma}
\begin{proof}
   Take any $\alpha \in M_U$.
   Since $f: U \to V$ is a $\catQBS$-morphism,
    we have that $f \circ \alpha \in M_V$.
   By assumption, we also have that $f(U) \subseteq V' \implies (f \circ \alpha)(\setReal) \subseteq V'$.
   Thus $f \circ \alpha \in M_{V'}$.
   We conclude that $f : U \to V'$ is a $\catQBS$-morphism.
\end{proof}
To speak precisely about the denotational models and program transformations
defined in the main text,
we use \emph{syntactic categories}.
\begin{definition}[Syntactic Category]
    For a language $\lang_0$, we let $\catSyntax(\lang_0)$
    denote its syntactic category.
    The objects of $\catSyntax(\lang_0)$ are the types $\tau$ of $\lang_0$.
    A morphism of $\catSyntax(\lang_0)$ is an equivalence class of well-typed terms
    $\Gamma \vdash t : \tau$ of $\lang_0$ under $\alpha$-equivalence, i.e., variable renaming.
    The target object of such a morphism is $\tau$.
    The source object of such a morphism, for $\Gamma = x_1 : \tau_1, \dots, x_n : \tau_n$,
    is $\tau_1 \typeProduct \dots \typeProduct \tau_n$, where the types are assumed to be ordered
    using de Bruijn indices~\citepapp{debruijn1972lambda}.
    For example, a well-typed term $x_1 : \setBool, x_2 : \setReal \vdash (\codeIfThenElse{x_1}{x_2}{-x_2}) : \setReal$
    corresponds to a morphism from $\setBool \typeProduct \setReal$ to $\setReal$ in the syntactic category.
    Composition of two morphisms in the syntactic category is given by substituting one well-typed term into
    the context of the other.
\end{definition}
Using the syntactic category,
the denotational model $\sem{\cdot}$ of $\langP$
is a functor $\sem{\cdot}: \catSyntax(\langP) \longrightarrow \catQBS$,
and a program transformation $\transform{\cdot}$
from one language $\lang_0$ to another language $\lang_1$
is a functor $\transform{\cdot}: \catSyntax(\lang_0) \longrightarrow \catSyntax(\lang_1)$.

\begin{remark}[Syntactic Definitions of Models and Transformations]
When we use concrete term syntax to define denotational models and program transformations in the main text,
this notation should be understood as formally defining a functor in the syntactic category.
For example, the syntax $\transform{c} = ((c, c), (c, c))$ defines $\transform{\cdot}$
on each well-typed term $\Gamma \vdash c : \sigma$ (up to $\alpha$-equivalence)
in the source language, where $\transform{\Gamma \vdash c : \sigma}$
is a morphism $\transform{\Gamma} \vdash ((c, c), (c, c)) : \transform{\sigma}$
in the target language (again taken up to $\alpha$-equivalence).
For this definition to ``type check'', it must be that
$\transform{\sigma} = (\sigma \typeProduct \sigma) \typeProduct (\sigma \typeProduct \sigma)$.
In such definitions, the introduction of fresh variables should be assumed to be hygienic, i.e., not in conflict
with variables in the subterms (since morphisms are taken up
to $\alpha$-equivalence, this problem disappears in the syntactic category).
\end{remark}
\begin{definition}[Full Denotational Model of $\langP$]
    \label{def:model-langP}
    The denotational model
    \begin{equation}
        \semP{\cdot} : \catSyntax(\langP) \longrightarrow \catQBS
    \end{equation}
    is given as follows on types:
    \begin{gather}
        \semP{\typeBool} \defas \setBool
        \quad
        \semP{\typeInteger} \defas \setInteger
        \quad
        \semP{\typeReal} \defas \setReal
        \quad
        \semP{\typeUnit} \defas \setUnit
        \\
        \semP{\type_1 \typeProduct \type_2} \defas \semP{\type_1} \times \semP{\type_2}
        \quad
        \semP{\type_1 \typeFunction \type_2} \defas \semP{\type_1} \to \semP{\type_2}
        \quad
        \semP{\typeProb~\type} \defas P\left(\semP{\type}\right),
    \end{gather}
    and as follows on terms:
    \begin{align}
         & \semP{x}(\gamma) \defas \gamma(x)
         &
         & \semP{r}(\gamma) \defas r
        \\
         & \semP{n}(\gamma) \defas n
         &
         & \semP{\codeUnit}(\gamma) \defas \semUnit
        \\
         & \semP{\codeTrue}(\gamma) \defas \semTrue
         &
         & \semP{\codeFalse}(\gamma) \defas \semFalse
        \\
         & \semP{(t_1, t_2)}(\gamma) \defas (\semP{t_1}(\gamma), \semP{t_2}(\gamma))
         &
         & \semP{\lambda(x:\tau).t}(\gamma) \defas
        u \mapsto \semP{t}(\gamma[x \setvar u])
        \\
         & \semP{t'~t}(\gamma) \defas \semP{t'}(\gamma)\left(\semP{t}(\gamma)\right)
         &
         & \semP{\codeLet~x = t;~t'}(\gamma) \defas
        \semP{t'}\left(\gamma[x \setvar \semP{t}(\gamma)]\right)
        \\
         & \semP{\codeFst~t}(\gamma) \defas \semFst(\semP{t}(\gamma))
         &
         & \semP{\codeSnd~t}(\gamma) \defas \semSnd(\semP{t}(\gamma))
        \\
         & \mathrlap{\semP{\codeIfThenElse{t}{t_1}{t_2}}(\gamma) \defas
            \semIfThenElse{\semP{t}(\gamma)}{\semP{t_1}(\gamma)}
            {\semP{t_2}(\gamma)}}
        \\
         & \mathrlap{\semP{\codePReturn~t}(\gamma) \defas
            \semPReturn\left(\semP{t}(\gamma)\right)}
        \\
         & \mathrlap{\semP{x \codePGets t;~t'}(\gamma) \defas
            \semP{t}(\gamma) \semPBind{}
            \left(u \mapsto \semP{t'}(\gamma[x \setvar u])\right)}
        \\
         & \semP{p}(\gamma) \defas \semP{p}.
    \end{align}
\end{definition}

\begin{definition}[Explicit Definition of Homomorphic Action]
    \label{def:homomorphic-action}
    The homomorphic action of a transformation $\transform{\cdot}$
    on each type and term of $\langP$ is given explicitly as follows:
    \begin{align}
         & \transform{\type_1 \typeProduct \type_2} \defas \transform{\type_1} \typeProduct \transform{\type_2}
         &
         & \transform{\type_1 \typeFunction \type_2} \defas \transform{\type_1} \typeFunction \transform{\type_2}
        \\
         & \transform{\typeProb~\type} \defas \typeProb~\transform{\type}
        \\
         & \transform{x} \defas x
         &
         & \transform{c} \defas c
        \\
         & \transform{(t_1, t_2)} \defas \left(\transform{t_1}, \transform{t_2}\right)
         &
         & \transform{\lambda(x:\tau).t} \defas
        \lambda (x: \tau). \transform{t}
        \\
         & \transform{t'~t} \defas \transform{t'}~\transform{t}
         &
         & \transform{\codeLet~x = t;~t'} \defas
        \codeLet~x = \transform{t};~\transform{t'}
        \\
         & \transform{\codeFst~t} \defas \codeFst~\transform{t}
         &
         & \transform{\codeSnd~t} \defas \codeSnd~\transform{t}
        \\
         & \mathrlap{\transform{\codeIfThenElse{t}{t_1}{t_2}} \defas
            \codeIfThenElse{\transform{t}}{\transform{t_1}}{\transform{t_2}}
        }
        \\
         & \mathrlap{\transform{\codePReturn~t} \defas
            \codePReturn~\transform{t}}
        \\
         & \mathrlap{\transform{x \codePGets t;~t'} \defas
            x \codePGets \transform{t};~\transform{t'}.}
    \end{align}
\end{definition}
\begin{definition}[Arithmetic primitives]
    We register standard arithmetic primitives for $\langP$. For example,
    we make the extension
    \begin{equation}
        p \Coloneq\, \cdots \mid +_{\setReal} \mid +_{\setInteger} \mid -_{\setReal} \mid -_{\setInteger}
        \mid \times_{\setReal} \mid \times_{\setInteger} \mid /_{\setReal}
        \mid \mathrm{pow}_{\setReal,\setReal} \mid \mathrm{pow}_{\setInteger,\setInteger}
    \end{equation}
    where these primitives have types
    \begin{align}
         \tau_{+_{\setReal}} & \defas (\typeReal \times \typeReal) \typeFunction \typeReal
        &
         \tau_{+_{\setInteger}} & \defas (\typeInteger \times \typeInteger) \typeFunction \typeInteger
        \\
         \tau_{-_{\setReal}} & \defas (\typeReal \times \typeReal) \typeFunction \typeReal
        &
         \tau_{-_{\setInteger}} & \defas (\typeInteger \times \typeInteger) \typeFunction \typeInteger
        \\
         \tau_{\times_{\setReal}} & \defas (\typeReal \times \typeReal) \typeFunction \typeReal
        &
         \tau_{\times_{\setInteger}} & \defas (\typeInteger \times \typeInteger) \typeFunction \typeInteger
        \\
         \tau_{/_{\setReal}} & \defas (\typeReal \times \typeReal) \typeFunction \typeReal
        & \tau_{\mathrm{pow}_{\setReal,\setReal}} & \defas (\typeReal \times \typeReal) \typeFunction \typeReal
        \\
        \tau_{\mathrm{pow}_{\setInteger,\setInteger}} & \defas (\typeInteger \times \typeInteger) \typeFunction \typeInteger.
    \end{align}
    As explained in \cref{notation:syntactic-sugar}, we permit standard infix use of such operations.
\end{definition}
\begin{notation}[Syntactic Sugar]
    \label{notation:syntactic-sugar}
    The following list describes syntactic sugar used in $\langP$ and $\langPP$ programs.
    \begin{itemize}[wide,leftmargin=*,topsep=5pt,itemsep=5pt]
        \item \textbf{Infix use of mathematical operations.} As explained in
              the main text, we permit infix use of mathematical notation,
              with the understanding that this desugars into a formal grammar in the expected
              way.
              For example, $3.0 + 4.5$ desugars to $+_{\setReal}~(3.0, 4.5)$,
              and $3^4$ desugars to $\mathrm{pow}_{\setInteger,\setInteger}~(3,4)$.
        \item \textbf{Omitting subscripts on families of primitives.} While $\langP$
              lacks polymorphism, we will often define a family of type-indexed primitives,
              such as the primitive $+_\alpha$ for addition, where $\alpha \in \set{\setReal, \setInteger}$.
              When the type is clear from context, we may drop the subscript.
        \item \textbf{Code blocks.} When chaining bind expressions together into
              blocks of code,
              we either use newlines and indentation to
              make the grouping of terms clear, or
              curly braces
              $\set{\dots}$ in the case of inline code.
              Let binds and/or monadic binds
              may be unambiguously sequenced together using a set of
              curly braces: e.g. $\{\codeLet~x = 3;~y \codePGets \codeFlip~0.5;~\codePReturn~(x + y)\}$.
        \item \textbf{Unpacking and $n$-tuples.}
              We allow for $n$-tuples, whose contents may be extracted
              via tuple unpacking, e.g. we may write
              $\lambda (z : \typeReal \typeProduct \typeReal \typeProduct \typeReal). \{\codeLet~(a, b, c) = z;~a + b + c\}$
              or even just $\lambda (a: \typeReal, b: \typeReal, c : \typeReal).~a + b + c$.
              When the type $\tau$ of an $n$-tuple is uniform, we may abbreviate to $\tau^n$,
              so the above can also be written as
              $\lambda (z : \typeReal^3). \{\codeLet~(a, b, c) = z;~a + b + c\}$.
        \item \textbf{Unrolled operations on $n$-tuples.}
              On these $n$-tuples, we allow use of standard Haskell-inspired
              functional constructs, such as $\codeMap_n$,
              $\codeZip_n$, $\codeScanl_n$, and $\codeSum_n$, as well as indexing via $\codeIndex_n$:
              all of these constructs can be defined as \emph{sugar} in $\langP$ by unrolling the operation
              over the input $n$-tuple.
              We also define ${\codeSequenceP}_{\alpha,n}$ and ${\codeSequenceR}_{\alpha,n}$ for lifting a type constructor
              out of an $n$-tuple, where
              $\Gamma \vdash {\codeSequenceP}_{\alpha,n} : (\typeProb~\alpha)^n \typeFunction \typeProb~(\alpha^n)$
              and $\Gamma \vdash {\codeSequenceR}_{\alpha,n} : (\typeResidual~\alpha)^n \typeFunction \typeResidual~(\alpha^n)$.
              These sequencing primitives are implemented as sugar for an unwrapped fold
              over an $n$-tuple, e.g., for ${\codeSequenceP}_{\alpha,n}$ the fold function would be
              $\lambda (\mu, \nu). \set{x \codePGets \mu;~y \codePGets \nu;~\codePReturn~(x,y)}$.
              For iteration, we define ${\codePIterate}_{n,\alpha}$ as sugar for repeatedly applying a measure kernel $n$ times,
              where $\Gamma \vdash {\codePIterate}_{n,\alpha} : (\alpha \typeFunction \typeProb~\alpha) \typeFunction \alpha \typeFunction \typeProb~\alpha$
              and $\semP{{\codePIterate}_{n,\alpha}}(f, x) = (\mu \mapsto \mu \semPBind{} f)^n(\semPReturn(x))$.
              Finally, we include $\codeNormalize_n$ as sugar for the operation that normalizes the values in a real-valued
              $n$-tuple so that they sum to 1.
        \item \textbf{Omitting type annotation on $\lambda$-abstractions}.
              When the type is clear from context, we may omit the type annotation
              on the argument of a $\lambda$-abstraction.
              For example, we may define a term $\emptyEnv \vdash t : \typeReal \typeFunction \typeReal$
              by $t \defas \lambda(x).~x + x$.
    \end{itemize}
\end{notation}

\clearpage
\section{Deferred Material for \texorpdfstring{\hyperref[sec:coupling]{\S\ref*{sec:coupling}}}{\S\ref{sec:coupling}}: Proof of Soundness of \texorpdfstring{$\transformCoupling{\cdot}$}{C\{·\}}}

In this appendix, we provide
full definitions and additional categorical theoretical detail for
the objects introduced in \cref{sec:coupling},
and supply the proof of \cref{thm:coupling-correctness}.

\label{appendix:deferred-coupling}

\begin{definition}[Coupling Transformation]
    \label{def:coupling-transformation-full}
    The coupling program transformation is a functor
    \begin{equation}
        \transformCoupling{\cdot} : \catSyntax(\langP) \longrightarrow \catSyntax(\langP),
    \end{equation}
    whose action on types and terms is given in \cref{def:coupling-transformation}.
\end{definition}

\begin{definition}[Coupling Logical Relation]
    \label{def:logical-relation-coupling}
    The proof-relevant coupling logical relation $\relationCouplingProof$
    associates each type $\tau$ with a quasi-Borel space $\relationCouplingProof(\tau)$
    of $\langP$,
    where
    \begin{equation}
        \relationCouplingProof(\tau) \subset
        \semP{\tau} \times \semP{\tau} \times \semP{\transformCoupling{\tau}} \times \proofCoupling(\tau).
    \end{equation}
    In words, $\relationCouplingProof(\tau)$ relates tuples
    $(x_A, x_B, x_{\transformCoupling}, \tilde{x})$,
    where $x_A : \semP{\tau}$,
    $x_B : \semP{\tau}$,
    $x_{\transformCoupling} : \semP{\transformCoupling{\tau}}$,
    and $\tilde{x} : \proofCoupling(\tau)$.
    Here, the fourth component $\tilde{x}$ is a \emph{proof value} belonging to the space $\proofCoupling(\tau)$,
    where $\proofCoupling$ itself associates types of $\langP$ with quasi-Borel spaces.
    %
\begin{listing}[t]
\small

\setlength{\abovedisplayshortskip}{0pt}
\setlength{\belowdisplayshortskip}{0pt}
\setlength{\abovedisplayskip}{0pt}
\setlength{\belowdisplayskip}{0pt}

\begin{syntaxbox}[title=\textbf{Proof Space $\proofCoupling$ on Types of $\langP$}]{\linewidth}
 \begin{align*}
        \proofCoupling(\sigma)  &\defas \setUnit
        &\qquad
        \proofCoupling(\tau_1 \typeProduct \tau_2) &\defas
            \proofCoupling(\tau_1) \times \proofCoupling(\tau_2)
        \\
        \proofCoupling(\tau_1 \typeFunction \tau_2) &\defas
            \relationCouplingProof(\tau_1) \to \proofCoupling(\tau_2)
        &\qquad
        \proofCoupling(\typeProb~\type) &\defas P\left(\relationCouplingProof(\type)\right)
 \end{align*}
\end{syntaxbox}

\begin{syntaxbox}[title=\textbf{Proof-Relevant Coupling Logical Relation $\relationCouplingProof$ on Types of $\langP$}]{\linewidth}
 \begin{align*}
        \relationCouplingProof(\sigma)                  & \defas
        \{(x_A, x_B, x_{\transformCoupling}, \semUnit)
        \mid x_{\transformCoupling} = (x_A, x_B)\}
        \\
        \relationCouplingProof(\tau_1 \typeProduct \tau_2) &\defas
        \Big\{(x_A, x_B, x_{\transformCoupling}, \tilde{x}) \mathrel{\big|}
        \begin{aligned}[t]
            &(\semFst(x_A), \semFst(x_B), \semFst(x_{\transformCoupling}), \semFst(\tilde{x})) \in \relationCouplingProof(\tau_1)
            \\
            &\hspace{-1em}\land (\semSnd(x_A), \semSnd(x_B), \semSnd(x_{\transformCoupling}), \semSnd(\tilde{x})) \in \relationCouplingProof(\tau_2)\Big\}
        \end{aligned}
        \\
        \relationCouplingProof(\tau_1 \typeFunction \tau_2) &\defas
        \Big\{(f_A, f_B, f_{\transformCoupling}, \tilde{f}) \mathrel{\big|}
        \begin{aligned}[t]
            &\forall (x_A, x_B, x_{\transformCoupling}, \tilde{x}) \in \relationCouplingProof(\tau_1).
            \\
            &\left(f_A(x_A), f_B(x_B), f_{\transformCoupling}(x_{\transformCoupling}), \tilde{f}(x_A, x_B, x_{\transformCoupling}, \tilde{x})\right)
            \in \relationCouplingProof(\tau_2)\Big\}
        \end{aligned}
        \\
        \relationCouplingProof(\typeProb~\tau) &\defas
        \Big\{\left(\mu_A, \mu_B, \mu_{\transformCoupling}, \tilde{\mu}\right) \mathrel{\big|}
        \begin{aligned}[t]
        &\smallint \tilde{\mu}(\diff x) = 1 \\
        &\hspace{-2em}\land
        \left(P(\semProj_1)(\tilde{\mu}), P(\semProj_2)(\tilde{\mu}), P(\semProj_3)(\tilde{\mu})\right)
        = (\mu_A, \mu_B, \mu_{\transformCoupling})\Big\}
        \end{aligned}
 \end{align*}
\end{syntaxbox}

\captionsetup{aboveskip=5pt,belowskip=-10pt}
\caption{The proof-relevant coupling logical relation on $\langP$ types.
}
\label{def:coupling-logical-relation}
\end{listing}

    \cref{def:coupling-logical-relation} gives the definition of the proof-relevant coupling logical relation
    on the types of $\langP$.
    We define the proof-irrelevant relation by
    \begin{equation}
        \relationCoupling(\tau) \defas
        \set{(x_A, x_B, x_{\transformCoupling}) \mid \exists \tilde{x}.
            (x_A, x_B, x_{\transformCoupling}, \tilde{x}) \in \relationCouplingProof(\tau)}.
    \end{equation}
\end{definition}

\thmCouplingCorrectness*
\begin{proof}
    We extend $\proofCoupling$ and $\relationCouplingProof$ into functors
    \begin{align}
        \proofCoupling         & : \catSyntax(\langP) \longrightarrow \catQBS
        \\
        \relationCouplingProof & : \catSyntax(\langP) \longrightarrow \catQBS
    \end{align}
    where we define $\relationCouplingProof$ on a term $\Gamma \vdash t : \tau$ by
    \begin{equation}
        \relationCouplingProof(t)(\gamma_A, \gamma_B, \gamma_{\transformCoupling}, \tilde{\gamma})
        \defas
        \left(\semP{t}(\gamma_A),
        \semP{t}(\gamma_B),
        \semP{\transformCoupling{t}}(\gamma_{\transformCoupling}),
        \proofCoupling(t)(\gamma_A, \gamma_B, \gamma_{\transformCoupling}, \tilde{\gamma})\right).
        \label{eq:proof-coupling-relation-defn}
    \end{equation}
    It remains to define the morphism $\proofCoupling(t) : \proofCoupling(\Gamma) \to \proofCoupling(\tau)$
    that constructs the proof value
    for each term $\Gamma \vdash t : \tau$,
    in such a way
    that $\relationCouplingProof(t)$ preserves the logical relation, i.e. $\relationCouplingProof(t)$
    is a morphism from $\relationCouplingProof(\Gamma)$ to $\relationCouplingProof(\tau)$.
    As is typical for a logical relation proof, we proceed by induction
    on the typing derivation of $\Gamma \vdash t : \tau$:
    in each case of the induction, we will first define
    $\proofCoupling(t)$,
    and then verify that $\relationCouplingProof(t)$
    preserves the logical relation.

    In the following, let $(\gamma_A, \gamma_B, \gamma_{\transformCoupling}, \tilde{\gamma})$
    denote an arbitrary tuple in $\relationCouplingProof(\Gamma)$.
    Note that $\relationCouplingProof(t)$ is a morphism from $\relationCouplingProof(\Gamma)$ to
    $\semP{\tau} \times \semP{\tau} \times \semP{\transformCoupling{\tau}} \times \proofCoupling(\tau)$
    by its definition in \cref{eq:proof-coupling-relation-defn}: by \cref{lemma:morphism-into-subspace},
    showing that $\relationCouplingProof(t)$ is a morphism from $\relationCouplingProof(\Gamma)$ to $\relationCouplingProof(\tau)$
    is equivalent to showing that $\relationCouplingProof(t)(\gamma_A, \gamma_B, \gamma_{\transformCoupling}, \tilde{\gamma}) \in \relationCouplingProof(\tau)$
    for each $(\gamma_A, \gamma_B, \gamma_{\transformCoupling}, \tilde{\gamma}) \in \relationCouplingProof(\Gamma)$.

    \begin{itemize}[leftmargin=*,topsep=10pt,itemsep=10pt]
        \item \textbf{Cases } \(
              \boxed{
                  \Gamma \vdash \codeTrue : \typeBool
              }
              \hspace{0.5em}
              \boxed{
                  \Gamma \vdash \codeFalse : \typeBool
              }
              \hspace{0.5em}
              \boxed{
                  \Gamma \vdash n : \typeInteger
              }
              \hspace{0.5em}
              \boxed{
                  \Gamma \vdash r : \typeReal
              }
              \hspace{0.5em}
              \boxed{
                  \Gamma \vdash \codeUnit : \typeUnit
              }
              \)
              \textbf{: }
              \par\vskip0.5em\relax

              For a constant $c$, we define
              \begin{equation}
                  \proofCoupling(c)(\gamma_A, \gamma_B, \gamma_{\transformCoupling}, \tilde{\gamma}) \defas \semUnit,
              \end{equation}
              i.e., $\proofCoupling(c)$ is the constant function to $\semUnit$.
              Since $\transformCoupling{c} = (c, c)$ for constants $c$
              and constants are of base type $\sigma$,
              it follows that
              \begin{equation}
                  \relationCouplingProof(c)(\gamma_A, \gamma_B, \gamma_{\transformCoupling}, \tilde{\gamma}) =
                  \left(c, c, (c, c), \semUnit\right)
                  \in \relationCouplingProof(\sigma).
              \end{equation}
        \item \textbf{Case } \(
              \boxed{
                  \Gamma \vdash p : \tau_p
              }
              \)
              \textbf{: }
              \par\vskip0.5em\relax

              By assumption, we have
              \begin{equation}
                  \left(\semP{p}(\gamma_A), \semP{p}(\gamma_B),
                  \semP{\transformCoupling{p}}(\gamma_{\transformCoupling})\right) =
                  \left(\semP{p}, \semP{p}, \semP{\transformCoupling{p}}\right)
                  \in \relationCoupling(\tau_p).
              \end{equation}
              By the definition of $\relationCoupling$, there exists an environment-independent proof value $\tilde{x}$ such that
              \begin{equation}
                  \left(\semP{p}(\gamma_A), \semP{p}(\gamma_B),
                  \semP{\transformCoupling{p}}(\gamma_{\transformCoupling}), \tilde{x}\right) \in
                  \relationCouplingProof(\tau_p),
              \end{equation}
              and hence we can ensure $\relationCouplingProof(p)$ preserves the logical relation by defining
              \begin{equation}
                  \proofCoupling(p)(\gamma_A, \gamma_B, \gamma_{\transformCoupling}, \tilde{\gamma}) \defas \tilde{x}.
              \end{equation}
              Note that $\proofCoupling(p)$ is not responsible for \emph{choosing} $\tilde{x}$ as a function
              of its input.
              Rather, $\proofCoupling(p)$ is a constant function for a fixed $\tilde{x}$ chosen in the meta-language,
              and thus a $\catQBS$-morphism.
        \item \textbf{Cases } \(
              \boxed{
                  \begin{array}{c}
                      \extend{\Gamma}{x}{\type} \vdash \ter : \type'
                      \\
                      \hline
                      \Gamma \vdash \lambda (x : \type). \ter : \type \typeFunction \type'
                  \end{array}
              }
              \quad
              \boxed{
                  \begin{array}{c}
                      \Gamma \vdash t': \type \typeFunction \type'
                      \quad \Gamma \vdash t: \type
                      \\
                      \hline
                      \Gamma \vdash t'~t : \type'
                  \end{array}
              }
              \)
              \textbf{: }
              \par\vskip0.5em\relax

              For $\lambda$-abstraction,
              we define
              \begin{align}
                  \proofCoupling(\lambda (x: \tau). t)
                  (\gamma_A, \gamma_B, \gamma_{\transformCoupling}, \tilde{\gamma})
                   & \defas
                  (u_A, u_B, u_{\transformCoupling}, \tilde{u}) \mapsto
                  \proofCoupling(t)
                  \begin{aligned}[t]
                      \big(
                       & \gamma_A[x \mapsto u_A], \gamma_B[x \mapsto u_B],
                      \\
                       & \gamma_{\transformCoupling}[x \mapsto u_{\transformCoupling}],
                      \tilde{\gamma}[x \mapsto \tilde{u}]\big).
                  \end{aligned}
              \end{align}
              For application, we define
              \begin{align}
                  \proofCoupling(t'~t)(\gamma_A, \gamma_B, \gamma_{\transformCoupling}, \tilde{\gamma})
                  \defas
                  \proofCoupling(t')(\gamma_A, \gamma_B, \gamma_{\transformCoupling}, \tilde{\gamma})\left(\relationCouplingProof(t)(\gamma_A, \gamma_B, \gamma_{\transformCoupling}, \tilde{\gamma})\right)
              \end{align}
              It follows from the induction hypotheses and the definition of $\relationCouplingProof(\tau \typeFunction \tau')$
              in \cref{def:coupling-logical-relation}
              that $\relationCouplingProof(\lambda (x: \tau). t)(\gamma_A, \gamma_B, \gamma_{\transformCoupling}, \tilde{\gamma})
                  \in \relationCouplingProof(\tau \typeFunction \tau')$
              and $\relationCouplingProof(t'~t)(\gamma_A, \gamma_B, \gamma_{\transformCoupling}, \tilde{\gamma})
                  \in \relationCouplingProof(\tau')$.
        \item \textbf{Cases } \(
              \boxed{
                  \begin{array}{c}
                      \Gamma \vdash t_1 : \type_1
                      \quad
                      \Gamma \vdash t_2 : \type_2
                      \\
                      \hline
                      \Gamma \vdash (t_1, t_2) : \type_1 \typeProduct \type_2
                  \end{array}
              }
              \quad
              \boxed{
                  \begin{array}{c}
                      \Gamma \vdash t: \type_1 \typeProduct \type_2
                      \\
                      \hline
                      \Gamma \vdash \codeFst~t: \type_1
                  \end{array}
              }
              \quad
              \boxed{
                  \begin{array}{c}
                      \Gamma \vdash t: \type_1 \typeProduct \type_2
                      \\
                      \hline
                      \Gamma \vdash \codeSnd~t: \type_2
                  \end{array}
              }
              \)
              \textbf{: }

              \par\vskip0.5em\relax
              For pair introduction, we define
              \begin{align}
                   \proofCoupling((t_1, t_2))(\gamma_A, \gamma_B, \gamma_{\transformCoupling}, \tilde{\gamma})
                   \defas
                  \left(\proofCoupling(t_1)(\gamma_A, \gamma_B, \gamma_{\transformCoupling}, \tilde{\gamma}),
                  \proofCoupling(t_2)(\gamma_A, \gamma_B, \gamma_{\transformCoupling}, \tilde{\gamma})\right).
              \end{align}
              For pair elimination, we define
              \begin{align}
                  \proofCoupling(\codeFst~t)
                   & \defas
                  \semFst \circ \proofCoupling(t),
                  \\
                  \proofCoupling(\codeSnd~t)
                   & \defas
                  \semSnd \circ \proofCoupling(t).
              \end{align}
              \sloppy
              It follows from the induction hypotheses and
              the definition of $\relationCouplingProof(\tau_1 \typeProduct \tau_2)$
              in \cref{def:coupling-logical-relation} that
              $\relationCouplingProof((t_1, t_2))(\gamma_A, \gamma_B, \gamma_{\transformCoupling}, \tilde{\gamma})
                  \in \relationCouplingProof(\tau_1 \typeProduct \tau_2)$,
              and that
              $\relationCouplingProof(\codeFst~t)(\gamma_A, \gamma_B, \gamma_{\transformCoupling}, \tilde{\gamma})
                  \in \relationCouplingProof(\tau_1)$
              and $\relationCouplingProof(\codeSnd~t)(\gamma_A, \gamma_B, \gamma_{\transformCoupling}, \tilde{\gamma})
                  \in \relationCouplingProof(\tau_2)$.
        \item \textbf{Case } \(
              \boxed{
                  \begin{array}{c}
                      \Gamma \vdash t: \type
                      \quad \extend{\Gamma}{x}{\type} \vdash t' : \type'
                      \\
                      \hline
                      \Gamma \vdash \codeLet~x = t;~t' : \type'
                  \end{array}
              }
              \)
              \textbf{: }
              \par\vskip0.5em\relax

              We define
              \begin{align}
                   & \proofCoupling(\codeLet~x = t;~t')(\gamma_A, \gamma_B, \gamma_{\transformCoupling}, \tilde{\gamma})
                  \\
                  \nonumber
                   & \defas
                  \proofCoupling(t')\Big(
                  \begin{aligned}[t]
                       & \gamma_A[x \mapsto \sem{t}(\gamma_A)], \gamma_B[x \mapsto \sem{t}(\gamma_B)], \gamma_{\transformCoupling}[x \mapsto \sem{\transformCoupling{t}}(\gamma_{\transformCoupling})], \\
                       & \tilde{\gamma}\left[x \mapsto \proofCoupling(t)(\gamma_A, \gamma_B, \gamma_{\transformCoupling}, \tilde{\gamma})\right]\Big)
                  \end{aligned}
              \end{align}
              It follows from the induction hypotheses and the semantics of the let bind that
              $\relationCouplingProof(\codeLet~x = t;~t')(\gamma_A, \gamma_B, \gamma_{\transformCoupling}, \tilde{\gamma})
                  \in \relationCouplingProof(\tau')$.
        \item \textbf{Case } \(
              \boxed{
                  \begin{array}{c}
                      \Gamma \vdash t: \type
                      \\
                      \hline
                      \Gamma \vdash \codePReturn~\ter : \typeProb~\type
                  \end{array}
              }
              \)
              \textbf{: }
              \par\vskip0.5em\relax

              We define
              \begin{align}
                   & \proofCoupling(\codePReturn~t)(\gamma_A, \gamma_B, \gamma_{\transformCoupling}, \tilde{\gamma})
                  \\
                  \nonumber
                   & \defas
                  \semPReturn\left(\semP{t}(\gamma_A), \semP{t}(\gamma_B), \semP{\transformCoupling{t}}(\gamma_{\transformCoupling}), \proofCoupling(t)(\gamma_A, \gamma_B, \gamma_{\transformCoupling}, \tilde{\gamma})\right),
              \end{align}
              i.e., the coupling proof value is itself a Dirac distribution over the underlying related values and their proof.
              Let $(x_A, x_B, x_{\transformCoupling}, \tilde{x}) \defas
                  \relationCouplingProof(\codePReturn~t)(\gamma_A, \gamma_B, \gamma_{\transformCoupling}, \tilde{\gamma})$.
              We have
              \begin{align}
                  P(\semProj_1)(\tilde{x}) & = \semPReturn(x_A)                     \\
                  P(\semProj_2)(\tilde{x}) & = \semPReturn(x_B)                     \\
                  P(\semProj_3)(\tilde{x}) & = \semPReturn(x_{\transformCoupling}).
              \end{align}
              It follows from the induction hypothesis and the definition of $\relationCouplingProof(\typeProb~\tau)$
              in \cref{def:coupling-logical-relation} that
              $\relationCouplingProof(\codePReturn~t)(\gamma_A, \gamma_B, \gamma_{\transformCoupling}, \tilde{\gamma})
                  \in \relationCouplingProof(\typeProb~\tau)$.
        \item \textbf{Case } \(
              \boxed{
                  \begin{array}{c}
                      \Gamma \vdash t: \typeProb~\type
                      \quad \extend{\Gamma}{x}{\type} \vdash t': \typeProb~\type'
                      \\
                      \hline
                      \Gamma \vdash x \codePGets t;~t' : \typeProb~\type'
                  \end{array}
              }
              \)
              \textbf{: }
              \par\vskip0.5em\relax

              We define
              \begin{align}
                   & \proofCoupling(x \codePGets t;~t')(\gamma_A, \gamma_B, \gamma_{\transformCoupling}, \tilde{\gamma})
                  \\
                  \nonumber
                   & \qquad\defas
                  \begin{aligned}[t]
                       & \semLet~\tilde{\mu} = \proofCoupling(t)(\gamma_A, \gamma_B, \gamma_{\transformCoupling}, \tilde{\gamma}) \\
                       & \semLet~g = (u_A, u_B, u_{\transformCoupling}, \tilde{u}) \mapsto \proofCoupling(t')
                      \Big(
                      \begin{aligned}[t]
                           & \gamma_A[x \mapsto u_A], \gamma_B[x \mapsto u_B],                                                        \\
                           & \gamma_{\transformCoupling}[x \mapsto u_{\transformCoupling}], \tilde{\gamma}[x \mapsto \tilde{u}]\Big)
                      \end{aligned}
                      \\
                       & \tilde{\mu} \semPBind{} g.
                  \end{aligned}
              \end{align}
              Let $(x_A, x_B, x_{\transformCoupling}, \tilde{x}) \defas
                  \relationCouplingProof(x \codePGets t;~t')(\gamma_A, \gamma_B, \gamma_{\transformCoupling}, \tilde{\gamma})$.
              We have
              \begin{align}
                   & P(\semProj_1)(\tilde{x})
                  \\
                   & \quad =
                  \begin{aligned}[t]
                       & \proofCoupling(t)(\gamma_A, \gamma_B, \gamma_{\transformCoupling}, \tilde{\gamma})
                      \semPBind{}
                      \\
                       & \quad(u_A, u_B, u_{\transformCoupling}, \tilde{u}) \mapsto
                      P(\semProj_1)\big(
                      \proofCoupling(t')(
                      \begin{aligned}[t]
                           & \gamma_A[x \mapsto u_A],
                          \gamma_B[x \mapsto u_B],
                          \\
                           & \gamma_{\transformCoupling}[x \mapsto u_{\transformCoupling}],
                          \tilde{\gamma}[x \mapsto \tilde{u}])\big)
                      \end{aligned}
                  \end{aligned}
                  \\
                   & \quad =
                  \proofCoupling(t)(\gamma_A, \gamma_B, \gamma_{\transformCoupling}, \tilde{\gamma})
                  \semPBind{}
                  \left((u_A, u_B, u_{\transformCoupling}, \tilde{u}) \mapsto
                  \semP{t'}(\gamma_A[x \mapsto u_A])\right)
                  \label{eq:coupling-proof-bind-proj-output}
                  \\
                   & \quad =
                  \semP{t}(\gamma_A)
                  \semPBind{}
                  \left(u_A \mapsto
                  \semP{t'}(\gamma_A[x \mapsto u_A])\right)
                  \label{eq:coupling-proof-bind-proj-intermediate}
                  \\
                   & \quad =
                  x_A,
              \end{align}
              where \cref{eq:coupling-proof-bind-proj-output} follows from the induction hypothesis
              on $\Gamma, x : \tau \vdash t' : \typeProb~\tau'$,
              and \cref{eq:coupling-proof-bind-proj-intermediate} follows from the induction hypothesis
              on $\Gamma \vdash t : \typeProb~\tau$.
              The identities $P(\semProj_2)(\tilde{x}) = x_B$ and $P(\semProj_3)(\tilde{x}) = x_{\transformCoupling}$
              follow by analogous logic.
              It follows from the definition of $\relationCouplingProof(\typeProb~\tau)$
              that $\relationCouplingProof(x \codePGets t;~t')(\gamma_A, \gamma_B, \gamma_{\transformCoupling}, \tilde{\gamma})
                  \in \relationCouplingProof(\typeProb~\tau)$.
        \item \textbf{Case } \(
              \boxed{
                  \begin{array}{c}
                      \Gamma \vdash t: \typeBool
                      \quad \Gamma \vdash t_1 : \type
                      \quad \Gamma \vdash t_2: \type
                      \\
                      \hline
                      \Gamma \vdash \codeIfThenElse{t}{t_1}{t_2} : \type
                  \end{array}
              }
              \)
              \textbf{: }
              \par\vskip0.5em\relax

              To propagate the proof value, we define a helper for merging proof values that behaves
              similarly to the original merge helper.
              Letting $\semMerge_\tau(u, v) \defas \semP{\transformMerge_\tau\set{x}\set{y}}([x \mapsto u, y \mapsto v])$
              be the semantic-space analog of the original merge helper, we define
              \begin{align}
                  \semMergeProof_{\sigma}(\tilde{x}_1, \tilde{x}_2)                      & \defas \semUnit
                  \\
                  \semMergeProof_{\tau_1 \typeProduct \tau_2}(\tilde{x}_1, \tilde{x}_2)  & \defas \left(\semMergeProof_{\tau_1}(\semFst(\tilde{x}_1), \semFst(\tilde{x}_2)), \semMergeProof_{\tau_2}(\semSnd(\tilde{x}_1), \semSnd(\tilde{x}_2))\right)
                  \\
                  \semMergeProof_{\tau_1 \typeFunction \tau_2}(\tilde{x}_1, \tilde{x}_2) & \defas (x_A, x_B, x_{\transformCoupling}, \tilde{x}) \mapsto \semMergeProof_{\tau_2}\left(\tilde{x}_1(x_A, x_B, x_{\transformCoupling}, \tilde{x}), \tilde{x}_2(x_A, x_B, x_{\transformCoupling}, \tilde{x})\right)
                  \\
                  \semMergeProof_{\typeProb~\tau}(\tilde{x}_1, \tilde{x}_2)              & \defas
                  \Big\{
                  \begin{aligned}[t]
                       & (x_{A,1}, x_{B, 1}, x_{\transformCoupling, 1}, \tilde{x}_1) \semPGets \tilde{x}_1;~(x_{A,2}, x_{B, 2}, x_{\transformCoupling, 2}, \tilde{x}_2) \semPGets \tilde{x}_2; \\
                       & \semPReturn\left(x_{A,1}, x_{B, 2}, \semMerge_\tau(x_{\transformCoupling, 1}, x_{\transformCoupling, 2}), \semMergeProof_\tau(\tilde{x}_1, \tilde{x}_2)\right)
                      \Big\}
                  \end{aligned}
              \end{align}
              We now define
              \begin{align}
                   & \proofCoupling(\Gamma \vdash \codeIfThenElse{t}{t_1}{t_2} : \tau)
                  (\gamma_A, \gamma_B, \gamma_{\transformCoupling}, \tilde{\gamma})
                  \\
                  \nonumber
                   & \quad\defas
                  \begin{aligned}[t]
                       & \semLet~(b_A, b_B) = \semP{\transformCoupling{t}}(\gamma_{\transformCoupling})
                      \\
                       & \semLet~\tilde{x}_A = \semIfThenElse{b_A}{
                          \proofCoupling(t_A)(\gamma_A, \gamma_B, \gamma_{\transformCoupling}, \tilde{\gamma})
                      }{
                          \proofCoupling(t_B)(\gamma_A, \gamma_B, \gamma_{\transformCoupling}, \tilde{\gamma})
                      }                                                                                 \\
                       & \semLet~\tilde{x}_B = \semIfThenElse{b_B}{
                          \proofCoupling(t_A)(\gamma_A, \gamma_B, \gamma_{\transformCoupling}, \tilde{\gamma})
                      }{
                          \proofCoupling(t_B)(\gamma_A, \gamma_B, \gamma_{\transformCoupling}, \tilde{\gamma})
                      }                                                                                 \\
                       & \semMergeProof_\tau(\tilde{x}_A, \tilde{x}_B).
                  \end{aligned}
              \end{align}
              As before, since $\proofCoupling(\codeIfThenElse{t}{t_1}{t_2})$
              is defined above using semantic-space analogs
              of core $\langP$ syntax, it is a $\catQBS$-morphism by construction.
              To conclude, we wish to justify that $\relationCouplingProof(\codeIfThenElse{t}{t_1}{t_2})$ preserves the logical relation.
              Define $t_A$ to be the term $t_1$ if $\semP{t}(\gamma_A) = \semTrue$, and $t_2$ if $\semP{t}(\gamma_A) = \semFalse$.
              Similarly, define $t_B$ to be the term $t_1$ if $\semP{t}(\gamma_B) = \semTrue$, and $t_2$ if $\semP{t}(\gamma_B) = \semFalse$.
              Then, we have
              \begin{align}
                   & \relationCouplingProof(\codeIfThenElse{t}{t_1}{t_2})(\gamma_A, \gamma_B, \gamma_{\transformCoupling}, \tilde{\gamma})
                  \\
                   & =
                  \Big(
                  \begin{aligned}[t]
                       & \semP{t_A}(\gamma_A), \semP{t_B}(\gamma_B), \semP{\transformMerge_\tau\set{\transformCoupling{t_A}}\set{\transformCoupling{t_B}}}(\gamma_{\transformCoupling}), \\
                       & \semMergeProof_\tau\left(\proofCoupling(t_A)(\gamma_A, \gamma_B, \gamma_{\transformCoupling}, \tilde{\gamma}),
                      \proofCoupling(t_B)(\gamma_A, \gamma_B, \gamma_{\transformCoupling}, \tilde{\gamma})\right) \Big),
                  \end{aligned}
              \end{align}
              whose membership in $\relationCouplingProof(\tau)$
              follows by induction on the definitions of $\transformMerge_\tau\set{\cdot}\set{\cdot}$
              and $\semMergeProof_\tau$,
              applying the induction hypothesis on $\Gamma \vdash t_1 : \tau$ and $\Gamma \vdash t_2 : \tau$.
    \end{itemize}
    To conclude the proof, consider a term $\Gamma \vdash t : \tau$.
    Suppose $(\gamma_A, \gamma_B, \gamma_{\transformCoupling}) \in \relationCoupling(\Gamma)$.
    Then there exists a $\tilde{\gamma}$ such that $(\gamma_A, \gamma_B, \gamma_{\transformCoupling}, \tilde{\gamma}) \in \relationCouplingProof(\Gamma)$.
    Hence
    \begin{equation}
        \left(\semP{t}(\gamma_A),
        \semP{t}(\gamma_B),
        \semP{\transformCoupling{t}}(\gamma_{\transformCoupling}),
        \proofCoupling(t)(\gamma_A, \gamma_B, \gamma_{\transformCoupling}, \tilde{\gamma})\right)
        =
        \relationCouplingProof(t)(\gamma_A, \gamma_B, \gamma_{\transformCoupling}, \tilde{\gamma})
        \in \relationCouplingProof(\tau),
    \end{equation}
    and thus
    \begin{equation}
        \left(\semP{t}(\gamma_A),
        \semP{t}(\gamma_B),
        \semP{\transformCoupling{t}}(\gamma_{\transformCoupling})\right)
        \in \relationCoupling(\tau),
    \end{equation}
    as desired.
\end{proof}

\clearpage
\section{Deferred Material for \texorpdfstring{\hyperref[sec:partial-evaluation]{\S\ref*{sec:partial-evaluation}}}{\S\ref{sec:partial-evaluation}}: Proofs of Soundness of \texorpdfstring{$\transformModCoupling{\cdot}$}{C*\{·\}}
  and \texorpdfstring{$\transformPPEval{\cdot}$}{P\{·\}}}
\label{appendix:deferred-partial-evaluation}

In this appendix, we provide
full definitions and additional categorical theoretical detail for
the objects introduced in \cref{sec:partial-evaluation},
and supply the proofs of \cref{thm:factorized-coupling-correctness,thm:langPP-realization}.

\begin{definition}[Factorized Coupling Transformation]
    \label{def:factorized-coupling-transformation-full}
    The factorized coupling program transformation is a functor
    \begin{equation}
        \transformModCoupling{\cdot} :
        \catSyntax(\langP)
        \longrightarrow
        \catSyntax(\langPP)
    \end{equation}
    whose action on types and terms is given in \cref{def:factorized-coupling},
    utilizing the modified merge macro $\transformModMerge_\tau\set{\cdot}\set{\cdot}$
    defined as follows:
    \begin{align}
         & \transformModMerge_\sigma\set{t_1}\set{t_2} \defas
        \begin{aligned}[t]
             & \codeLet~((x_A : \sigma, \bar{x}_B : \typeResidual~\sigma),
            \bar{y} : \typeResidual~(\sigma \typeProduct \typeResidual~\sigma)) = (t_1,t_2)                                      \\
             & (x_A, \{(y_A, \bar{y}_B) \codeRGets \bar{y};~y_B \codeRGets \bar{y}_B;~\codeRHide~y_B\})
        \end{aligned}
        \\
         & \transformModMerge_{\tau_1 \typeProduct \tau_2}\set{t_1}\set{t_2} \defas
         \big(\begin{aligned}[t]
         &\transformModMerge_{\tau_1}\set{\codeFst~t_1}\set{y \codeRGets t_2;~\codeRHide~(\codeFst~y)}, \\
         &\transformModMerge_{\tau_2}\set{\codeSnd~t_1}\set{y \codeRGets t_2;~\codeRHide~(\codeSnd~y)}\big)
         \end{aligned}
         \\
         & \transformModMerge_{\tau_1 \typeFunction \tau_2}\set{t_1}\set{t_2} \defas
        \lambda(x : \transformModCoupling{\tau_1}).
        \transformModMerge_{\tau_2}\set{t_1~x}\set{y \codeRGets t_2;~\codeRHide~(y~x)} \\
         & \transformModMerge_{\typeProb~\tau}\set{t_1}\set{t_2}\defas
        \begin{aligned}[t]
        &x \codePPGets t_1\\
        &\bar{y} \codePPGets \{z \codeRGets t_2;~y \codePPGets z;~\codePPReturn~(\codeRHide~y)\}\\
        &\codePPReturn~\transformModMerge_\tau\set{x}\set{\bar{y}}.
        \end{aligned}
    \end{align}
\end{definition}

\thmFactorizedCouplingCorrectness*
\begin{proof}
    We first show that $\transformErasure{\cdot} \circ \transformModCoupling{\cdot}$
    satisfies the same inductive definitions on types and terms as
    $\transformCoupling{\cdot}$ in \cref{def:coupling-transformation}.
    We have
    \begin{align}
         & \transformErasure{\transformModCoupling{\base}}
        = \transformErasure{\base \typeProduct \typeResidual~\base}
        = \base \typeProduct \transformErasure{\typeResidual~\base}
        = \base \typeProduct \base
        \\
         & \transformErasure{\transformModCoupling{c}} = (c, \transformErasure{\codeRHide~c})
        = (c, c)
        \\
         & \transformErasure{\transformModCoupling{\typeProb~\tau}}  =
        \transformErasure{\typePProb~(\transformModCoupling{\tau})}
        = \typeProb~(\transformErasure{\transformModCoupling{\tau}})
        \\
         &
        \transformErasure{\transformModCoupling{\codePReturn~t}}      =
        \transformErasure{\codePPReturn~(\transformModCoupling{t})}
        = \codePReturn~(\transformErasure{\transformModCoupling{t}})
        \\
         & \transformErasure{\transformModCoupling{x \codePGets t;~t'}} =
        \transformErasure{\set{x \codePPGets \transformModCoupling{t};~\transformModCoupling{t'}}}
        = \set{x \codePGets \transformErasure{\transformModCoupling{t}};~\transformErasure{\transformModCoupling{t'}}}.
    \end{align}
    For the case of $\codeIfThenElse{t}{t_1}{t_2}$, by comparing the definitions of $\transformMerge_\tau\set{\cdot}\set{\cdot}$
    and $\transformModMerge_\tau\set{\cdot}\set{\cdot}$ we can verify that $\transformErasure{\transformModMerge_\tau\set{t_1}\set{t_2}}
        = \transformMerge_\tau\set{\transformErasure{t_1}}\set{\transformErasure{t_2}}$, and thus
    \begin{align}
         & \transformErasure{\transformModCoupling{\Gamma \vdash \codeIfThenElse{t}{t_1}{t_2} : \tau}} \\
         & \quad =
        \begin{aligned}[t]
             & \codeLet~(b_A, b^*_B) = \transformErasure{\transformModCoupling{t}};                                                                               \\
             & \hspace{-2.5em} \transformErasure\big\{\transformModMerge_{\tau}\set{\codeIfThenElse{b_A}{\transformModCoupling{t_1}}{\transformModCoupling{t_2}}}
            \set{b_B \codeRGets b_B^*;~\codeRHide~(\codeIfThenElse{b_B}{\transformModCoupling{t_1}}{\transformModCoupling{t_2}})}\big\}
        \end{aligned}
        \\
         & \quad =
        \begin{aligned}[t]
             & \codeLet~(b_A, b^*_B) = \transformErasure{\transformModCoupling{t}}; \\
             &
            \transformMerge_\tau\set{\codeIfThenElse{b_A}{\transformErasure{\transformModCoupling{t_1}}}{\transformErasure{\transformModCoupling{t_2}}}}
            \set{\codeIfThenElse{b_B^*}{\transformErasure{\transformModCoupling{t_1}}}{\transformErasure{\transformModCoupling{t_2}}}}.
        \end{aligned}
    \end{align}
    Hence the inductive definitions match for all types and terms.
    Since
    $(\semP{p}, \semP{p}, \semP{\transformErasure{\transformModCoupling{p}}})
        \in \relationCoupling(\tau_p)$ for each primitive $p$ appearing in $t$
    by assumption,
    the desired follows from \cref{thm:coupling-correctness}.
\end{proof}
\begin{definition}[Partial Probability Evaluation Transformation]
    \label{def:partial-probability-evaluation-transformation-full}
    The partial probability evaluation program transformation is a functor
    \begin{equation}
        \transformPPEval{\cdot} : \catSyntax(\langPP) \longrightarrow \catSyntax(\langP),
    \end{equation}
    whose action on types and terms is given in \cref{def:factorized-coupling}.
\end{definition}
\begin{definition}[Partial Probability Evaluation Logical Relation]
    \label{def:logical-relation-realization}
    The proof-relevant logical relation for the partial probability evaluation of $\langPP$
    associates each type $\tau$ with a quasi-Borel space $\relationPPEvalProof(\tau)$
    of $\langPP$,
    where
    \begin{equation}
        \relationPPEvalProof(\tau) \subset \semP{\transformErasure{\tau}} \times \semP{\transformPPEval{\tau}} \times
        \proofPPEval(\tau).
    \end{equation}
    In words, $\relationPPEvalProof(\tau)$ relates tuples $(x_{\transformErasure}, x_{\transformPPEval}, \tilde{x})$,
    where $x_{\transformErasure} : \semP{\transformErasure{\tau}}$,
    $x_{\transformPPEval} : \semP{\transformPPEval{\tau}}$,
    and $\tilde{x} : \proofPPEval(\tau)$.
    Here, the third component $\tilde{x}$ is a \emph{proof value} belonging to the space $\proofPPEval(\tau)$,
    where $\proofPPEval$ itself associates types of $\langP$ with quasi-Borel spaces.
    %
\begin{listing}[t]
\small

\setlength{\abovedisplayshortskip}{0pt}
\setlength{\belowdisplayshortskip}{0pt}
\setlength{\abovedisplayskip}{0pt}
\setlength{\belowdisplayskip}{0pt}

\begin{syntaxbox}[title=\textbf{Proof Space $\proofPPEval$ on Types of $\langPP$}]{\linewidth}
 \begin{align*}
        \proofPPEval(\sigma)  &\defas \setUnit
        &\qquad
        \proofPPEval(\tau_1 \typeProduct \tau_2) &\defas
            \proofPPEval(\tau_1) \times \proofPPEval(\tau_2)
        \\
        \proofPPEval(\tau_1 \typeFunction \tau_2) &\defas
            \relationPPEvalProof(\tau_1) \to \proofPPEval(\tau_2)
        &\qquad
        \proofPPEval(\typeResidual~\tau) &\defas \proofPPEval(\tau)
        \\
        \proofPPEval(\typeProb~\type) &\defas P(\relationPPEvalProof(\type)) \times \left([0,1] \to \relationPPEvalProof(\type)\right)
        &\qquad
        \proofPPEval(\typePProb~\type) &\defas [0, 1] \to P\left(\relationPPEvalProof(\type)\right)
 \end{align*}
\end{syntaxbox}

\begin{syntaxbox}[title=\textbf{Proof-Relevant Partial Probability Evaluation Logical Relation $\relationPPEvalProof$ on Types of $\langPP$}]{\linewidth}
 \begin{align*}
        \relationPPEvalProof(\sigma)                  & \defas
        \{(x_{\transformErasure}, x_{\transformPPEval}, \semUnit)
        \mid x_{\transformErasure} = x_{\transformPPEval}\}.
        \\
        \relationPPEvalProof(\tau_1 \typeProduct \tau_2) &=
        \Big\{(x_{\transformErasure}, x_{\transformPPEval}, \tilde{x}) \mathrel{\big|}
        \begin{aligned}[t]
            &(\semFst(x_{\transformErasure}), \semFst(x_{\transformPPEval}), \semFst(\tilde{x})) \in \relationPPEvalProof(\tau_1)
            \\
            &\land (\semSnd(x_{\transformErasure}), \semSnd(x_{\transformPPEval}), \semSnd(\tilde{x})) \in \relationPPEvalProof(\tau_2)\Big\}
        \end{aligned}
        \\
        \relationPPEvalProof(\tau_1 \typeFunction \tau_2) &=
        \Big\{(f_{\transformErasure}, f_{\transformPPEval}, \tilde{f}) \mathrel{\big|}
        \begin{aligned}[t]
            &\forall (x_{\transformErasure}, x_{\transformPPEval}, \tilde{x}) \in \relationPPEvalProof(\tau_1).
            \\
            &\left(f_{\transformErasure}(x_{\transformErasure}), f_{\transformPPEval}(x_{\transformPPEval}), \tilde{f}(x_{\transformErasure}, x_{\transformPPEval}, \tilde{x})\right)
            \in \relationPPEvalProof(\tau_2)\Big\}
        \end{aligned}
        \\
        \relationPPEvalProof(\typeResidual~\type)     & \defas \relationPPEvalProof(\type)
        \\
        \relationPPEvalProof(\typeProb~\type)         & \defas
        \Big\{
        (x_{\transformErasure}, (x_{\transformPPEval,1}, x_{\transformPPEval,2}), (\tilde{x}_1, \tilde{x}_2))
        \mathrel{\big|}
        \begin{aligned}[t]
        \\
        &\hspace{-8em}
        \left(P(\semProj_1)(\tilde{x}_1), P(\semProj_2)(\tilde{x}_1)\right)
        = \left(x_{\transformErasure}, x_{\transformPPEval,1}\right)
        \\
        &\hspace{-8em}\land
        \left(P(\semProj_1 \circ \tilde{x}_2)(\semSeed), P(\semProj_2 \circ \tilde{x}_2)(\semSeed)\right)
        = \left(x_{\transformErasure}, P(x_{\transformPPEval,2})(\semSeed)\right)
        \Big\}
        \end{aligned}
        \label{eq:relation-realization-prob}
        \\
        \relationPPEvalProof(\typePProb~\type)       & \defas
        \Big\{
        (x_{\transformErasure}, x_{\transformPPEval}, \tilde{x})
        \mathrel{\big|}
        \begin{aligned}[t]
        \\
        &\hspace{-3em}\left(P(\semProj_1)(\semSeed \semPBind{} \tilde{x}),
        P(\semProj_2)(\semSeed \semPBind{} \tilde{x})\right)
        = \left(x_{\transformErasure}, \semSeed \semPBind{} x_{\transformPPEval}\right)
        \Big\}
        \end{aligned}
 \end{align*}
\end{syntaxbox}

\captionsetup{aboveskip=5pt,belowskip=-10pt}
\caption{The proof-relevant partial probability evaluation logical relation on $\langPP$ types.
}
\label{def:realization-logical-relation}
\end{listing}

    \cref{def:realization-logical-relation} gives the definition of the proof-relevant
    partial probability evaluation logical relation on the types of $\langPP$.
    We define the proof-irrelevant relation by
    \begin{equation}
        \relationPPEval(\tau) \defas
        \set{(x_{\transformErasure}, x_{\transformPPEval}) \mid
            \exists \tilde{x}. (x_{\transformErasure}, x_{\transformPPEval}, \tilde{x}) \in \relationPPEvalProof(\tau)}.
    \end{equation}
\end{definition}
\thmLangPPRealization*
\begin{proof}
    We extend $\proofPPEval$ and $\relationPPEvalProof$ into functors
    \begin{align}
        \proofPPEval         & : \catSyntax(\langPP) \longrightarrow \catQBS
        \\
        \relationPPEvalProof & : \catSyntax(\langPP) \longrightarrow \catQBS
    \end{align}
    where we define $\relationPPEvalProof$ on a term $\Gamma \vdash t : \tau$ by
    \begin{equation}
        \relationPPEvalProof(t)(\gamma_{\transformErasure}, \gamma_{\transformPPEval}, \tilde{\gamma})
        \defas
        \left(\semP{\transformErasure{t}}(\gamma_{\transformErasure}),
        \semP{\transformPPEval{t}}(\gamma_{\transformPPEval}),
        \proofPPEval(t)(\gamma_{\transformErasure}, \gamma_{\transformPPEval}, \tilde{\gamma})\right).
        \label{eq:proof-langpp-realization-relation-defn}
    \end{equation}
    It remains to define the morphism $\proofPPEval(t) : \proofPPEval(\Gamma) \to \proofPPEval(\tau)$
    that constructs the proof value
    for each term $\Gamma \vdash t : \tau$,
    in such a way
    that $\relationPPEvalProof(t)$ preserves the logical relation, i.e. $\relationPPEvalProof(t)$
    is a morphism from $\relationPPEvalProof(\Gamma)$ to $\relationPPEvalProof(\tau)$.
    As is typical for a logical relation proof, we proceed by induction
    on the typing derivation of $\Gamma \vdash t : \tau$:
    in each case of the induction, we will first define
    $\proofPPEval(t)$,
    and then verify that $\relationPPEvalProof(t)$
    preserves the logical relation.

    In the following, let $(\gamma_{\transformErasure}, \gamma_{\transformPPEval}, \tilde{\gamma})$
    denote an arbitrary tuple in $\relationPPEvalProof(\Gamma)$.
    Note that $\relationPPEvalProof(t)$ is a morphism from $\relationPPEvalProof(\Gamma)$ to
    $\semP{\transformErasure{\tau}} \times \semP{\transformPPEval{\tau}} \times \proofPPEval(\tau)$
    by its definition in \cref{eq:proof-langpp-realization-relation-defn}:
    by \cref{lemma:morphism-into-subspace},
    showing that $\relationPPEvalProof(t)$ is a morphism from $\relationPPEvalProof(\Gamma)$ to $\relationPPEvalProof(\tau)$
    is equivalent to showing that $\relationPPEvalProof(t)(\gamma_{\transformErasure}, \gamma_{\transformPPEval}, \tilde{\gamma}) \in \relationPPEvalProof(\tau)$
    for each $(\gamma_{\transformErasure}, \gamma_{\transformPPEval}, \tilde{\gamma}) \in \relationPPEvalProof(\Gamma)$.
    \begin{itemize}[leftmargin=*,topsep=10pt,itemsep=10pt]
        \item \textbf{Cases } \(
              \boxed{
                  \Gamma \vdash \codeTrue : \typeBool
              }
              \hspace{0.5em}
              \boxed{
                  \Gamma \vdash \codeFalse : \typeBool
              }
              \hspace{0.5em}
              \boxed{
                  \Gamma \vdash n : \typeInteger
              }
              \hspace{0.5em}
              \boxed{
                  \Gamma \vdash r : \typeReal
              }
              \hspace{0.5em}
              \boxed{
                  \Gamma \vdash \codeUnit : \typeUnit
              }
              \)
              \textbf{: }
              \par\vskip0.5em\relax

              For a constant $c$, we define
              \begin{equation}
                  \proofPPEval(c)(\gamma_{\transformErasure}, \gamma_{\transformPPEval}, \tilde{\gamma}) \defas \semUnit.
              \end{equation}
              Since $\transformErasure{c} = c$, $\transformPPEval{c} = c$, and constants $c$
              are of base type $\sigma$, it follows that
              \begin{equation}
                  \relationPPEvalProof(c)(\gamma_{\transformErasure}, \gamma_{\transformPPEval}, \tilde{\gamma})
                  = (c, c, \semUnit) \in \relationPPEvalProof(\sigma).
              \end{equation}
        \item \textbf{Case } \(
              \boxed{
                  \Gamma \vdash p : \tau_p
              }
              \)
              \textbf{: }
              \par\vskip0.5em\relax

              Follows from the precondition on primitives in the theorem statement, analogously to the corresponding case in \cref{thm:coupling-correctness}.
              Specifically, we define $\proofPPEval(p)$ to be a constant function,
              \begin{equation}
                  \proofPPEval(p)(\gamma_{\transformErasure}, \gamma_{\transformPPEval}, \tilde{\gamma}) \defas \tilde{x},
              \end{equation}
              where $\tilde{x}$ is chosen as an (environment-independent) proof value such that
              \begin{equation}
                  (\semP{\transformErasure{p}}, \semP{\transformPPEval{p}}, \tilde{x})
                  =
                  (\semP{p}, \semP{\transformPPEval{p}}, \tilde{x})
                  \in \relationPPEvalProof(\tau_p),
              \end{equation}
              whose existence is guaranteed by the assumption $(\semP{p}, \semP{\transformPPEval{p}}) \in \relationPPEval(\tau_p)$.
        \item \textbf{Cases} \(\boxed{\vphantom{(j_1)~t'}\lambda(x : \tau).t}\hspace{0.5em}
              \boxed{\vphantom{(j_1)~t'}t'~t}\hspace{0.5em}
              \boxed{\vphantom{(j_1)~t'}(t_1, t_2)}\hspace{0.5em}
              \boxed{\vphantom{(j_1)~t'}\codeFst~t}\hspace{0.5em}
              \boxed{\vphantom{(j_1)~t'}\codeSnd~t}\hspace{0.5em}
              \boxed{\vphantom{(j_1)~t'}\codeLet~x=t;~t'}\)
              \textbf{: }
              \par\vskip0.5em\relax

              These cases follow by analogous reasoning to the corresponding
              cases in \cref{thm:coupling-correctness}, as the inductive definitions we rely upon
              for these cases are the same for $\relationPPEvalProof$ and $\relationCouplingProof$.
        \item \textbf{Case } \(
              \boxed{
                  \begin{array}{c}
                      \Gamma \vdash t: \type
                      \\
                      \hline
                      \Gamma \vdash \codePReturn~\ter : \typeProb~\type
                  \end{array}
              }
              \)
              \textbf{: }
              \par\vskip0.5em\relax

              We define
              \begin{equation}
                  \proofPPEval(\codePReturn~t)(\gamma_{\transformErasure}, \gamma_{\transformPPEval}, \tilde{\gamma}) \defas
                  \begin{aligned}[t]
                      \\
                      \hspace{-8em}\left(\semPReturn(\relationPPEvalProof(t)(\gamma_{\transformErasure}, \gamma_{\transformPPEval}, \tilde{\gamma})),
                      \lambda\_. \relationPPEvalProof(t)(\gamma_{\transformErasure}, \gamma_{\transformPPEval}, \tilde{\gamma})\right).
                  \end{aligned}
              \end{equation}
              Let
              $(x_{\transformErasure}, (x_{\transformPPEval,1}, x_{\transformPPEval, 2}), (\tilde{x}_1, \tilde{x}_2)) \defas
                  \relationPPEvalProof(\codePReturn~t)(\gamma_{\transformErasure}, \gamma_{\transformPPEval}, \tilde{\gamma})$.
              We have
              \begin{align}
                  x_{\transformErasure}   & = \semPReturn(\semP{\transformErasure{t}}(\gamma_{\transformErasure}))
                  \label{eq:proof-langpp-realization-preturn-erasure}
                  \\
                  x_{\transformPPEval, 1} & = P(x_{\transformPPEval, 2})(\semSeed)
                  = \semPReturn(\semP{\transformPPEval{t}}(\gamma_{\transformPPEval}))
                  \label{eq:proof-langpp-realization-preturn-peval}
                  \\
                  \tilde{x}_{1}           & = P(\tilde{x}_2)(\semSeed)
                  = \semPReturn(\relationPPEvalProof(t)(\gamma_{\transformErasure}, \gamma_{\transformPPEval}, \tilde{\gamma})).
              \end{align}
              From the induction hypothesis on $\Gamma \vdash t : \tau$, it follows that
              $(x_{\transformErasure}, (x_{\transformPPEval,1}, x_{\transformPPEval, 2}), (\tilde{x}_1, \tilde{x}_2)) \in \relationPPEvalProof(\typeProb~\tau)$.

        \item \textbf{Case } \(
              \boxed{
                  \begin{array}{c}
                      \Gamma \vdash t: \typeProb~\type
                      \quad \extend{\Gamma}{x}{\type} \vdash t': \typeProb~\type'
                      \\
                      \hline
                      \Gamma \vdash x \codePGets t;~t' : \typeProb~\type'
                  \end{array}
              }
              \)
              \textbf{: }
              \par\vskip0.5em\relax

              We define
              \begin{equation}
                  \begin{aligned}[t]
                       & \proofPPEval(x \codePGets t;~t')(\gamma_{\transformErasure}, \gamma_{\transformPPEval}, \tilde{\gamma}) \defas
                      \\
                       & \quad
                      \begin{aligned}[t]
                          \Big( & \semFst(\proofPPEval(t)(\gamma_{\transformErasure}, \gamma_{\transformPPEval}, \tilde{\gamma}))
                          \semPBind{} \big((u_{\transformErasure}, u_{\transformPPEval}, \tilde{u}) \mapsto
                          \semFst(\proofPPEval(t')(
                          \begin{aligned}[t]
                               & \gamma_{\transformErasure}[x \setvar u_{\transformErasure}],\gamma_{\transformPPEval}[x \setvar u_{\transformPPEval}], \\
                               & \tilde{\gamma}[x \setvar \tilde{u}]))\big),
                          \end{aligned}
                          \\
                                & s \mapsto
                          \begin{aligned}[t]
                               & \semLet~(s_0, s_1) = \semSplit(s)
                              \\
                               & \semLet~(u_{\transformErasure}, u_{\transformPPEval}, \tilde{u}) = \semSnd(\proofPPEval(t)(\gamma_{\transformErasure}, \gamma_{\transformPPEval}, \tilde{\gamma}))(s_0)
                              \\
                               & \semSnd(\proofPPEval(t')(\gamma_{\transformErasure}[x \setvar u_{\transformErasure}], \gamma_{\transformPPEval}[x \setvar u_{\transformPPEval}], \tilde{\gamma}[x \setvar \tilde{u}]))(s_1)\Big).
                          \end{aligned}
                      \end{aligned}
                  \end{aligned}
              \end{equation}
              Let
              $(x_{\transformErasure}, (x_{\transformPPEval,1}, x_{\transformPPEval, 2}), (\tilde{x}_1, \tilde{x}_2)) \defas
                  \relationPPEvalProof(x \codePGets t;~t')(\gamma_{\transformErasure}, \gamma_{\transformPPEval}, \tilde{\gamma})$.
              We have that
              \begin{align}
                  x_{\transformErasure} & =
                  \semP{\transformErasure{t}}(\gamma_{\transformErasure}) \semPBind{} \left(u \mapsto
                  \semP{\transformErasure{t'}}(\gamma_{\transformErasure}[x \setvar u])\right)
                  \label{eq:proof-langpp-realization-pbind-erasure}
              \end{align}
              and
              \begin{align}
                  x_{\transformPPEval, 1} & =
                  \semFst(\semP{\transformPPEval{t}}(\gamma_{\transformPPEval})) \semPBind{} \left(u \mapsto
                  \semFst(\semP{\transformPPEval{t'}}(\gamma_{\transformPPEval}[x \setvar u]))\right)
                  \label{eq:proof-langpp-realization-pbind-peval1}
              \end{align}
              From the induction hypothesis on $\Gamma \vdash t : \typeProb~\tau$
              and $\Gamma, x : \tau \vdash t' : \typeProb~\tau'$,
              utilizing the \emph{first} condition imposed in the constructions of $\relationPPEvalProof(\typeProb~\tau)$
              and $\relationPPEvalProof(\typeProb~\tau')$ in \cref{def:realization-logical-relation},
              it follows that
              \begin{equation}
                  \left(P(\semProj_1)(\tilde{x}_1), P(\semProj_2)(\tilde{x}_1)\right)
                  = (x_{\transformErasure}, x_{\transformPPEval, 1}).
                  \label{eq:proof-langpp-realization-pbind-first-condition}
              \end{equation}
              Now, note that
              \begin{align}
                  P(x_{\transformPPEval, 2})(\semSeed) & =
                  \begin{aligned}[t]
                       & P\left(\semSnd(\semP{\transformPPEval{t}}(\gamma_{\transformPPEval}))\right)(\semSeed)
                      \semPBind{}
                      \\
                       & \quad\left(u \mapsto
                      P\left(\semSnd(\semP{\transformPPEval{t}}(\gamma_{\transformPPEval}))(u)\right)(\semSeed)
                      \right)
                  \end{aligned}
                  \label{eq:proof-langpp-realization-pbind-peval2}
                  \\
                  P(\tilde{x}_2)(\semSeed)             & =
                  \begin{aligned}[t]
                       & P\left(\semSnd(\proofPPEval(t)(\gamma_{\transformErasure}, \gamma_{\transformPPEval}, \tilde{\gamma}))\right)(\semSeed)
                      \semPBind{}
                      \\
                       & \hspace{-5em}\left((u_{\transformErasure}, u_{\transformPPEval}, \tilde{u}) \mapsto
                      P\left(\semSnd(\proofPPEval(t')(\gamma_{\transformErasure}[x \setvar u_{\transformErasure}], \gamma_{\transformPPEval}[x \setvar u_{\transformPPEval}], \tilde{\gamma}[x \setvar \tilde{u}]))\right)(\semSeed)
                      \right).
                  \end{aligned}
                  \label{eq:proof-langpp-realization-pbind-proof2}
              \end{align}
              From the induction hypothesis on $\Gamma \vdash t : \typeProb~\tau$
              and $\Gamma, x : \tau \vdash t' : \typeProb~\tau'$,
              utilizing the \emph{second} condition imposed in the constructions of
              $\relationPPEvalProof(\typeProb~\tau)$
              and $\relationPPEvalProof(\typeProb~\tau')$ in \cref{def:realization-logical-relation},
              it follows that
              \begin{equation}
                  \left(P(\semProj_1 \circ \tilde{x}_2)(\semSeed), P(\semProj_2 \circ \tilde{x}_2)(\semSeed)\right)
                  = (x_{\transformErasure}, P(x_{\transformPPEval, 2})(\semSeed)).
                  \label{eq:proof-langpp-realization-pbind-second-condition}
              \end{equation}
              By \cref{eq:proof-langpp-realization-pbind-first-condition,eq:proof-langpp-realization-pbind-second-condition},
              it follows that
              $(x_{\transformErasure}, (x_{\transformPPEval,1}, x_{\transformPPEval, 2}), (\tilde{x}_1, \tilde{x}_2)) \in \relationPPEvalProof(\typeProb~\tau')$.

        \item \textbf{Case } \(
              \boxed{
                  \begin{array}{c}
                      \Gamma \vdash t: \typeBool
                      \quad \Gamma \vdash t_1 : \type
                      \quad \Gamma \vdash t_2: \type
                      \\
                      \hline
                      \Gamma \vdash \codeIfThenElse{t}{t_1}{t_2} : \type
                  \end{array}
              }
              \)
              \textbf{: }
              \par\vskip0.5em\relax

              We define
              \begin{equation}
                  \begin{aligned}[t]
                       & \proofPPEval(\codeIfThenElse{t}{t_1}{t_2})(\gamma_{\transformErasure}, \gamma_{\transformPPEval}, \tilde{\gamma})
                      \defas
                      \\
                       & \quad\semIfThenElse{\semP{\transformErasure{t}}(\gamma_{\transformErasure})}
                      {\proofPPEval(t_1)(\gamma_{\transformErasure}, \gamma_{\transformPPEval}, \tilde{\gamma})}
                      {\proofPPEval(t_2)(\gamma_{\transformErasure}, \gamma_{\transformPPEval}, \tilde{\gamma})}
                  \end{aligned}
              \end{equation}
              Let $(x_{\transformErasure}, x_{\transformPPEval}, \tilde{x}) \defas
                  \relationPPEvalProof(\codeIfThenElse{t}{t_1}{t_2})(\gamma_{\transformErasure}, \gamma_{\transformPPEval}, \tilde{\gamma})$.
              Assume that $\semP{\transformErasure{t}}(\gamma_{\transformErasure}) = \semTrue$.
              By the induction hypothesis on $\Gamma \vdash t : \typeBool$, it holds that
              $\semP{\transformPPEval{t}}(\gamma_{\transformPPEval})
                  = \semP{\transformErasure{t}}(\gamma_{\transformErasure}) = \semTrue$.
              Thus, it holds that
              \begin{align}
                  (x_{\transformErasure}, x_{\transformPPEval}, \tilde{x})
                  = (\semP{\transformErasure{t_1}}(\gamma_{\transformErasure}),
                  \semP{\transformPPEval{t_1}}(\gamma_{\transformPPEval}),
                  \proofPPEval(t_1)(\gamma_{\transformErasure}, \gamma_{\transformPPEval}, \tilde{\gamma}))
                  \in \relationPPEvalProof(\tau),
              \end{align}
              where the final inclusion follows by the induction hypothesis on $\Gamma \vdash t_1 : \tau$.
              The case of $\semP{\transformErasure{t}}(\gamma_{\transformErasure}) = \semFalse$ follows
              by analogous reasoning, where the induction hypothesis on $\Gamma \vdash t_2 : \tau$ is applied instead
              in the final step.
        \item \textbf{Case } \(
              \boxed{
                  \begin{array}{c}
                      \Gamma \vdash t: \type
                      \\
                      \hline
                      \Gamma \vdash \codeRHide~t: \typeResidual~\type
                  \end{array}
              }
              \)
              \textbf{: }
              \par\vskip0.5em\relax

              We define
              \begin{equation}
                  \proofPPEval(\codeRHide~t)(\gamma_{\transformErasure}, \gamma_{\transformPPEval}, \tilde{\gamma}) \defas
                  \proofPPEval(t)(\gamma_{\transformErasure}, \gamma_{\transformPPEval}, \tilde{\gamma}).
              \end{equation}
              It follows
              that
              $\relationPPEvalProof(\codeRHide~t)(\gamma_{\transformErasure}, \gamma_{\transformPPEval}, \tilde{\gamma})
                  = \relationPPEvalProof(t)(\gamma_{\transformErasure}, \gamma_{\transformPPEval}, \tilde{\gamma})
                  \in \relationPPEvalProof(\tau)
                  = \relationPPEvalProof(\typeResidual~\tau)$.
        \item \textbf{Case } \(
              \boxed{
                  \begin{array}{c}
                      \Gamma \vdash t :\typeResidual~\type
                      \quad
                      \extend{\Gamma}{x}{\type} \vdash t' : \type'
                      \quad
                      \judgeHidden~\type'
                      \\
                      \hline
                      \Gamma \vdash x \codeRGets t;~t': \type'
                  \end{array}
              }
              \)
              \textbf{: }
              \par\vskip0.5em\relax

              We define
              \begin{equation}
                  \proofPPEval(x \codeRGets t;~t')(\gamma_{\transformErasure}, \gamma_{\transformPPEval}, \tilde{\gamma}) \defas
                  \proofPPEval(\codeLet~x = t;~t')(\gamma_{\transformErasure}, \gamma_{\transformPPEval}, \tilde{\gamma}).
              \end{equation}
              It follows
              that
              \begin{equation}
                  \relationPPEvalProof(x \codeRGets t;~t')(\gamma_{\transformErasure}, \gamma_{\transformPPEval}, \tilde{\gamma})
                  = \relationPPEvalProof(\codeLet~x = t;~t')(\gamma_{\transformErasure}, \gamma_{\transformPPEval}, \tilde{\gamma})
                  \in \relationPPEvalProof(\tau')
                  = \relationPPEvalProof(\typeResidual~\tau').
              \end{equation}
        \item \textbf{Case } \(
              \boxed{
                  \begin{array}{c}
                      \Gamma \vdash t: \type
                      \\
                      \hline
                      \Gamma \vdash \codePPReturn~t: \typePProb~\type
                  \end{array}
              }
              \)
              \textbf{: }
              \par\vskip0.5em\relax

              We define
              \begin{equation}
                  \proofPPEval(\codePPReturn~t)(\gamma_{\transformErasure}, \gamma_{\transformPPEval}, \tilde{\gamma}) \defas
                  \_ \mapsto \semPReturn\left(\relationPPEvalProof(t)(\gamma_{\transformErasure}, \gamma_{\transformPPEval}, \tilde{\gamma})\right).
              \end{equation}
              Let
              $(x_{\transformErasure}, x_{\transformPPEval}, \tilde{x}) \defas
                  \relationPPEvalProof(\codePPReturn~t)(\gamma_{\transformErasure}, \gamma_{\transformPPEval}, \tilde{\gamma})$.
              We have
              \begin{align}
                  x_{\transformErasure}                     & =
                  \semPReturn(\semP{\transformErasure{t}}(\gamma_{\transformErasure}))
                  \label{eq:proof-langpp-realization-ppreturn-erasure}
                  \\
                  \semSeed \semPBind{} x_{\transformPPEval} & =
                  \semPReturn(\semP{\transformPPEval{t}}(\gamma_{\transformPPEval}))
                  \label{eq:proof-langpp-realization-ppreturn-peval}
                  \\
                  \semSeed \semPBind{} \tilde{x}            & =
                  \semPReturn(\relationPPEvalProof(t)(\gamma_{\transformErasure}, \gamma_{\transformPPEval}, \tilde{\gamma})).
                  \label{eq:proof-langpp-realization-ppreturn-proof}
              \end{align}
              From the induction hypothesis on $\Gamma \vdash t : \tau$, it follows that
              $(x_{\transformErasure}, x_{\transformPPEval}, \tilde{x}) \in \relationPPEvalProof(\typePProb~\tau)$.
        \item \textbf{Case } \(
              \boxed{
                  \begin{array}{c}
                      \Gamma \vdash t: \typeProb~\type
                      \quad
                      \judgeTransparent~\type
                      \\
                      \hline
                      \Gamma \vdash \codePPPrimal~t: \typePProb~\type
                  \end{array}
              }
              \)
              \textbf{: }
              \par\vskip0.5em\relax

              We define
              \begin{equation}
                  \proofPPEval(\codePPPrimal~t)(\gamma_{\transformErasure}, \gamma_{\transformPPEval}, \tilde{\gamma}) \defas
                  s \mapsto
                  \semPReturn\left(
                  \semSnd\left(\proofPPEval(t)(\gamma_{\transformErasure}, \gamma_{\transformPPEval}, \tilde{\gamma})\right)(s)
                  \right).
              \end{equation}
              Let
              $(x_{\transformErasure}, x_{\transformPPEval}, \tilde{x}) \defas
                  \relationPPEvalProof(\codePPPrimal~t)(\gamma_{\transformErasure}, \gamma_{\transformPPEval}, \tilde{\gamma})$.
              We have
              \begin{align}
                  x_{\transformErasure}                     & =
                  \semP{\transformErasure{t}}(\gamma_{\transformErasure})
                  \label{eq:proof-langpp-realization-ppprimal-erasure}
                  \\
                  \semSeed \semPBind{} x_{\transformPPEval} & =
                  P(\semSnd(\semP{\transformPPEval{t}}(\gamma_{\transformPPEval})))(\semSeed)
                  \label{eq:proof-langpp-realization-ppprimal-peval}
                  \\
                  \semSeed \semPBind{} \tilde{x}            & =
                  P(\semSnd(\proofPPEval(t)(\gamma_{\transformErasure}, \gamma_{\transformPPEval}, \tilde{\gamma})))(\semSeed).
                  \label{eq:proof-langpp-realization-ppprimal-proof}
              \end{align}
              From the induction hypothesis on $\Gamma \vdash t : \typeProb~\tau$,
              utilizing the second condition imposed in the construction of $\relationPPEvalProof(\typeProb~\tau)$ in \cref{def:realization-logical-relation},
              it follows that
              $(x_{\transformErasure}, x_{\transformPPEval}, \tilde{x}) \in \relationPPEvalProof(\typePProb~\tau)$.
        \item \textbf{Case } \(
              \boxed{
                  \begin{array}{c}
                      \Gamma \vdash t: \typeResidual~(\typeProb~\type)
                      \quad
                      \judgeHidden~\type
                      \\
                      \hline
                      \Gamma \vdash \codePPResidual~t: \typePProb~\type
                  \end{array}
              }
              \)
              \textbf{: }
              \par\vskip0.5em\relax

              We define
              \begin{equation}
                  \proofPPEval(\codePPResidual~t)(\gamma_{\transformErasure}, \gamma_{\transformPPEval}, \tilde{\gamma}) \defas
                  \_ \mapsto
                  \semFst\left(\proofPPEval(t)(\gamma_{\transformErasure}, \gamma_{\transformPPEval}, \tilde{\gamma})\right).
              \end{equation}
              Let
              $(x_{\transformErasure}, x_{\transformPPEval}, \tilde{x}) \defas
                  \relationPPEvalProof(\codePPResidual~t)(\gamma_{\transformErasure}, \gamma_{\transformPPEval}, \tilde{\gamma})$.
              We have
              \begin{align}
                  x_{\transformErasure}                     & =
                  \semP{\transformErasure{t}}(\gamma_{\transformErasure})
                  \label{eq:proof-langpp-realization-ppresidual-erasure}
                  \\
                  \semSeed \semPBind{} x_{\transformPPEval} & =
                  \semFst(\semP{\transformPPEval{t}}(\gamma_{\transformPPEval}))
                  \label{eq:proof-langpp-realization-ppresidual-peval}
                  \\
                  \semSeed \semPBind{} \tilde{x}            & =
                  \semFst\left(\proofPPEval(t)(\gamma_{\transformErasure}, \gamma_{\transformPPEval}, \tilde{\gamma})\right).
                  \label{eq:proof-langpp-realization-ppresidual-proof}
              \end{align}
              From the induction hypothesis on $\Gamma \vdash t : \typeProb~\tau$,
              utilizing the first condition imposed in the construction of $\relationPPEvalProof(\typeProb~\tau)$ in \cref{def:realization-logical-relation},
              it follows that
              $(x_{\transformErasure}, x_{\transformPPEval}, \tilde{x}) \in \relationPPEvalProof(\typePProb~\tau)$.
        \item \textbf{Case } \(
              \boxed{
                  \begin{array}{c}
                      \Gamma \vdash t: \typePProb~\type
                      \quad
                      \extend{\Gamma}{x}{\type} \vdash t': \typePProb~\type'
                      \\
                      \hline
                      \Gamma \vdash x \codePPGets t;~t':
                      \typePProb~\type'
                  \end{array}
              }
              \)
              \textbf{: }
              \par\vskip0.5em\relax

              We define
              \begin{align}
                  \proofPPEval(x \codePPGets t;~t')(\gamma_{\transformErasure}, \gamma_{\transformPPEval}, \tilde{\gamma}) \defas
                  \begin{aligned}[t]
                       & s \mapsto
                      \semLet~(s_0, s_1) = \semSplit(s)
                      \\
                       & \hspace{-9em}\Big(\proofPPEval(t)(\gamma_{\transformErasure}, \gamma_{\transformPPEval}, \tilde{\gamma})(s_0) \semPBind{}
                      \\
                       & \hspace{-9em}\left((u_{\transformErasure}, u_{\transformPPEval}, \tilde{u}) \mapsto \proofPPEval(t')(
                      \gamma_{\transformErasure}[x \setvar u_{\transformErasure}],
                      \gamma_{\transformPPEval}[x \setvar u_{\transformPPEval}],
                      \tilde{\gamma}[x \setvar \tilde{u}])(s_1)\right)\Big).
                  \end{aligned}
              \end{align}
              Let
              $(x_{\transformErasure}, x_{\transformPPEval}, \tilde{x}) \defas
                  \relationPPEvalProof(x \codePPGets t;~t')(\gamma_{\transformErasure}, \gamma_{\transformPPEval}, \tilde{\gamma})$.
              From the induction hypothesis on $\Gamma \vdash t : \typePProb~\tau$, we have
              \begin{align}
                  P(\semProj_1)(\semSeed \semPBind{} \proofPPEval(t)(\gamma_{\transformErasure}, \gamma_{\transformPPEval}, \tilde{\gamma})) & =
                  \semP{\transformErasure{t}}(\gamma_{\transformErasure}) \quad
                  \label{eq:proof-langpp-realization-term-t-first-proj}
                  \\
                  P(\semProj_2)(\semSeed \semPBind{} \proofPPEval(t)(\gamma_{\transformErasure}, \gamma_{\transformPPEval}, \tilde{\gamma})) & =
                  \semSeed \semPBind{} \semP{\transformPPEval{t}}(\gamma_{\transformPPEval}).
                  \label{eq:proof-langpp-realization-term-t-second-proj}
              \end{align}
              From the induction hypothesis on $\Gamma, x : \tau \vdash t' : \typePProb~\tau'$,
              for all $(u_{\transformErasure}, u_{\transformPPEval}, \tilde{u}) \in \relationPPEvalProof(\tau)$
              we have
              \begin{align}
                   & P(\semProj_1)\left(\semSeed \semPBind{} \proofPPEval(t')(
                      \gamma_{\transformErasure}[x \setvar u_{\transformErasure}],
                      \gamma_{\transformPPEval}[x \setvar u_{\transformPPEval}],
                      \tilde{\gamma}[x \setvar \tilde{u}])\right)
                  \label{eq:proof-langpp-realization-term-tprime-first-proj}
                  \\
                   & \quad =
                  \semP{\transformErasure{t'}}(\gamma_{\transformErasure}[x \setvar u_{\transformErasure}])
                  \nonumber
                  \\
                   & P(\semProj_2)\left(\semSeed \semPBind{} \proofPPEval(t')(
                      \gamma_{\transformErasure}[x \setvar u_{\transformErasure}],
                      \gamma_{\transformPPEval}[x \setvar u_{\transformPPEval}],
                      \tilde{\gamma}[x \setvar \tilde{u}])\right)
                  \label{eq:proof-langpp-realization-term-tprime-second-proj}
                  \\
                   & \quad =
                  \semSeed \semPBind{} \semP{\transformPPEval{t'}}(\gamma_{\transformPPEval}[x \setvar u_{\transformPPEval}]).
                  \nonumber
              \end{align}
              Observe that
              \begin{align}
                   & P(\semProj_1)\left(\semSeed \semPBind{} \tilde{x}\right)
                  \\
                   & \quad =
                  \begin{aligned}[t]
                       & s_0 \semPGets \semSeed                                                                                                                                              \\
                       & s_1 \semPGets \semSeed                                                                                                                                              \\
                       & (u_{\transformErasure}, u_{\transformPPEval}, \tilde{u}) \semPGets \proofPPEval(t)(\gamma_{\transformErasure}, \gamma_{\transformPPEval}, \tilde{\gamma})(s_0) \\
                       & P(\semProj_1)\left(\proofPPEval(t')(
                          \gamma_{\transformErasure}[x \setvar u_{\transformErasure}],
                          \gamma_{\transformPPEval}[x \setvar u_{\transformPPEval}],
                          \tilde{\gamma}[x \setvar \tilde{u}])(s_1)\right)
                  \end{aligned}
                  \label{eq:proof-langpp-realization-first-proj-initial}
                  \\
                   & \quad =
                  \begin{aligned}[t]
                       & s_0 \semPGets \semSeed                                                                                                                                              \\
                       & (u_{\transformErasure}, u_{\transformPPEval}, \tilde{u}) \semPGets \proofPPEval(t)(\gamma_{\transformErasure}, \gamma_{\transformPPEval}, \tilde{\gamma})(s_0) \\
                       & \semP{\transformErasure{t'}}(\gamma_{\transformErasure}[x \setvar u_{\transformErasure}])
                  \end{aligned}
                  \label{eq:proof-langpp-realization-first-proj-handle-s1}
                  \\
                   & \quad =
                  \begin{aligned}[t]
                       & u_{\transformErasure} \semPGets \semP{\transformErasure{t}}(\gamma_{\transformErasure})   \\
                       & \semP{\transformErasure{t'}}(\gamma_{\transformErasure}[x \setvar u_{\transformErasure}])
                  \end{aligned}
                  \label{eq:proof-langpp-realization-first-proj-handle-s0}
                  \\
                   & \quad =
                  x_{\transformErasure},
                  \label{eq:proof-langpp-realization-first-proj-conclude}
              \end{align}
              where \cref{eq:proof-langpp-realization-first-proj-handle-s1} follows from
              \cref{eq:proof-langpp-realization-term-tprime-first-proj}
              and \cref{eq:proof-langpp-realization-first-proj-handle-s0} follows from
              \cref{eq:proof-langpp-realization-term-t-first-proj}.
              Additionally,
              \begin{align}
                   & P(\semProj_2)(\semSeed \semPBind{} \tilde{x})
                  \\
                   & \quad=
                  \begin{aligned}[t]
                       & s_0 \semPGets \semSeed                                                                                                                                              \\
                       & s_1 \semPGets \semSeed                                                                                                                                              \\
                       & (u_{\transformErasure}, u_{\transformPPEval}, \tilde{u}) \semPGets \proofPPEval(t)(\gamma_{\transformErasure}, \gamma_{\transformPPEval}, \tilde{\gamma})(s_0) \\
                       & P(\semProj_2)\left(\proofPPEval(t')(
                          \gamma_{\transformErasure}[x \setvar u_{\transformErasure}],
                          \gamma_{\transformPPEval}[x \setvar u_{\transformPPEval}],
                          \tilde{\gamma}[x \setvar \tilde{u}])(s_1)\right)
                  \end{aligned}
                  \label{eq:proof-langpp-realization-second-proj-initial}
                  \\
                   & \quad= \begin{aligned}[t]
                                 & s_0 \semPGets \semSeed                                                                                                                                              \\
                                 & s_1 \semPGets \semSeed                                                                                                                                              \\
                                 & (u_{\transformErasure}, u_{\transformPPEval}, \tilde{u}) \semPGets \proofPPEval(t)(\gamma_{\transformErasure}, \gamma_{\transformPPEval}, \tilde{\gamma})(s_0) \\
                                 & \semP{\transformPPEval{t'}}(\gamma_{\transformPPEval}[x \setvar u_{\transformPPEval}])(s_1)
                            \end{aligned}
                  \label{eq:proof-langpp-realization-second-proj-handle-s1}
                  \\
                   & \quad= \begin{aligned}[t]
                                 & s_0 \semPGets \semSeed                                                                      \\
                                 & s_1 \semPGets \semSeed                                                                      \\
                                 & u_{\transformPPEval} \semPGets \semP{\transformPPEval{t}}(\gamma_{\transformPPEval})(s_0)          \\
                                 & \semP{\transformPPEval{t'}}(\gamma_{\transformPPEval}[x \setvar u_{\transformPPEval}])(s_1)
                            \end{aligned}
                  \label{eq:proof-langpp-realization-second-proj-handle-s0}
                  \\
                   & \quad= x_{\transformPPEval},
                  \label{eq:proof-langpp-realization-second-proj-conclude}
              \end{align}
              where \cref{eq:proof-langpp-realization-second-proj-handle-s1} follows from
              \cref{eq:proof-langpp-realization-term-tprime-second-proj}
              and \cref{eq:proof-langpp-realization-second-proj-handle-s0} follows from
              \cref{eq:proof-langpp-realization-term-t-second-proj}.
              By \cref{eq:proof-langpp-realization-first-proj-conclude,eq:proof-langpp-realization-second-proj-conclude},
              we conclude that
              $(x_{\transformErasure}, x_{\transformPPEval}, \tilde{x}) \in \relationPPEvalProof(\typePProb~\tau')$.
    \end{itemize}
    To conclude the proof, consider a term $\Gamma \vdash t : \typePProb~\tau$ of the form specified in the theorem statement.
    Suppose $(\gamma_{\transformErasure}, \gamma_{\transformPPEval}) \in \relationPPEval(\Gamma)$.
    Then there exists a $\tilde\gamma$ such that
    $(\gamma_{\transformErasure}, \gamma_{\transformPPEval}, \tilde{\gamma}) \in \relationPPEvalProof(\Gamma)$.
    Hence
    \begin{equation}
        \left(\semP{\transformErasure{t}}(\gamma_{\transformErasure}),
        \semP{\transformPPEval{t}}(\gamma_{\transformPPEval}),
        \proofPPEval(t)(\gamma_{\transformErasure}, \gamma_{\transformPPEval}, \tilde{\gamma})\right)
        =
        \relationPPEvalProof(t)(\gamma_{\transformErasure}, \gamma_{\transformPPEval}, \tilde{\gamma})
        \in \relationPPEvalProof(\tau),
        \label{eq:proof-langpp-realization-pprob-apply-functor}
    \end{equation}
    which implies that $(\semP{\transformErasure{t}}(\gamma_{\transformErasure}),
        \semP{\transformPPEval{t}}(\gamma_{\transformPPEval})) \in \relationPPEval(\typePProb~\tau)$.
    By assumption, $\tau$ is constructed from base types, products, functions,
    and the $\typeResidual~\tau$ constructor, which implies that
    $\relationPPEvalProof(\tau) = \{(x_{\transformErasure}, x_{\transformPPEval}, \semUnit) \mid x_{\transformErasure} = x_{\transformPPEval}\}$,
    and therefore
    \begin{equation}
        \relationPPEval(\typePProb~\tau) = \{(x_{\transformErasure}, x_{\transformPPEval}) \mid x_{\transformErasure} = \semSeed \semPBind{} x_{\transformPPEval}\}.
        \label{eq:proof-langpp-realization-pprob-relation-specialized}
    \end{equation}
    Substituting \cref{eq:proof-langpp-realization-pprob-relation-specialized} into \cref{eq:proof-langpp-realization-pprob-apply-functor}
    recovers the theorem statement.
\end{proof}

\clearpage
\section{Deferred Material for \texorpdfstring{\hyperref[sec:gradient-inference]{\S\ref*{sec:gradient-inference}}}{\S\ref{sec:gradient-inference}}: Proof of Soundness of Gradient Inference}
\label{appendix:deferred-gradient-inference}

In this appendix, we supply the proof of \cref{thm:grad-inference-correctness},
first recalling the main definitions needed.
\begin{definition}[Expectation map]
    We define $\semExpect : P(\setReal^n) \to \setReal^n$ by
    $\semExpect(\mu) \defas \int_{\setReal^n} x \mu(\diff x)$.
\end{definition}
\begin{definition}[Sound inference]
    \label{def:unbiased-inference}
    An inference macro
    $\transformInference{\cdot}$
    acts on a term
    $\Gamma \vdash t : \typeProb~\typeReal$
    of $\langP$
    and produces a term
    $\Gamma \vdash \transformInference{t} : \typeProb~\typeReal$.
    We say that $\transformInference{\cdot}$ is sound if
    for all $\gamma : \semP{\Gamma}$,
    \begin{align}
    \semExpect \left(\semP{\transformInference{t}}(\gamma)\right)
    = \semExpect \left(\semP{t}(\gamma)\right).
    \tag{\ref*{eq:inference-soundness}}
    \end{align}
\end{definition}
\begin{definition}[Sound AD macro]
    \label{def:unbiased-differentiation}
    An AD macro
    $\transformAD{\cdot}$
    acts on a term
    $\Gamma \vdash t : \typeReal^n \typeFunction \typeProb~\typeReal$
    of $\langP$
    and produces a term
    $\Gamma \vdash \transformAD{t} : \typeReal^n \typeFunction \typeProb~\typeReal^n$.
    We say that $\transformAD{\cdot}$ is sound if for all $\gamma : \semP{\Gamma}$ and $\theta : \setReal^n$,
    \begin{align}
        \semExpect\left(\semP{\transformAD{t}}(\gamma)(\theta)\right)
        = \nabla_\theta \left[\semExpect\left(\semP{t}(\gamma)(\theta)\right)\right].
        \tag{\ref*{eq:ad-soundness}}
    \end{align}
\end{definition}
\begin{definition}[Gradient inference macro]
    \label{def:system-transformations}
    Given an inference macro $\transformInference{\cdot}$,
    we define the difference macro
    $\transformGradInfFD{\cdot}$,
    acting on a term $\emptyEnv \vdash t : \typeReal^n \typeFunction \typeProb~\typeReal$
    and producing a term $\emptyEnv \vdash \transformGradInfFD{t} : (\typeReal\typeProduct\typeReal)^n \typeFunction \typeProb~\typeReal$,
    by
    \begin{align}
        \transformGradInfFD{t} \defas
        {\lambda \theta. \set{s \codePGets \codeSeed;~\transformInference{(x_A, x_B) \codePGets \transformPPEval{\transformModCoupling{t}}~\theta~s;~\codePReturn~(x_B - x_A)}}}
        \tag{\ref*{eq:grad-inf-fd}}
    \end{align}
    Given also an AD macro $\transformAD{\cdot}$,
    we define the gradient inference macro
    $\transformGradInfAD{\cdot}$,
    acting on a term $\emptyEnv \vdash t : \typeReal^n \typeFunction \typeProb~\typeReal$
    and producing a term $\emptyEnv \vdash \transformGradInfAD{t} : \typeReal^n \typeFunction \typeProb~\typeReal^n$, by
    \begin{align}
        \transformGradInfAD{t} \defas
        \lambda(\theta_1, \dots, \theta_n).
        \big(\transformAD{\lambda(\varepsilon_1, \dots, \varepsilon_n). \transformGradInfFD{t}~\left((\theta_1, \theta_1+\varepsilon_1), \dots, (\theta_n, \theta_n+\varepsilon_n)\right)}
        \,(0, \dots, 0)\big).
        \tag{\ref*{eq:grad-inf-ad}}
    \end{align}
\end{definition}

\thmGradInferenceCorrectness*
\begin{proof}
    We have
    \begin{align}
        &\semExpect\left(\semP{\transformGradInfFD{t}}(\theta_{\transformCoupling})\right)
        \nonumber
        \\
        &\quad=
        \semExpect\left(
        s \semPGets \semSeed;
        ~\transformInference{(x_A, x_B) \semPGets \semP{\transformPPEval{\transformModCoupling{t}}}(\theta)(s);
        ~\semPReturn~(x_B - x_A)}\right)
        \label{eq:proof-grad-inference-correctness-macro-def}
        \\
        &\quad=
        \semExpect\left(
        s \semPGets \semSeed;
        (x_A, x_B) \semPGets \semP{\transformPPEval{\transformModCoupling{t}}}(\theta)(s);
        ~\semPReturn~(x_B - x_A)\right)
        \label{eq:proof-grad-inference-correctness-sound-inference}
        \\
        &\quad=
        \semExpect\left((x_A, x_B) \semPGets
        \semSeed \semPBind{} \semP{\transformPPEval{\transformModCoupling{t}}}(\theta);
        ~\semPReturn~(x_B-x_A)\right)
        \label{eq:proof-grad-inference-correctness-move-randomize}
         \\
        &\quad=
        \semExpect\left((x_A, x_B) \semPGets
        \semP{\transformErasure{\transformModCoupling{t}}}(\theta);
        ~\semPReturn~(x_B-x_A)\right)
        \label{eq:proof-grad-inference-correctness-apply-erasure}
        \\
        &\quad=
        \semExpect\left(\semP{t}(\theta_B)\right)
        -
        \semExpect\left(\semP{t}(\theta_A)\right),
        \label{eq:proof-grad-inference-correctness-apply-coupling}
    \end{align}
    where \cref{eq:proof-grad-inference-correctness-macro-def} follows from the definition of $\transformGradInfFD{\cdot}$,
    \cref{eq:proof-grad-inference-correctness-sound-inference} follows from the soundness of the inference macro,
    \cref{eq:proof-grad-inference-correctness-move-randomize} follows from associativity
    of the monadic bind,
    \cref{eq:proof-grad-inference-correctness-apply-erasure} follows from \cref{thm:langPP-realization},
    and \cref{eq:proof-grad-inference-correctness-apply-coupling} follows from \cref{thm:factorized-coupling-correctness}.
    Now, for $\mathbf{0}$ the zero vector note that
    \begin{align}
         & \semExpect(\semP{\transformGradInfAD{t}}(\theta)) \nonumber
        \\
         & \quad =
        \semExpect\left(\semP{\transformAD{\lambda(\varepsilon_1, \dots, \varepsilon_n).
                    \transformGradInfFD{t}~\left((\theta_1, \theta_1+\varepsilon_1), \dots, (\theta_n, \theta_n+\varepsilon_n)\right)}}(\theta)(\mathbf{0})\right)
        \\
         & \quad = \nabla_\varepsilon \left[\semExpect\left(\semP{\lambda(\varepsilon_1, \dots, \varepsilon_n).
            \transformGradInfFD{t}~\left((\theta_1, \theta_1+\varepsilon_1), \dots, (\theta_n, \theta_n+\varepsilon_n)\right)}(\theta)(\varepsilon)\right)\right]_{\varepsilon = \mathbf{0}}
        \label{eq:grad-inference-correctness-sound-AD}
        \\
         & \quad = \nabla_\varepsilon \left[\semExpect(\semP{t}(\theta + \varepsilon)) - \semExpect(\semP{t}(\theta))\right]_{\varepsilon = \mathbf{0}}
        \\
         & \quad = \nabla_\theta \left[\semExpect(\semP{t}(\theta))\right],
    \end{align}
    where \cref{eq:grad-inference-correctness-sound-AD} follows from the soundness of AD (\cref{def:unbiased-differentiation}).
\end{proof}

\clearpage
\section{Deferred Material for \texorpdfstring{\hyperref[sec:applications]{\S\ref*{sec:applications}}}{\S\ref{sec:applications}}: Complete List of Numerical Parameters}
\label{appendix:deferred-applications}

In this appendix, we provide a complete list of the numerical parameters used for the case studies
in \cref{sec:applications}.

\noindentparagraph{\normalfont\bfseries M/M/c Queuing Model (\cref{sec:overview,sec:applications-queuing}).}

The generative model is given in \cref{fig:overview-mathematical-model}.
The experiment parameters are as follows:
\begin{itemize}
    \item Number of queuing events $n \in \set{25, 50, 75, 100, 125, 150, 175, 200}$
    \item Packet arrival rate $\theta = 15$
    \item Number of Monte Carlo samples per estimator per data point $s = 1000$
    \item Control variate for score method: $s$-sample empirical mean of $X_n$
\end{itemize}

\noindentparagraph{\normalfont\bfseries Option Pricing Model (\cref{sec:applications-option-pricing}).}
The generative model is given in \cref{fig:option-pricing-tree-model}.
The experiment parameters are as follows:
\begin{itemize}
    \item Expiration time of option
    $T \in \set{0.05,\allowbreak 0.10,\allowbreak 0.15,\allowbreak 0.20,\allowbreak
    0.25,\allowbreak 0.30,\allowbreak 0.35,\allowbreak 0.40,\allowbreak 0.45,\allowbreak
    0.50,\allowbreak 0.55,\allowbreak 0.60,\allowbreak 0.65,\allowbreak 0.70,\allowbreak
    0.75,\allowbreak 0.80,\allowbreak 0.85,\allowbreak 0.90,\allowbreak 0.95,\allowbreak 1.00}$ years
    \item Simulation time step $\Delta T = 0.01$ years
    \item Interest rate $r = 15\% / \mathrm{year}$
    \item Price volatility $\sigma = 5\% / \mathrm{year}$
    \item Initial asset price $S_0 = 40$
    \item Strike price $K = 41$
    \item Number of Monte Carlo samples per estimator per data point $s = 1000$
    \item Control variate for score method: $s$-sample empirical mean of $\ell$
    \item Number of particles for twisted SMC (for \GradInfB{}): 3
\end{itemize}

\noindentparagraph{\normalfont\bfseries Gene Transcription Model (\cref{sec:applications-gene-transcription}).}
The generative model is given in \cref{fig:gene-transcription-gillespie-algorithm}.
The experiment parameters are as follows:
\begin{itemize}
    \item Initial state: $M_0 = 5$, $P_0 = 40$
    \item Ground truth (hidden) parameters: $\alpha = 18$, $\beta = 8$, $\gamma = 1.5$, $\delta = 4$
    \item Initial optimization parameters: $\alpha = 1$, $\beta = 1$, $\gamma = 1$, $\delta = 1$
    \item Simulation time period $T = 2.5$
    \item Number of simulation time steps in generative model: $N = 1500$
    \item Number of sample traces to calculate relative MSE: $k = 100$
    \item Number of Monte Carlo samples per estimator per data point (for table) $s = 100$
    \item Control variate for score method: $s$-sample empirical mean of $\ell$
    \item Stochastic gradient descent learning rate $\eta$ tuned separately for each method
    in multiplicative steps of $\sqrt{10}$
    \item Number of particles for twisted SMC (for \GradInfD{}): 2
\end{itemize}

\clearpage

\makeatletter
\setcounter{NAT@ctr}{0}%
\begingroup
\renewcommand{\NAT@wrout}[5]{%
  \if@filesw
      {\let\protect\noexpand\let~\relax
       \immediate
       \write\@auxout{\string\bibcite{#5}{{S\the\c@NAT@ctr}{#2}{{#3}}{{#4}}}}}\fi
  \ignorespaces}%
\renewcommand{\@biblabel}[1]{[#1]}%
\bibliographystyleapp{ACM-Reference-Format}%
\bibliographyapp{paper}%
\endgroup
\makeatother
}

\end{document}